\documentclass[12pt]{article}
\usepackage{amsmath}
\usepackage{amsfonts}

\usepackage{graphicx}
\usepackage{color}

\def\epsfig#1{}

\newtheorem{theorem}{Theorem}[section]
\newtheorem{corollary}[theorem]{Corollary}
\newtheorem{definition}[theorem]{Definition}

\newtheorem{lemma}[theorem]{Lemma}
\newtheorem{proposition}[theorem]{Proposition}
\newtheorem{remark}[theorem]{Remark}
\newenvironment{proof}[1][Proof]{\textbf{#1.} }{\ \rule{0.5em}{0.5em}}

\newcommand{\dae}{\mbox{\rm -a.e.}}
\newcommand{\dom}{\mathop{\rm Dom}}
\newcommand{\ds}{\displaystyle}

\newcommand{\Ei}{E_0}
\newcommand{\ei}{e}
\newcommand{\EV}{E_V}
\newcommand{\gab}{\gamma_{(a,b)}}
\newcommand{\gst}{\gamma_{(s,t)}}

\newcommand{\LL}{{\ell}} 

\newcommand{\N}{\mathbf{N}}
\newcommand{\0}{\mathbf{0}}

\newcommand{\p}{\partial}
\newcommand{\pl}{\{\ell > 0\}}
\newcommand{\Pac}{\PP_c^{ac}}
\newcommand{\PPac}{{\PP^{ac}}}
\newcommand{\Pic}{\Pi_\preceq}
\newcommand{\po}{{\bar p}}
\newcommand{\proj}{\mathop{\rm proj}}
\newcommand{\PP}{{\mathcal P}}

\newcommand{\R}{\mathbf{R}}

\newcommand{\Rc}{{\rm Ric}}
\newcommand{\NRc}{{\Rc}^{(N,V)}}
\newcommand{\npl}{\{\ell \le 0\}}

\newcommand{\singl}{\mathop{\rm sing}(\ell)}

\newcommand{\spt}{\mathop{\rm spt}}

\newcommand{\td}{\tilde}

\newcommand{\trace}{\mathop{\rm Tr}}
\newcommand{\ts}{\textstyle}

\newcommand{\vo}{{\bar v}}
\newcommand{\vol}{{\mathop{\rm vol}}}
\newcommand{\vg}{{\vol_{g}}}
\newcommand{\wdot}{{\,\cdot\,}}

\newcommand{\xo}{{\bar x}}

\newcommand{\yo}{{\bar y}}

\begin{document}

\author{
Robert J. McCann\thanks{
Department of Mathematics,
University of Toronto, Toronto Ontario M5S 2E4 Canada,
{\tt mccann@math.toronto.edu}}}


\title{
Displacement convexity of Boltzmann's entropy
characterizes the strong energy condition from 
general relativity%
\thanks{%
{\em Dedicated to my muse, Carolyn.} MSC Primary 53C50, Secondary 49J52 58Z05 83C99 82C35.
The author is pleased to thank Spyros Alexakis, Stefanos Aretakis,
Eric Carlen, Christian Ketterer and Eric Woolgar for stimulating conversations
and relevant references, and to Stefan Suhr, Martin Kell and Andrea Mondino for communicating 
their works to him in preprint form.  
He thanks Robert Wald for early encouragement,  and Elliott Lieb,
who drew Boltzmann's entropy to his attention upon hearing about displacement convexity. 
He is grateful for the hospitality of the
University of Chicago, the Mathematical Sciences Research Institute (MSRI) at Berkeley CA, 
Fields Insitute for the Mathematical Sciences, and Mittag-Leffler Institute
during various stages of this work. He acknowledges partial support of his research by
Natural Sciences and Engineering Research Council of Canada Grants 217006-08, -15 and -20,
by a Simons Foundation Fellowship,  and by US National Science Foundation Grant No. DMS-144041140
while in residence at MSRI during thematic programs in 2013 and 2016.
\copyright 2020 by the author.}}
\date{\today}

\maketitle

\begin{abstract}
On a Riemannian manifold, 
lower Ricci curvature bounds are known to be characterized
by geodesic convexity properties of various entropies
with respect to the Kantorovich-Rubinstein-Wasserstein
square distance from optimal transportation. These notions also make sense in a (nonsmooth) metric measure setting,
where they have found powerful applications.
This article initiates the development of an analogous theory for lower Ricci curvature bounds
in timelike directions on a (globally hyperbolic) Lorentzian manifold. 
In particular,  we lift fractional powers of the Lorentz distance (a.k.a.
time separation function) to probability measures on spacetime,  and show the strong energy condition
of Hawking and Penrose is equivalent to geodesic convexity of the Boltzmann-Shannon entropy there.
This represents a significant first step towards a formulation of the 
strong energy condition and exploration of its consequences
in nonsmooth spacetimes, and 
hints at new connections linking the theory of gravity to the second law of thermodynamics.
\end{abstract}

\tableofcontents

\section{Introduction}


The second law of thermodynamics is amongst the most remarkable and universal laws
in all of physics. It asserts that for an isolated (or adiabatic) system,
only certain physical processes are permitted. Moreover the distinction between
permitted and forbidden processes is determined by the principal that the associated
entropy be non-decreasing in time.  In other words,
these dynamical processes exhibit a preferred direction of time.  There is an analogous
law governing the dynamics of black holes in {\em general relativity}, Einstein's theory
of gravity,  which states that the area of the event horizon of a black hole is proportional
to its entropy, hence 
can only increase \cite{Bekenstein73} \cite{BardeenCarterHawking73}. 
In fact, this idea can also be turned on its head,  with the postulated proportionality used to derive 
 general relativity \cite{Jacobson95} and other forms of gravity \cite{Verlinde11} \cite{Verlinde17} 
as emergent, entropic (i.e.~statistical) forces.

In this paper we investigate
another, quite different law concerning the information-theoretic entropy of
probability measures on spacetime. 
Using the $q$-Lorentz-Wasserstein distance from optimal transportation \cite{EcksteinMiller17}
in place of a metric on such measures,  
we claim convexity of this entropy along 
the geodesics of probability measures which result is equivalent to the strong energy condition of Hawking and Penrose 
\cite{Penrose65a} \cite{Hawking66a} \cite{HawkingPenrose70},
which --- despite its more limited range of validity %
than the dominant energy condition, e.g.~\cite{Carroll04}
--- plays an important role in gravitational theory.

The strong energy condition is a positive-definiteness condition on the stress-energy tensor $T_{ab}$,
which encodes the energy and momenta densities and fluxes at each point in spacetime.  
It asserts that in each normalized timelike direction $v^a$,  this tensor dominates half its trace: 
$T_{ab} v^a v^b \ge \frac12 T$.  Equivalently, since the Einstein equation postulates proportionality 
of $T_{ab} - \frac12 Tg_{ab}$ to the Ricci tensor $\Rc_{ab}$, in the absence of cosmological constant
the strong energy condition boils down to Ricci non-negativity in timelike directions: $\Rc_{ab} v^a v^b \ge 0$.
In the presence of trapped or other distinguished surfaces,
it implies a spacetime must have singularities, e.g. \cite{HawkingEllis73} \cite{Wald84} \cite{TreudeGrant13}.
It is also understood to be responsible for the empirical fact that gravity is purely
attractive,  and never repulsive (at least, in the absence of rotation; c.f. \cite{Carroll04} and 
Raychaudhuri's equation).


That the Hawking and Penrose condition has anything to do with entropy or the second law of thermodynamics
may seem surprising.  However, this relationship is foreshadowed
by recent developments in Riemannian geometry,
the mathematical progenitor of general relativity.  There a line of research
due to the author \cite{McCann97} and his collaborators \cite{CorderoMcCannSchmuckenschlager01}
\cite{CorderoMcCannSchmuckenschlager06}, Otto and Villani \cite{OttoVillani00}, and
von Renesse and Sturm, 
has culminated in a characterization of Ricci-curvature
lower bounds involving only the displacement convexity of certain
information-theoretic entropies \cite{SturmvonRenesse05}.
This in turn led Sturm \cite{Sturm06ab} and independently 
Lott and Villani \cite{LottVillani09} to the development of a robust theory for lower Ricci
curvature bounds in a (non-smooth) metric-measure space setting.  
A vibrant theory of such spaces has emerged rapidly since that time,
which would take us too far afield to survey; see 
e.g. \cite{AmbrosioGigli13}
\cite{AmbrosioGigliMondinoRajala15} 
\cite{AmbrosioGigliSavare14i} 
\cite{AmbrosioGigliSavare14d} 
\cite{AmbrosioGigliSavare15}
\cite{CavallettiMilman16p}
\cite{CavallettiMondino17i} \cite{CavallettiMondino17g} \cite{ErbarKuwadaSturm15}
\cite{Gigli14} \cite{GigliRajalaSturm16} \cite{Ketterer15} \cite{MondinoNaber14p} \cite{Ohta15};
competing approaches to the complementary upper bounds may be found in 
\cite{Naber13p} \cite{HaslhoferNaber16p} and \cite{Sturm16p} \cite{Sturm17p}. 

Our purpose is to initiate the development of an analogous theory
in the equally tantalizing and more physically relevant setting of gravitation.
A particular consequence of our theory is that it becomes possible to define what
it means for a volume measure on a (nonsmooth) Lorentzian geodesic space \cite{KunzingerSamann17p} to satisfy 
the strong energy condition, and to show that many familiar implications of this 
condition in the smooth setting extend to the nonsmooth geometries
representing super-solutions of the vaccuum Einstein equations. 
This is particularly relevant to gravitational theory since a wide variety of smooth spacetimes contain
timelike geodesics which are neither extendible nor complete~\cite{HawkingEllis73}.
However,  the present manuscript is devoted to showing consistency of the proposed definition with the classical one
in the smooth setting, and defers the development of a theory of non-smooth spacetimes to a forthcoming work.

From the technical point of view,  our major innovations include the introduction of the 
strictly convex Lagrangian \eqref{q-Lagrangian} inducing the fractional time-separation function
$\ell(x,y)^q$ for $0<q<1$, 
and the development of techniques for resolving the resulting optimal transportation problem on 
spacetime posed by Eckstein and Miller \cite{EcksteinMiller17}
and its dual, despite the singularities of these objective functionals.  Even for the semi-relativistic
Lagrangian of Brenier \cite{Brenier03} (corresponding to $q=1$ on Minkowski hyperplanes),  
the challenges such singularities present have plagued researchers 
for more than a decade, 
and stimulated a series of works 
\cite{McCannPuel09} \cite{BertrandPuel13} \cite{BertrandPratelliPuel18} \cite{Suhr16p}
leading up to a preprint by Kell and Suhr \cite{KellSuhr18p}
which, for $q=1$, resolves certain analogous issues described below 
simultaneously and independently of the present manuscript.

\subsection{Optimal transport with Lorentz distances}

Let $(M^{n},g)$ be a smooth, connected, Hausdorff, time-oriented Lorentzian manifold,  
with a signature $ (+, -, \ldots, -)$ metric.  It follows from results of
Nomizu, Ozeki \cite{NomizuOzeki61} and Geroch \cite{Geroch68}
that 
$M$ is second countable and that its topology also
arises from a complete Riemannian metric $\tilde g$.
Hereafter these hypotheses may be abbreviated by saying $M$ is a
{\em spacetime}.
A tangent vector $v \in T_x M$ is said to be {\em timelike} $(>0)$, {\em spacelike} $(<0)$ or {\em null} $(=0)$,  
according to the sign of $v^a g_{ab} v^b$.  The time-orientation of the manifold allows {\em causal}
(i.e. non-spacelike) vectors 
to be classified 
continuously 
as either {\em future-} or {\em past-directed}, with $v \ne 0$ being future-directed
if and only if $-v$ is past-directed. 

For $q \in (0,1]$, define the convex Lagrangian 
$L(v,x;q) := -(g_{ab}(x) v^a v^b)^{q/2}/q$ on the tangent bundle of $M$, with the convention
that $L(v,x;q) = +\infty$ unless $v$ is future-directed. 
For curves $\sigma \in C^{0,1}([0,1];M)$,  the associated action is
\begin{equation} \label{q-action}
A[\sigma;q] := \int_0^1 L(\sigma'(s),\sigma(s);q) ds.
\end{equation}
We define the ($q$-dependent) {\em Lorentz distance} between any two points as
the infimum
\begin{equation} \label{q-lapse}
\ell(x,y;q) =  - \inf   \{  A[\sigma;q] \mid \sigma \in C^{0,1}([0,1];M), \sigma(0) = x, \sigma(1)=y\}
\end{equation}
over Lipschitz curves with fixed endpoints. 
For notational simplicity, we adopt
the convention 
\begin{equation}\label{q-convention}
(-\infty)^{1/q}:= -\infty =: (-\infty)^q
\end{equation}
 throughout.
With this convention,
$(q\ell(x,y,q))^{1/q} =: \ell(x,y)$
is independent of $q\in(0,1]$ and satisfies the reverse triangle inequality
\begin{equation}\label{reverse triangle inequality}
\ell(x,y) \ge \ell(x,z) + \ell (z,y);
\end{equation}
it represents the maximum amount that a physical particle can age while travelling from $x$ to $y$,
and is therefore also known as the {\em time-separation} function.
It differs from the textbook definitions \cite{HawkingEllis73} \cite{O'Neill83}
\cite{BeemEhrlichEasley96} of the Lorentz distance, which is non-negative,
only in that $\ell(x,y) = -\infty$ if there is no future-directed curve from $x$ to $y$.
This has the convenient consequence that 
$y$ lies in the {\em causal future} of $x$ if and only if $\ell(x,y)\ge 0$,
and in the {\em chronological future} of $x$ if and only if $\ell(x,y)>0$,
(which one may also take as definitions of the italicized terms).
It is also important for ensuring that whenever possible,  the solutions to the 
transportation problem introduced in the next paragraphs
couple only causally related events.

We henceforth assume our spacetimes $(M,g)$ to be
{\em globally hyperbolic}, meaning $M$ has no closed causal curves,  and
for each $x,y \in M$ the intersection
\begin{equation}\label{global hyperbolicity}
J^+(x) \cap J^-(y) := \{z \in M \mid \ell(x,z) \ge 0\} \cap \{z \in M \mid \ell(z,y) \ge 0\}
\end{equation}
of the causal future of $x$ with the causal past of $y$ is compact \cite{BernalSanchez07}.
These assumptions
ensure that the infimum \eqref{q-lapse} is actually attained \cite{Avez63} \cite{Seifert67};  
Jensen's inequality ensures
the curve that attains it is independent of $q$.
When $\ell(x,y) > 0$ any such curve is a geodesic; moreover
this geodesic is affinely parameterized if $q<1$,  in which case we call it an action
minimizing (or proper-time maximizing) 
segment. 
Each point $z$ on this segment saturates
the triangle inequality \eqref{reverse triangle inequality};  conversely, when $\ell(x,y)>0$, inequality 
\eqref{reverse triangle inequality} holds strictly unless $z$ lies on an action minimizing segment joining 
$x$ to $y$.

Let $\PP(M)$ be the set of Borel probability measures on $M$, and 
$\PP_c(M) := \{ \mu \in \PP(M) \mid \spt \mu \mbox{\rm\ is compact}\}$,
where $\spt \mu$ denotes the smallest closed subset of $M$ carrying the full mass of $\mu \ge 0$. 
We lift the Lorentz distance $\ell$ from $M$ to $\PP(M)$ as follows:
given $\mu$ and $\nu$ on $M$, we recall the {\em $q$-Lorentz-Wasserstein distance}
introduced independently from the present manuscript by Eckstein and Miller \cite{EcksteinMiller17},
\begin{equation}\label{MK}
\LL_q(\mu,\nu):= \sup_{\pi \in \Pic(\mu,\nu)} \left( \int_{M \times M} \ell(x,y)^q d\pi(x,y)\right)^{1/q}
\end{equation}
where $\Pic(\mu,\nu)$ denotes the set of joint measures $\pi\ge 0$ on $M^2$ with 
$\spt \pi \subset \ell^{-1}([0,\infty])$ 
and having $\mu$ and $\nu$ for marginals.
When $q=1$ it represents the maximum expected proper-time which can elapse between
the distributions of events represented by $\mu$ and those represented by $\nu$.
If $\ell$ is dominated by a lower semicontinuous function in $L^1(d\mu) \oplus L^1(d\nu)$ --- 
as when $\mu$ and $\nu$ are compactly supported --- 
 then \eqref{MK} is  
attained, e.g.~\cite{Suhr16p} \cite{Villani09}, since
$\ell$ is upper semi-continuous, $d_{\tilde g}$ makes $M$ into a complete separable metric space, and
our convention \eqref{q-convention} ensures
the value of the supremum \eqref{MK} is unchanged
if we replace $\Pic(\mu,\nu)$ by the set $\Pi(\mu,\nu)$ of all measures on $M^2$ having $\mu$ and $\nu$ 
for marginals.
In addition, if $\Pic(\mu,\nu)$ is non-empty 
then \eqref{MK} is finite; c.f.~\cite{EcksteinMiller17}.
The maximizing $\pi$ will be called {\em $\ell^q$-optimal}, since $\ell(x,y;q) = \frac1q\ell(x,y)^q$.
Our convention 
that
$\LL_q(\mu,\nu)=-\infty$ unless there exists a pairing $\pi \in \Pi(\mu,\nu)$ with $y$ lying in the causal
future of $x$ for $\pi$-a.e.\ $(x,y)$ is at variance with the convention of Eckstein and Miller
(who instead define the $q$-Lorentz Wasserstein distance to be zero in this case).
Nevertheless,  the reverse triangle inequality \eqref{reverse triangle inequality}
implies \cite{EcksteinMiller17}
\begin{equation}\label{RTIM}
\LL_q(\mu,\nu) \ge \LL_q(\mu,\lambda) + \LL_q(\lambda,\nu)
\end{equation}
for all $\lambda,\mu,\nu \in \PP(M)$, where we use the convention $\infty-\infty=-\infty$ to interpret the right-hand side 
of \eqref{RTIM}.
This precisely parallels the fact that the Kantorovich-Rubinstein-Wasserstein distances 
$W_p$ in the metric space setting described e.g.\ in \cite[Definition 6.1]{Villani09}
satisfy the usual triangle inequality. 
Continuing this analogy allows us to lift the notion of maximizing geodesic segment
from points to (chronologically sequenced) probability measures:

\begin{definition}[Geodesics of probability measures on spacetime]
We say $s \in [0,1] \mapsto \mu_s \in \PP(M)$ is a 
$q$-geodesic if and only if 
\begin{equation} \label{q-geodesic}
\LL_q(\mu_s,\mu_t) = (t-s) \LL_q(\mu_0,\mu_1) \in (0,\infty)
\end{equation}
for each $0 \le s<t \le 1$.
\end{definition}

With this terminology, $q$-geodesics are  implicitly future-directed and timelike.
Subsequent sections establish the existence of $q$-geodesics connecting
fairly arbitrary endpoints $\mu_0$ and $\mu_1$.  The interpolating measures $\mu_s$
turn out to inherit
compact support from the endpoints.
Let $\PPac(M) \subset \PP(M)$ denote the set of measures $\mu$ which are absolutely continuous
with respect to the Lorentzian volume $\vg$, and $\Pac(M) := \PPac(M) \cap \PP_c(M)$.
If one of the compactly supported measures $\mu_0$ or $\mu_1$ is absolutely continuous and $q<1$,
we show the $q$-geodesic joining them to be unique
under the technical restriction  
of {\em $q$-separation} proposed in Definition~\ref{D:q-separated},
which amounts essentially to 
the 
positivity of
$\ell$ throughout the supports of all $\ell^q$-optimal 
$\pi \in \Pi(\mu_0,\mu_1)$. 
This implies in particular that all $\ell^q$-optimal $\pi$ correlate the endpoint distributions 
of events chronologically,  so that $x$ lies in the chronological past 
of $y$ whenever $(x,y) \in \spt \pi$.
Apart from the second endpoint, the whole $q$-geodesic lies in $\Pac(M)$ in this case.
We define the {\em relative entropy} $\EV$ and
{\em Boltzmann-Shannon entropy} $\Ei$
on $\PPac(M)$:


\begin{definition}[Entropy]\label{D:V-tropy}
Given $V \in C^2(M)$ 
and $\mu \in \PPac(M)$ having density  $\rho := d\mu / dm$ with respect to the weighted 
Lorentzian volume $dm = e^{-V} d\vg$, we define
\begin{equation}\label{V-tropy}
\EV(\mu) := \int_{M} \rho(x) \log \rho(x) e^{-V(x)}  dvol_g(x) 
\end{equation}
if the integral has a well-defined value in $[-\infty,\infty]$,  and set $\EV(\mu) := -\infty$ otherwise.
When $\mu$ vanishes outside a set $U \subset M$ of finite volume $\vg[U]<\infty$, as when $\mu$ is compactly supported, Jensen's inequality shows 
\begin{equation}
\EV(\mu) \ge - \log \int_{U} e^{-V(x)} d\vg(x) > -\infty.
\label{U Jensen}
\end{equation}
One can define $\EV(\mu) =+\infty$ if $\mu \in \PP(M) \setminus \PPac(M)$.
When $V:=0$ the Boltzmann-Shannon entropy $\Ei$  results (but with sign 
differing from the usual convention of the physics literature).
\end{definition}

Our central results (Theorems \ref{T:Boltzmann Hessian} 
and \ref{T:necessity}) are 
foreshadowed by the following corollary,
which characterizes the strong energy condition 
of Hawking and Penrose \cite{HawkingPenrose70} via the convexity of 
Boltzmann-Shannon entropy along $q$-geodesics in $\Pac(M)$.
It incorporates the possibility of non-vanishing cosmological constant $K \ge 0$.
To avoid technical complications
associated with the lack of smoothness of the Lorentz distance $\ell(x,y)$ at points where it vanishes,  
we prefer
to focus our attention on geodesics whose endpoints $\mu_0$ and $\mu_1$ are totally chronologically
related in the sense that $\spt [\mu_0 \times \mu_1] \subset \pl$; 
i.e. each point $y \in \spt \mu_1$,
 lies in the timelike future of each point 
in $\spt \mu_0 \subset M$, as in the following corollary.  
Unfortunately,  this cannot remain true for $\spt [\mu_s \times \mu_t]$ when $t-s>0$
is small, which is more delicate yet apparently unavoidable.  We resolve this difficulty by showing
the aforementioned $q$-separation propagates from the endpoints to the interior of a $q$-geodesic.

\begin{corollary}[Positive energy = entropic displacement concavity]
Let $(M^{n},g)$ be a globally hyperbolic spacetime.
Fix  $0<q < 1$.  (i) If the Lorentzian metric satisfies 
$\Rc_{ab} v^a v^b <K \in \R$ in some timelike direction $(v,x) \in TM$ normalized so that 
$g_{ab} v^a v^b=1$,  then a $q$-geodesic $s \in[0,1] \mapsto \mu_s \in \Pac(M)$ exists
along which $\ei(s) := \Ei(\mu_s)$ is $C^2$-smooth and satisfies $\ei''(0) < K\LL_q(\mu_0,\mu_1)^2$;
moreover, $\spt [\mu_0 \times \mu_1]$ is disjoint from $\npl$
and can be chosen to be contained in any specified neighbourhood of $(x,x)$.
(ii)  Conversely, if the metric tensor satisfies $\Rc_{ab} v^a v^b \ge K v^av^b g_{ab} \ge 0$ in all 
timelike directions $(v,x) \in TM$ of the tangent bundle,  
then  $\ei''(s) \ge \frac1{n} \ei'(s)^2 +  K\LL_q(\mu_0,\mu_1)^2$
holds 
along all $q$-geodesics $s \in[0,1] \mapsto \mu_s \in \Pac(M)$
with finite entropy endpoints and $\spt[\mu_0 \times \mu_1]$ disjoint from $\npl$.
Here $e''(s)$ is interpreted distributionally.
\end{corollary}

Apart from possible aesthetic or philosophical considerations,  
the advantage of the reformulation of the strong energy condition provided by this corollary
is that the notions it relies on
--- namely, $q$-geodesics, entropy, and convexity --- require only a time-separation function $\ell(x,y)$ (which determines
the causal structure),  and a reference measure
$m$ (given in this case by $dm = e^{-V} d\vg$). 
As a result,  they 
can be adapted to non-smooth settings, including the
Lorentzian geodesic spaces of Kunzinger and S\"amann~\cite{KunzingerSamann17p},
where they offer a promising approach to the development of a synthetic theory of spaces which enjoy uniform 
lower Ricci curvature bounds in all timelike directions.  We call such spaces $TCD^e_q(K,N)$ spaces, in analogy with the corresponding 
theory of curvature dimension conditions in metric-measure spaces
pioneered by Lott, Villani \cite{LottVillani09} and Sturm \cite{Sturm06ab}. 
Here the superscript $e$ refers to the simpler alternative but equivalent definition of these conditions by 
Erbar, Kuwada and Sturm \cite{ErbarKuwadaSturm15};  the possibility $q \ne 2$ 
was explored in the metric-measure setting by Kell \cite{Kell17}  for $q \ge 1$, and the leading T
is a mnemonic for 
timelike, following the terminology used by Woolgar and Wylie in their work on singularities
and splitting theorems for $N$-Bakry-\'Emery spacetimes \cite{WoolgarWylie16}.
We develop such a theory for nonsmooth spacetimes in a forthcoming work.

\subsection{Main results and discussion}

Our main result is considerably more general than the corollary indicated above.
It concerns lower bounds for the following modified version of the Ricci tensor,
called the $N$-Bakry-\'Emery-Ricci tensor 
 in honor of 
 \cite{BakryEmery83}:

\begin{definition}[$N$-Bakry-\'Emery-Ricci tensor]\label{D:NBER tensor}
Given $n \ne N \in [-\infty,\infty]$ and $V \in C^2(M)$ on a Lorentzian manifold $(M^{n},g)$, we define 
the modified Ricci tensor by
\begin{equation}\label{NBER tensor}
\NRc_{ab} : = \Rc_{ab} + \nabla_a \nabla_b V - \frac{1}{N-n}(\nabla_a V)( \nabla_b V),
\end{equation}
and adopt the conventions $\Rc^{(n,V)}_{ab}:=\Rc_{ab}$ and $V=0$ if $N=n$.
\end{definition}

Explored in the Lorentzian context by Case \cite{Case10},
it was also used by Woolgar and collaborators (see \cite{WoolgarWylie16} and the references there)
to extend Hawking and Penrose type singularity theorems to manifolds-with-density ---
on which the Lorentzian volume $d\vg(x)$ is replaced by $dm(x) := e^{-V(x)} d\vg(x)$.

For fixed $K\ge 0$, $N\ge n$, $V \in C^2(M)$ and $0<q<1$,  the results of
Corollary \ref{C:sufficiency} and Theorem \ref{T:necessity} below show that 
\begin{equation}\label{I:N-Ricci bound}
\NRc_{ab} v^a v^b \ge K g_{ab} v^a v^b 
\end{equation}
holds for each timelike vector $(v,x) \in TM$ tangent to a globally hyperbolic spacetime $(M^n,g)$  
if and only if the distributional second-derivative of the relative entropy $e(s):= \EV(\mu_s)$ satisfies
\begin{equation}\label{I:KN-convex}
e''(s) \ge \frac1N e'(s)^2 + K \LL_q(\mu_0,\mu_1)^2
\end{equation}
on each $q$-geodesic $s \in [0,1] \mapsto \mu_s \in \Pac(M)$ 
with $q$-separated, finite entropy endpoints.  
The requirement that each pair of endpoints be compactly supported and $q$-separated can be relaxed
if we are content to have {\em weak} displacement convexity, meaning the existence of a {\em single} $q$-geodesic joining them which satisfies the required inequalities;
 see Corollary \ref{C:sufficiency2}.
This is the Lorentzian analog of Erbar, Kuwada and Sturm's  reformulation 
$CD^e(K,N)$ \cite{ErbarKuwadaSturm15}
of Sturm's original curvature-dimension condition $CD(K,N)$ \cite{Sturm06ab}  for a metric measure space $(M,d,m)$
 (also formulated independently, for $K/N=0$, by Lott and Villani \cite{LottVillani09}); as long as 
geodesics in $M$ are essentially non-branching the two formulations are shown to be equivalent by combining the
results of Bacher and Sturm \cite{BacherSturm10} 
and Cavalletti and Milman \cite{CavallettiMilman16p} with those of~\cite{ErbarKuwadaSturm15}.

Several points deserve further mention.  First,  the coefficient $\LL_q(\mu_0,\mu_1)^2$ of $K$
is natural,  in the sense that it disappears from \eqref{I:KN-convex} 
if we `arc-length' reparameterize the $q$-geodesic $(\mu_s)_{s \in[0,1]}$ over $ [0,\LL_q(\mu_0,\mu_1)]$
instead of $[0,1]$.  Second,  in contradistinction to theories of Lott-Villani and Sturm,
our theory does not encompass negative lower Ricci curvature bounds: 
although \eqref{I:KN-convex} continues to imply \eqref{I:N-Ricci bound} when $K<0$,  
by way of converse we can only deduce that \eqref{I:N-Ricci bound} implies
\begin{equation}\label{I:K<0}
e''(s) \ge \frac1N e'(s)^2 + K \int_{M \times M} \ell(x,y)^2 d\pi(x,y)
\end{equation}
where $\pi \in \Pi(\mu_0,\mu_1)$ is $\ell^q$-optimal.  To obtain \eqref{I:KN-convex} from this using
Jensen's inequality requires $K \ge 0$, as in Remark \ref{R:Jensen} below.

In the present smooth context,  its equivalence to \eqref{I:N-Ricci bound}
shows independence of \eqref{I:KN-convex} on $q \in (0,1)$.  It is not clear whether this
$q$-independence extends to the non-smooth $TCD^e_q(K,N)$ spaces of our sequel.  In the metric-measure context,
the analogous class of spaces are those satisfying the $CD_q(K,N)$ condition defined using the $q$-Wasserstein
metric by Kell for $q > 1$ \cite{Kell17}.
The subset of $CD_q(K,N)$ that consists of non-branching spaces 
does not dependent on $q$ according to Akdemir, Cavalletti, Colinet, McCann and Santarcangelo \cite{AkdemirCavallettiColinetMcCannSantarcangelo20p}.

Finally,  the astute reader will note we have established convexity,
rather than the monotonicity which would be required of the thermodynamic entropy by the second law.
But our entropy is not the thermodynamic entropy, 
and the extent to which there is a connection, 
if any, between them remains mysterious.








\subsection{Further related works}

The need to extend concepts from 
Lorentzian geometry to non-smooth settings
is discussed, e.g., in \cite{KunzingerSamann17p}.
The idea of approaching this problem through displacement convexity of the entropy on the space
of probability measures is inspired by its success in the Riemannian context following
\cite{LottVillani09} \cite{Sturm06ab}.

Optimal transportation with respect to Lagrangians which are smooth and strictly convex is laid out in 
Villani \cite{Villani09}, following works of Benamou, Brenier \cite{BenamouBrenier00} \cite{Brenier03},
Bernard and Buffoni \cite{BernardBuffoni06} \cite{BernardBuffoni07}.   
As for the square distance \cite{McCann97} \cite{CorderoMcCannSchmuckenschlager01},
these initial investigations established existence, uniqueness and
regularity of optimal maps and interpolants~$\mu_s$. 
Ohta~\cite{Ohta09} \cite{Ohta14h}, Lee \cite{Lee13}, Kell~\cite{Kell17}, and Schachter \cite{Schachter17},
continued this line of research 
by exploring entropic displacement convexity 
and its relation to notions of curvature for Lagrangians in varying degrees of generality,  
always assuming smoothness 
of $L$ except perhaps at the zero vector. 
Relatively little 
attention has been devoted to singular 
Lagrangians, 
apart from the subRiemannian case
\cite{AmbrosioRigot04} \cite{AgrachevLee09} \cite{AgrachevLee14} \cite{FigalliRifford10} \cite{LeeLiZelenko16} \cite{BaloghKristalySipos18}.

The most notable exceptions appear in work of Eckstein and Miller \cite{EcksteinMiller17}, who introduced
the $q$-Lorentz Wasserstein distance \eqref{q-lapse} as a means of exploring causality relations between
spacetime probability measures independently of the present manuscript,
and Suhr \cite{Suhr16p},  who studied the $q=1$ maximization problem \eqref{MK}
along with various generalizations complementary to ours,
and focused especially on measures $\mu$ and $\nu$ which,  instead of being absolutely continuous
with respect to $d\vg$, are supported on spacelike hypersurfaces.  
His manuscript, which we learned of only during the writing of this work, 
provides analogs to several of our results from sections 
\ref{S:spacetime geodesy} 
and~\ref{S:map} in this rather different context.
As antecedents for his study
he cites the cosmic initial velocity reconstruction problem addressed by Frisch et al \cite{FrischMatarreseMohayaeeSobolevskii02}
\cite{BrenierFrischHenonLoeperMatarreseMohayaeeSobolevskii03},  
and the work of Bertrand and Puel 
\cite{BertrandPuel13}
on Brenier's relativistic heat equation~\cite{Brenier03},
which involves the special case of Suhr's problem set on parallel planes in Minkowski space
(and was also explored in \cite{McCannPuel09} \cite{BertrandPratelliPuel18}).
The enhancing effect of {\em Newtonian} self-gravity on the displacement convexity of various entropies
was first discovered by Loeper \cite{Loeper06g}.

After the present results had been announced \cite{McCann18p},  we learned of work of
Kell and Suhr which, particularly for $q=1$, develops a duality theory analogous to that of
\S 4,  under hypotheses which are related to but different from our $q$-separation;
their conditions are phrased in terms of the existence of dynamical transport 
plans (= measures on action minimizing segments) which need not a priori be optimal, but 
whose velocities are locally bounded away from the light cone
\cite{KellSuhr18p}. 
They also
address the absolute continuity of $1$-geodesics,  using an approach different 
from both Corollary~\ref{C:q-geodesic} and Remark~\ref{R:relevance of inverse maps},
and indicate possible extensions to $q<1$.  
We similarly learned of a heuristic argument by Gomes and Seneci \cite{GomesSeneci18p} 
extending displacement convexity
to planning problems from mean-field games which involve rather general smooth convex Hamiltonians and,
strikingly, incorporate local congestion effects.
Some months later,  we learned of Mondino's work with Suhr,
in which they independently showed how to express not only our lower bound but
also the complementary upper Ricci bound in timelike directions,
by testing displacement convexity not along all $q$-geodesics,  but only along those
which are localized and smooth; this complements our work nicely by enabling them to give a
sense to the full Einstein equations with sources \cite{MondinoSuhr18p}
(and to relax our global hyperbolicity assumption).

\subsection{Plan of the paper}

The plan of the paper is the following.  In the next section 
we establish the existence and uniqueness of 
$q$-geodesics connecting chronologically related probability measures on spacetime.  
It is followed by a section which recalls various notions from 
non-smooth analysis, and lays out needed properties of the Lorentz distance $\ell(x,y;q)$
and the $q$-dependent family of Lagrangians and Hamiltonians which define it.
In Section \S\ref{S:dual} we develop a Kantorovich-Koopmans duality theory for the optimal 
transportation problem \eqref{MK}, under the 
aforementioned restriction that the probability
measures $\mu$ and $\nu$ be $q$-separated.  This duality theory allows us to 
develop a Lagrangian calculus for $q$-geodesics in \S \ref{S:map}, based on the 
existence and uniqueness of optimal maps.  The proof that intermediate-time maps have 
Lipschitz inverses is 
relegated to Appendix \ref{S:Monge-Mather proof}, see also
Suhr for $q=1$  \cite{Suhr16p}.
In \S \ref{S:dc} this calculus is employed to compute derivatives of the entropy 
along $q$-geodesics and establish our claim that Ricci non-negativity in timelike
directions implies entropic displacement convexity --- at least along geodesics
with $q$-separated endpoints.
Section \S \ref{S:relaxing separation} shows this $q$-separation restriction can relaxed 
if we are content to conclude weak displacement convexity of the entropy.  In this section we
also extend our results concerning existence and uniqueness of optimal maps to situations
where it is unclear whether strong duality is attained.
The converse implication,  that weak entropic displacement convexity implies timelike
Ricci non-negativity, is established in Section \S \ref{S:RLB from convexity}.

\section{Geodesics of probability measures on spacetime}
\label{S:spacetime geodesy}

The main goal of this section is to derive conditions which guarantee the 
existence and uniqueness of $q$-geodesics in $\PP(M)$.  
 A more thorough characterization
of their properties relies on the development of a strong duality theory,  both of which require additional hypotheses and are deferred to subsequent sections. 
 See also Suhr for the special case $q=1$ \cite{Suhr16p}.

We begin by introducing the singular set $\singl$ of the Lorentz distance, which
consists of the timelike cut locus of $M$ together with all pairs of points not 
in chronological sequence. It is well-known to be closed,
and can also be characterized as the set where $\ell$ fails to be smooth;
see Theorem~\ref{T:lapse smoothness}.


\begin{definition}[Singularities of the Lorentz distance]\label{D:singl}
Let $(M,g)$ be a globally hyperbolic spacetime.
We say $(x,y) \in \singl$ unless $\ell(x,y)>0$ and $x$ and $y$ both lie in the relative interior of some
affinely parameterized proper-time maximizing 
geodesic segment.
\end{definition}

We also make frequent use of the following construction familiar from 
optimal transportation. 


\begin{definition}[Push-forward]\label{D:push-forward}
Given a Borel map $F:M\longrightarrow N$ between two 
metric spaces,
and a Borel measure $\mu \ge 0$ on $M$,  we define the {\em push-forward} $F_\#\mu$ to be 
the Borel measure on $N$ given by $F_\#\mu(V) = \mu(F^{-1}(V))$ for each $V \subset N$.
\end{definition}

Global hyperbolicity 
is used to ensure the interpolating measures $(\mu_s)_{s\in[0,1]}$ which make up each $q$-geodesic
inherit compact support from the endpoints $\mu_0$ and $\mu_1$. 
It also ensures various familiar properties of the Lorentz distance $\ell(x,y)$
recalled for the reader's convenience in the next two lemmas.


\begin{lemma}[Semicontinuity of Lorentz distance]\label{L:lapse continuity}
Let $(M,g)$ be a globally hyperbolic spacetime.
The Lorentz distance $\ell:M^2 \longrightarrow [0,\infty)\cup\{-\infty\}$ 
defined by $q=1$ in \eqref{q-lapse} is (a) upper semicontinuous on $M \times M$,
(b) continuous on $\ell^{-1}([0,\infty))$ and (c) smooth precisely on the complement 
of the closed set $\singl$.
  \end{lemma}


\begin{proof} Continuity of the function $\ell_+ := \max\{\ell,0\}$ is well known 
\cite[Corollary 4.7]{BeemEhrlichEasley96} \cite[Lemma 14.21-22]{O'Neill83}.
Claims (a)-(b) follow immediately since $\ell^{-1}([c,\infty]) = \ell_+^{-1}([c_+,\infty])$ 
 is closed for each $c \in \R$ with $c_+=\max\{c,0\}$.

The proof of (c) is deferred to Theorem~\ref{T:lapse smoothness} below;  
see also Proposition~9.29 of  \cite{BeemEhrlichEasley96}.
\end{proof}


\begin{lemma}[Midpoint continuity away from cut locus]\label{L:midpoint continuity}
For each $s \in [0,1]$ and $(x,y) \in M \times M \setminus \singl$ there is a unique $z=z_s(x,y) \in M$ such that
\begin{equation}\label{affine segment}
\ell(x,z) = s\ell (x,y) {\rm\ and}\ \ell(z,y) =(1-s)\ell(x,y). 
\end{equation}
Moreover, $z$ depends smoothly on $(s,x,y) \in [0,1] \times (M \times M \setminus \singl)$.
\end{lemma}

\begin{proof}
Let 
$(\xo,\yo) \in M \times M \setminus \singl$.
The definition of $\singl$ implies both $\xo$ and $\yo$ lie in the relative interior of some timelike action minimizing segment $s \in [0,1] \mapsto \sigma(s)$,  and $\yo$ lies in the chronological future of $\xo$.  
Thus $\yo$ is strictly within the timelike cut locus of $\xo$ which means (i) that there is a unique proper-time parameterized action minimizing geodesic $z_s(\xo,\yo)$ joining $\xo$ to $\yo$ (e.g.~Corollary 9.4 of \cite{BeemEhrlichEasley96}), hence a unique solution to \eqref{affine segment},  
(ii) it is given by $z_s(\xo,\yo) = \exp_{\xo} s\vo$ for some
$\vo = \vo(\xo,\yo) \in T_xM$, and 
(iii) $\xo$ and $\yo$ are non-conjugate, so 
the smooth map $(x,v) \in TM \mapsto \exp_x v \in M \times M$ acts diffeomorphically
on a neighbourhood of $(\xo,\vo)$.  Thus $\vo(x,y)$ depends smoothly on $(x,y)$ near $(\xo,\yo)$,
which implies $z_s(x,y)$ is smooth outside the closed set $\singl$ of 
Lemma~\ref{L:lapse continuity}.
\end{proof}

The preceding and following lemmas 
establish the interpolating point $z_s(x,y)$ and set $Z_s(\cdot)$ notations used throughout.

\begin{lemma}[Midpoint sets inherit compactness]
\label{L:compact support}
Fix a globally hyperbolic spacetime $(M,g)$.
Given $S \subset M \times M$ 
and $s \in [0,1]$ 
let 
\begin{eqnarray}\label{Z_s(X,Y)}
Z_s(S) &:=& \bigcup_{(x,y) \in S} Z_s(x,y) \qquad {\rm where}
\\ Z_s(x,y) &:=& \Big\{ z \in M \;\Big|\; 
\ell(x,z) = s\ell (x,y) {\rm\ and}\ \ell(z,y) =(1-s)\ell(x,y)  \Big\}
\label{Z_s(x,y)}
\end{eqnarray}
if $\ell(x,y) \ge 0$ and $Z_s(x,y) :=\emptyset$ otherwise.
%
%
If $S$ 
is precompact 
 then 
$\displaystyle Z(S) :=
\cup_{s \in[0,1]} Z_s(S)$ is precompact.
If, in addition, $S$ is compact then $Z(S)$ and $Z_s(S)$ are compact.
\end{lemma}

\begin{proof}
Since $Z_s(x,y) := \emptyset$ unless $\ell(x,y) \ge 0$ and
Lemma \ref{L:lapse continuity} implies $\{ \ell \ge 0\}$ is closed,
it costs no generality to restrict our attention to precompact sets
$S \subset \{ \ell \ge 0 \}$.  For such a set, fix an arbitrary sequence
$\{z_k\}_{k=1}^\infty$ in $Z(S)$.
Then there are a sequence of times $s_k \in [0,1]$ and
 timelike action minimizing geodesic segments $\sigma_k:[0,1] \longrightarrow M$
with endpoints $(x_k,y_k) := (\sigma_k(0),\sigma_k(1))$ in $S$ such that $\sigma_k(s_k)=z_k$.
Precompactness of $S$ yields a subsequential limit 
$(x_\infty,y_\infty) = \lim_{j \to \infty} (x_{k(j)},y_{k(j)})$ for the endpoints. Corollary~3.32 of \cite{BeemEhrlichEasley96}
yields a future-directed limit curve $\sigma_\infty$ of this subsequence which joins $x_\infty$ to $y_\infty$.
From its Arzel\`a-Ascoli based proof, we see more is true: letting $\tilde \sigma_k$ denote the reparameterization
of $\sigma_k$ with respect to its arclength for the Riemannian metric $\tilde g$ 
we have uniform convergence of $\tilde \sigma_{k(j)}$ to $\tilde \sigma_\infty$;
moreover this sequence of curves has arclength bounded by $c$ independent of $k$.
For each $k$ there exists $c_k \le c$ such that $z_k=\sigma_k(s_k)=\tilde \sigma_{k}(c_k)$.
Extracting a further subsequence without relabelling yields a limit $c_{\infty} = \lim_{j\to\infty}c_{k(j)}$.
Uniform convergence of the $1$-Lipschitz curves $\tilde \sigma_{k(j)}$ then gives
$\tilde \sigma_\infty(c_\infty) = \lim_{j \to \infty} \tilde \sigma_{k(j)}(c_{k(j)})$
to establish the desired subsequential limit of $\{z_{k}\}_{k=1}^\infty$.

On the other hand, if $S \subset \{ \ell \ge 0 \}$ is compact and $z_\infty$ is any accumulation point of the sequence
$\{z_k\}_{k=1}^\infty \subset Z(S)$ mentioned above,  then taking the limit of 
$$
\ell(x_k,z_k)= s_k \ell(x_k,y_k) \quad {\rm and} \quad  \ell(z_k,y_k)= (1-s_k) \ell(x_k,y_k)
$$
along a subsequence $(x_{k(j)},y_{k(j)}) \to (x_\infty,y_\infty)$ in $S$ with $z_{k(j)} \to z_\infty$ and $s_{k(j)} \to \bar s$,
the continuity of $\ell$ from Lemma~\ref{L:lapse continuity}
shows 
$z_\infty \in Z_{\bar s}(x_\infty,y_\infty)$
to establish compactness of $Z(S)$.  If $s_{k} =s$ for each $k$ then $\bar s=s$, so we have also established compactness of $Z_s(S)$.
\end{proof}

\begin{remark}\label{R:endpoints}
Note $\ell(x,y)>0$ implies $Z_0(x,y) = \{x\}$ and $Z_1(x,y) = \{y\}$.
Indeed,  if e.g.~$x \ne z \in Z_0(x,y)$, concatenating the action minimizing segment
linking $x$ to $z$ with that linking $z$ to $y$ yields an action minimizing segment from $x$ to $y$
which changes causal type from null to timelike, contradicting the smoothness of geodesics which 
follows from the Euler-Lagrange equation they satisfy.
\end{remark}

Lemma~\ref{L:lapse continuity} asserts $\singl$ to be closed.  
Define the {\em timelike injectivity locus} $TIL \subset TM$ to be the (unique) connected 
component of $\exp^{-1}[M \times M \setminus \singl]$ containing the zero section in its boundary. 
Let $TIL_0$ be the subset of the closure of $TIL$ on which the exponential map remains well-defined,
and $TIL_+ := TIL_0 \cap \exp^{-1}[\{\ell>0\}]$.
Recall that a map between topological spaces is {\em proper} if the preimage
of any compact set is compact.

\begin{corollary}[Proper action of the Lorentzian exponential]
\label{C:proper exponential}
Global hyperbolicity of $(M,g)$ implies the Lorentzian exponential restricts to a proper map
$\overline{\exp}:TIL_+\longrightarrow \{ \ell > 0\}$ on 
$TIL_+ \subset TM$.
\end{corollary}

\begin{proof} Let $\overline{\exp}$ denote the restriction of $\exp$ to $TIL_+$.
Given $S \subset \{ \ell > 0\} \subset M \times M$ compact and
$(p_i,x_i) \in (\overline{\exp})^{-1}S$, set $y_i = \exp_{x_i} p_i$ and $z_{i}:=\exp_{x_{i}} \frac12 p_{i}$.
The compactness of $S$ and $Z_{1/2}(S)$ shown in Lemma \ref{L:compact support}
provide a subsequence $(x_{i(k)}, y_{i(k)}, z_{i(k)})$ converging to a limit 
$(\xo,\yo, z_0) \in S \times Z_{1/2}(S)$,  with $z_0$ being the chronological midpoint of an
action minimizing segment joining $\xo$ to $\yo$.  Thus $(\xo,z_0) \not \in \singl$,  so 
$(\overline{\exp})^{-1}$ acts diffeomorphically near $(\xo,z_0)$.  Since 
${\ds z_0 = \lim_{k\to\infty} \exp_{x_{i(k)}}} \frac12 p_{i(k)}$ we conclude $\ds p_0 := \lim_{k \to\infty} p_{i(k)}$
exists and deduce $z_0 = \exp_{\xo} \frac12 p_0$ and $\yo = \exp_{\xo} p_0$.  Thus 
$(p_0,\xo) \in (\overline{\exp})^{-1} S$, to establish that $( \overline{\exp})^{-1} S$ is compact and
$ \overline{\exp}$ is proper.
\end{proof}

\begin{lemma}[Selecting midpoints on the timelike cut locus]\label{L:midpoint selection}
The maps  $z_s$ from Lemma \ref{L:midpoint continuity} can be measurably extended 
to $\{\ell>0\}$ by $\bar z_s$ so that $\ell(x,y)>0$ implies $s\in [0,1] \mapsto \bar z_s(x,y)$ is a proper-time maximizing geodesic segment joining $x$ to $y$. 
\end{lemma}

\begin{proof}
Global hyperbolicity of $(M,g)$ implies the infimum \eqref{q-lapse} is attained \cite{Avez63} \cite{Seifert67}
hence $Z_{1/2}(S)$ is non-empty \eqref{Z_s(X,Y)}  for each $\emptyset \ne S \subset \{\ell>0\}$;
it is closed if $S$ is, due to the continuity of $\ell(x,y)$ stated in Lemma~\ref{L:lapse continuity}; in particular
$Z_s(x,y)$ is closed and non-empty for $(x,y) \in  \{\ell>0\}$. 
The same lemma shows $\npl$ to be closed. Thus $\exp^{-1}[\{\ell>0\}]$ is open.
Define the {\em timelike injectivity locus}, $TIL \subset TM$ to be the connected component of 
$\exp^{-1}[M^2\setminus\singl]$ containing the zero section in its boundary.  Set $TIL(x) := \{ v \in T_x M \mid (x,v) \in TIL\}$
and let $\overline{TIL}(x)$ denote its closure.  Then $V(x,y):=\overline{TIL}(x) \cap \exp_x^{-1} Z_{1/2}(x,y)$
is closed and non-empty for each $(x,y) \in \{\ell>0\}$,  and $V(S) := \cup_{(x,y) \in S}V(x,y)$ is closed 
if $S \subset \{\ell>0\}$ is.  This shows $V$ to be a measurable 
correspondence with non-empty closed values between the space $\{\ell>0\}$ 
metrized by $d_{\tilde g} \oplus d_{\tilde g}$ and 
the Polish space $TM$, in the terminology of Aliprantis and Border; 
it therefore admits a measurable selection $\bar v:\{\ell>0\} \longrightarrow M$
such that $v(x,y) \in V(x,y)$ according to the Kuratowski-Ryll-Nardzewski Theorem, 
e.g.~18.13 of \cite{AliprantisBorder06}.
Now $\bar z_s(x,y) := \exp_s s\bar v(x,y)$ gives the desired measurable extension of 
$z_s$ from $M\times M \setminus \singl$
to $\{\ell>0\}$. By construction, $\ell(x,y) >0$ implies 
that $s\in [0,1] \mapsto \bar z_s(x,y)$ is a proper-time maximizing geodesic segment joining $x$ to $y$.
\end{proof}


We next identify the cases of equality in Eckstein and Miller's
reverse triangle inequality 
\cite{EcksteinMiller17} under the simplifying hypothesis $q \ne 1$:

\begin{proposition}[Reverse triangle inequality and cases of equality]
\label{P:RTIE}
Fix $0 < q < 1$. If $\mu_1,\mu_2,\mu_3 \in \PP(M)$
and $ 
\LL_q(\mu_1,\mu_2) \ne-\infty \ne \LL_q(\mu_2,\mu_3) 
$
then
\begin{equation}\label{L_q triangle inequality}
\LL_q(\mu_1,\mu_3) \ge \LL_q(\mu_1,\mu_2) + \LL_q(\mu_2,\mu_3);
\end{equation}
moreover, 
if $\mu_1[X_1]=1=\mu_3[X_3]$ and $ \LL_q(\mu_1,\mu_2) + \LL_q(\mu_2,\mu_3) <\infty$
then the inequality is strict
unless $\mu_2[Z(X_1 \times X_3)]=1$.

Conversely, if (i)
$\LL_q(\mu_1,\mu_3) \in (0,\infty)$,
(ii) equality holds in \eqref{L_q triangle inequality}, and (iii) 
both suprema \eqref{MK} defining $\LL_q(\mu_1,\mu_2)$ and $\LL_q(\mu_2,\mu_3)$ are attained, 
then there exists $\omega \in \PP(M^3)$ for which 
$\pi_{ij} :=\proj_{ij\#}\omega \in \Pi(\mu_i,\mu_j)$ is $\ell^q$-optimal for each $i<j$ with $i,j \in \{1,2,3\}$
and each $(x,y,z) \in \spt \omega$ satisfies
\begin{equation}\label{L_q triangle equality}
\ell(x,y) = s\ell(x,z)\ {\rm and}\ \ell(y,z) = (1-s)\ell(x,z)
\end{equation}
with 
$s:=\LL_q(\mu_1,\mu_2)/\LL_q(\mu_1,\mu_3)$ and $\proj_{ij}(x_1,x_2,x_3)=(x_i,x_j)$.
If, in addition, $\pi_{13}[S]=1$ for some $S\subset M \times M$ 
then $\mu_2$ vanishes outside $Z_s(S)$. 
In particular, if
$Z_s(x,y)=\{z_s(x,y)\}$ holds for $\pi_{13}$-$a.e.\ (x,y)$, 
then $\omega = (z_0 \times z_s \times z_1)_\#\pi_{13}$ and $\mu_2 = z_{s\#}\pi_{13}$
in the notation of Definition \ref{D:push-forward}.
\end{proposition}

\begin{proof}  
If either term on the right hand side of \eqref{L_q triangle inequality}
diverges to $-\infty$ there is nothing to prove. Otherwise, given $\epsilon>0$  
 there exist $\pi_{12} \in \Pic(\mu_1,\mu_2)$ and $\pi_{23} \in \Pic(\mu_2,\mu_3)$ which are 
nearly $\ell^q$-optimal, 
in the sense that 
$$
\ell^q[\pi_{ij}] :=
\int_{M^2} \ell^q d\pi_{ij} \ge \min\{\LL_q(\mu_i,\mu_j) - \epsilon, \epsilon^{-1}\}
$$ 
for $j=i+1 \in \{2,3\}$.
Disintegrate $d\pi_{12}(x,y) = d\mu_2(y) d\pi_{12}^y(x)$ and 
$d\pi_{23}(y,z) = d\mu_2(y) d\pi_{23}^y(z)$ and define $\omega$ by `gluing': i.e.,
$$
\int _{M^3}\phi(x,y,z) d\omega(x,y,z) = \int_M d\mu_2(y) \int_{M^2} \phi(x,y,z) d\pi_{12}^y(x) d\pi_{23}^y(z),
$$
as in e.g. \cite[Definition 16.1]{Villani09}. Then $\pi_{13} := \proj_{13\#} \omega  \in \Pi(\mu_1,\mu_3)$ and
\begin{eqnarray*}
\LL_q(\mu_1,\mu_3) 
&\ge& \| \ell(x,z) \|_{L^q(d\pi_{13})}
\\ &=& \| \ell(x,z) \|_{L^q(d\omega)}
\\ &\ge&  \| \ell(x,y) + \ell(y,z) \|_{L^q(d\omega)}
\\ &\ge&   \| \ell(x,y) \|_{L^q(d\omega)} +  \| \ell(y,z) \|_{L^q(d\omega),}
\\ &\ge& \min\{\LL_q(\mu_1,\mu_2) + \LL_q(\mu_2,\mu_3) - 2\epsilon,\epsilon^{-1} -\epsilon\}
\end{eqnarray*}
where the inequalities follow from the definition \eqref{MK} of $\LL_q$, reverse triangle inequality
\eqref{reverse triangle inequality} for Lorentz distance, and the (reverse) Minkowski inequality for $q \in (0,1]$.
In particular $\pi_{13} \in \Pic(\mu_1,\mu_3)$.  Since $\epsilon>0$ was arbitrary, 
\eqref{L_q triangle inequality} is established.
Assume $\mu_i[X_i]=1$, so $\omega$ vanishes outside $X_1 \times M \times X_3$. 
For $(x,y,z) \in X_1 \times M \times X_3$ with $\ell(x,z) \ge 0$, inequality  \eqref{reverse triangle inequality}
holds strictly unless $y \in Z(X_1 \times X_3)$; 
since $\ell(x,z) \ge 0$ holds $\omega$-a.e.,  $ \LL_q(\mu_1,\mu_2) + \LL_q(\mu_2,\mu_3) <\infty$
implies \eqref{L_q triangle inequality}
is strict unless $\omega$ vanishes outside $X_1 \times Z(X_1 \times X_3) \times X_3$,
or equivalently, unless $\mu_2=\proj_{2\#}\omega$ vanishes outside $Z(X_1 \times X_3)$.

Now assume the suprema \eqref{MK} defining $\LL_q(\mu_1,\mu_2)$ and $\LL_q(\mu_2,\mu_3)$ are both finite and attained, so that we can henceforth fix $\epsilon=0$ in the argument above.
When \eqref{L_q triangle inequality} is saturated, each of the three inequalities in the 
preceding chain of claims must be saturated as well.  
Saturation of the first asserts $\ell^q$-optimality of $\pi_{13}$.
For $\omega$-a.e. $(x,y,z)$, saturation of the second shows $\ell(x,z)=\ell(x,y)+ \ell(y,z)$,
while the third (Minkowski) asserts the existence of $s \in [0,1]$ such that
that $(1-s)\ell(x,y) = s\ell(y,z)$.
Combining the last two identities 
asserts that \eqref{L_q triangle equality} holds $\omega$-a.e.; comparison with \eqref{L_q triangle inequality}
forces $s:=\LL_q(\mu_1,\mu_2)/\LL_q(\mu_1,\mu_3)$. Since $\omega$ vanishes outside the closed set
$\{ (x,y,z) \mid \min\{\ell(x,y),\ell(x,z), \ell(y,z) \} \ge 0\}$, the continuity of $\ell$ from Lemma~\ref{L:lapse continuity}
implies \eqref{L_q triangle equality} extends to all $(x,y,z) \in \spt \omega$.

Now suppose  $\pi_{13}[S]=1$ for some $S \subset M \times M$.  
Then $\omega[\tilde S]=1$ where $\tilde S:= \{(x,y,z) \mid (x,z) \in S\}$.
For any Borel set $A \subset M$ disjoint from $Z_s(S)$,
it follows that $M \times A \times M$ is disjoint from $\tilde S$,
hence $\mu_2(A) = \omega[M \times A \times M]=0$ as desired.
For example, suppose $Z_s(x,y)=\{z_s(x,y)\}$ holds for $\pi_{13}$-$a.e.\ (x,y)$.
Then $\omega$ vanishes outside the graph of 
$z_0 \times z_s \times z_1:\spt \pi_{13} \longrightarrow M \times M \times M$,  whence
$\omega = (z_0 \times z_s \times z_1)_\# \pi_{13}$ by e.g.\ Lemma 3.1 of \cite{AhmadKimMcCann11}.
\end{proof}

\begin{corollary}[Interpolants inherit compact support]\label{C:compact support}
Let $(\mu_s)_{s \in [0,1]} \subset \PP(M)$ be a $q$-geodesic on globally hyperbolic spacetime $(M,g)$.
If $\mu_0$ and $\mu_1$ have compact support,  then $\spt \mu_s \subset Z_s(\spt[\mu_0 \times \mu_1])$ 
and the latter is compact for each $s \in [0,1]$.
\end{corollary}

\begin{proof}
Fix a $q$-geodesic $(\nu_t)_{t \in [0,1]} \subset \PP(M)$ with compactly supported endpoints.
Setting $X_t := \spt \nu_t$, 
Lemma \ref{L:compact support} shows the compactness of $Z:=  Z(X_0 \times X_1)$
and $Z_s:=Z_s(X_0 \times X_1)$.
Let $0\le s<t \le 1$ be arbitrary.
From definition \eqref{q-geodesic} we see $0<\LL_q(\nu_s,\nu_t) <\infty$, so
taking $(\mu_1,\mu_2,\mu_3)=(\nu_0,\nu_s, \nu_1)$ yields equality in \eqref{L_q triangle inequality}.
The first part of Proposition \ref{P:RTIE} asserts $\nu_s$  vanishes outside of $Z$
--- hence is compactly supported.  Since $\sup_{Z} \ell <\infty$,  the suprema \eqref{MK}
defining $\LL_q(\nu_s,\nu_t)$ is attained, and the second part of Proposition~\ref{P:RTIE}
asserts $\spt \nu_s \subset Z_s$ as desired.
\end{proof}

\begin{theorem}[Existence of $q$-geodesics]\label{T:q-geodesics exist}
Let $(M,g)$ be a globally hyperbolic spacetime and $0<q\le 1$. 
Fix $\mu,\nu \in \PP(M)$ 
and {\em suppose \eqref{MK} is finite and attained by some}
$\pi \in \Pi(\mu,\nu)$ with 
$\ell>0$ holding $\pi$-a.e.
 Then (i) 
$\mu_s := \bar z_{s \#} \pi$ defines a $q$-geodesic $s \in [0,1] \mapsto \mu_s \in \PP(M)$ 
where $\bar z_s(x,y)$ is from Lemma~\ref{L:midpoint selection}.
(ii) If $0 \le s < t \le 1$ then $(\bar z_s \times \bar z_t)_\#\pi$ is $\ell^q$-optimal.
(iii) If 
$\mu_0$ and $\mu_1$ are compactly supported
and the maximum \eqref{MK} is uniquely attained and $\pi[\singl]=0$,
then the $q$-geodesic joining $\mu=\mu_0$ to $\nu = \mu_1$ is unique.
\end{theorem}

\begin{proof}
%
(i)-(ii) Taking $\mu,\nu \in \PP(M)$ 
and $0 \le s < t \le 1$ as hypothesized,
suppose \eqref{MK} is finite and attained by some $\pi \in \Pic(\mu,\nu)$ with $\ell>0$ holding $\pi$-a.e.
Use the extension $\bar z_s$ of $z_s$ from Lemma \ref{L:midpoint selection} to define 
$\mu_s := \bar z_{s \#}\pi$.
 Trying $(\bar z_s \times \bar z_t)_\#\pi \in \Pic(\mu_s,\mu_t)$ in \eqref{MK} shows
\begin{eqnarray*}
\LL_q(\mu_s,\mu_t)^q 
&\ge& \int \ell(\bar z_s(x,y),\bar z_t(x,y))^q d\pi(x,y) 
\\ &=& (t-s)^q \int \ell(x,y)^q d\pi(x,y)
\\ &=&(t-s)^q  \LL_q(\mu_0,\mu_1)^q
\end{eqnarray*}
from the fact that $s \in [0,1] \mapsto \bar z(x,y)$ is a proper time maximizing segment for $\pi$-a.e.~$(x,y)$
and the optimality of $\pi \in \Pic(\mu_0,\mu_1)$.

These lower bounds are finite and positive by hypothesis, and imply 
$$
\LL_q(\mu_0,\mu_s) + \LL_q(\mu_s,\mu_t) + \LL_q(\mu_t,\mu_0)
\ge \LL_q(\mu_0,\mu_1).
$$
The reverse triangle inequality proved in Proposition \ref{P:RTIE} forces both inequalities to become equalities.
Thus  $s \in [0,1] \mapsto \mu_s \in \PP(M)$ is a $q$-geodesic, 
%
and $(\bar z_s \times \bar z_t)_\#\pi$ is $\ell^q$-optimal.



(iii) If $\mu_0$ and $\mu_1$ are compactly supported,
Corollary \ref{C:compact support} asserts the same is true for $\mu_s$. 
When the maximum \eqref{MK} is uniquely attained by $\pi \in \Pic(\mu,\nu)$,  uniqueness of $\mu_s$
follows from the last assertion in Proposition \ref{P:RTIE},
whose hypotheses are satisfied because $\pi_{13}=\pi$ was assumed to vanish on $\singl$,
and because compact support guarantees the suprema \eqref{MK} defining 
$\ell(\mu_s,\mu_t)$ in \eqref{q-geodesic} is attained for each $0 \le s<t \le 1$. 
\end{proof}



\section{
Lagrangian, Hamiltonian, Lorentz distance}
\label{S:lapse}

In this section, first- and second-variation formulas are used to establish
properties of the Lorentz distance which will
be useful throughout, 
along with convex-analytic properties 
of the Lagrangian and Hamiltonian which define it.
Although it would not  be surprising to learn they have been studied elsewhere, 
we have not seen the family of Lagrangians 
\begin{equation}\label{q-Lagrangian}
L(v;q) := \left\{ \begin{array}{cl} 
- 
 (g_{ab}(x) v^a v^b)^{q/2}/q
& \mbox{\rm if $v$ is future-directed and}\ g_{ab}(x) v^a v^b \ge 0,
\\ +\infty &  {\rm else}  
\end{array}\right.
\end{equation}
discussed previously --- apart from the case $q=1$ \cite{Suhr16p}.
Propositions \ref{P:lapse semiconvexity} and Theorem \ref{T:lapse semiconcavity fails}
are inspired by corresponding results from the Riemannian 
setting~\cite{CorderoMcCannSchmuckenschlager01}, but the Lorentzian versions appear to be new.
They are based on concepts from non-smooth analysis recalled here which will also be useful later:
sub- and superdifferentiability,  semiconvexity and -concavity, approximate derivatives.


On a Riemannian manifold $(M^n,\tilde g)$,  a function $u:M \longrightarrow [-\infty,+\infty]$
is said to be {\em subdifferentiable} at $x$ with subgradient $p \in T_x^*M$
 if  $x \in \dom u:= u^{-1}(\R)$ and
\begin{equation}\label{subgradient}
u(\exp_x v) \ge u(x) + p[v] + o(|v|_{\tilde g})
\end{equation}
holds for small $v \in T_xM$.  Here $p[v]$ denotes the duality pairing of $p$ with $v$.
It doesn't matter whether the Riemannian or Lorentzian exponential
is used in this definition,  since they agree to order $o(|v|_{\tilde g})$.
The set of subgradients for $u$ at $x$ is denoted by $\p u(x)$,  or by
$\p_\cdot u(x)$ when we need to distinguish it from the set $\p^\cdot u(x)$ of supergradients.
Here $p$ is a supergradient if inequality \eqref{subgradient} is reversed,  in which case
we say $u$ is superdifferentiable at $x$.  If $u$ has both a sub- and a supergradient at $x$,
then 
$u$ is differentiable at $x$,  in which case the super- and subdifferentials
 $\p_\cdot u(x) = \p^\cdot u(x) = \{Du(x)\}$ agree,
and we write $x \in \dom Du$.

\begin{lemma}[Convex Lagrangian and Hamiltonian]
\label{L:q-Lagrangian}
Fix $0<q<1$ and a point $x$ on a Lorentzian manifold $(M^n,g)$. (i) The Lagrangian
\eqref{q-Lagrangian}
is convex on $T_xM$;  where $L<0$ it is smooth and its Hessian 
\begin{equation}\label{q-Lagrangian Hessian}
|v|^{2-q} g^{ij}g^{kl} \frac{\p^2 L}{\p v^k \p v^l} = (2-q) |v|^{-2} v^i v^j - g^{ij}
\end{equation}
is positive-definite,  so
strict convexity fails only on the future light cone.  
(ii) Subdifferentiability of 
$L(\,\cdot\,;q)$ fails throughout the light cone. 
(iii) The corresponding Hamiltonian on the cotangent space $T_x^*M$ is given by
\begin{equation}\label{q-Hamiltonian}
H(p;q) := \left\{ \begin{array}{cl} 
-(g^{ab}(x) p_a p_b)^{q'/2}/q' 
& \mbox{\rm if $p$ is past-directed and}\ g^{ab}(x) p_a p_b>0
\\ +\infty &  {\rm else}
\end{array}\right.
\end{equation}
with $\frac1q + \frac1{q'} 
=1$; it satisfies $v=DH(DL(v;q);q)$ and $p=DL(DH(p;q);q)$
for all timelike future-directed $v\in T_xM$ and timelike past-directed $p \in T_x^*M$.
\end{lemma}

\begin{proof}
(i) In the interior of the future cone,  $L$ is smooth and we compute that
$\frac{\p L}{\p v^i} = - |v|^{q-2} g_{ij}v^j$ is past-directed (because of the minus sign) and
thence \eqref{q-Lagrangian Hessian}
with $|v|:= \sqrt{v^a g_{ab} v^b}$.
Since $2-q>1$ and the reverse Cauchy-Schwartz inequality asserts
$(w^a g_{ab} v^b)^2 \ge  (w^q g_{ab} w^b)(v^a g_{ab} v^b)$ whenever $v$ is timelike, 
we conclude $D^2L$ is non-negative definite, c.f.\ \S2.4 of \cite{BeemEhrlichEasley96}; 
the obsevation that $w\ne 0$ is spacelike whenever it
is orthogonal to $v$ improves this to positive-definiteness.
Since the future cone $L \ge 0$ is convex and $L$ is continuous on it and $+\infty$ outside,
it follows that $L$ is convex on $T_xM$.  

(ii) For $q<1$ we see $|DL(v;q)|=|v|^{q-1}$ diverges as 
$|v| \to 0$.  This shows the subdifferential $\p L(v;q)$ is empty at each point $v$ on the light cone,
since \cite[Theorem 25.6]{Rockafellar70} asserts $\p L(v;q) = N + A$ where $A$ is the set of
accumulation points of $DL(v_k;q)$ with $v_k \to v$, and $N$ is the normal ray to the light cone at $v$.
 In this case $A$ is empty.

(iii) We readily see that  $DH$ and $DL$ invert each other on the specified
cones 
using $\frac{\p H}{\p p_k}= - |p|^{ q' -2} g^{kj}p_j$ and   $(q-1)(q'-1)=1$.
\end{proof}

\begin{corollary}[The classical Lagrangian and Hamiltonian] 
\label{C:1-Lagrangian}
The limit $L(v;1) =\ds \lim_{q \to 1^-} L(v;q)$ inherits convexity from $L(v;q)$,
but fails to be strictly convex along any ray through the origin.
(ii) Its convex 
dual Hamilton $H(p;1) := \ds \sup_{v \in T_xM} p[v] - L(v;q)$ 
is the indicator function of a past-directed solid hyperboloid:
\begin{equation}\label{indicator Hamiltonian}
H(p;1) = \left\{ \begin{array}{cl} 
0 &   \mbox{\rm if $p$ is past-directed and}\ g^{ab}(x) p_a p_b \ge 1,
\\ +\infty &{\rm else}.
\end{array}\right.
\end{equation}
(iii) Although $L$ is smooth in the interior of the future cone,  subdifferentiability of $L(v;1)$ 
fails at each point on the lightcone apart from the origin,  where its subdifferential
$\p L({\bf 0};1) = \{p 
\mid H(p;1)=0\}$ is the solid hyperboloid.
\end{corollary}

\begin{proof}
(i) Lemma \ref{L:q-Lagrangian} makes convexity of $L(v;q)$ and hence its $q \to 1^-$ limit $L(v;1)$ 
clear. Strict convexity fails along each ray 
 due to the positive $1$-homogeneity of $L(\lambda v;1) = \lambda L(v;1)$ for each $\lambda>0$.  

(ii) The Legendre transform of the limit is the limit of the Legendre transforms:
$$
H(p;1) = \lim_{q \to 1^-} H(p;q).
$$
Formula \eqref{indicator Hamiltonian} now follows from \eqref{q-Hamiltonian}.

(iii) 
Since $p \in \p L(v;1)$ if and only if $v \in \p H(p;1)$,
we deduce $\{p \mid H(p,1)=0\} \subset \p L({\bf 0};1)$.
Equality must hold since $H(\,\cdot\,;1)$ is not subdifferentiable outside its zero set.
Each slope $p$ in the interior of the hyperboloid $\dom H(\,\cdot\,;1)$ supports the graph
of $L(\,\cdot\,;1)$ only at the origin.  Each slope on the hyperboloid boundary supports the graph
of $L(\,\cdot\,;1)$ on an entire ray. Apart from the origin,  this ray lies in the interior of the future-cone
$\dom L(\,\cdot\,;1)$,  since the hyperboloid is strictly convex.  Thus $L(\,\cdot\,;1)$ cannot be subdifferentiable
on the lightcone,  except at the origin.
\end{proof}

Recall the following definition of semiconvexity from, e.g.~\cite{Bangert79} \cite{CorderoMcCannSchmuckenschlager01},  
which is independent of the choice of Riemannian metric $\tilde g$ on $M$ (though the semiconvexity constant $C$ may depend on $\tilde g$).

\begin{definition}[Semiconvexity]\label{D:semiconvex}
Fix $U \subset M$ open. 
A function 
$u:U \longrightarrow \R$ is {\em semiconvex} on $U$ if
there is a constant
$C\in \R$ 
such that
$$
\mathop{\lim \inf}\limits_{w \to 0} \frac{ u(\exp^{\tilde g}_x w) + u(\exp^{\tilde g}_x -w) - 2 u(x)}{2|w|_{\tilde g}^2} \ge C.
$$
for all $x \in U$.
The largest such $C$ is called the the {\em semiconvexity constant} of $u$ on $U$.
Similarly,  $u$ is called {\em semiconcave} if $-u$ is semiconvex.
\end{definition}


\begin{proposition}[Semiconvexity of Lorentz distance] 
\label{P:lapse semiconvexity}
For any smooth Riemannian metric $\tilde g$ on a globally hyperbolic manifold $(M,g)$, the limit
\begin{equation}\label{lapse semiconvexity}
\tilde C(x,y) 
:= \mathop{\lim \inf}\limits_{w \to 0} \frac{ \ell(\exp^{\tilde g}_x w,y) + \ell(\exp^{\tilde g}_x -w,x) - 2 \ell(x,y)}{2|w|_{\tilde g}^2} 
\end{equation}
is continuous and real-valued on $\{(x,y) \mid \ell(x,y) >0 \}$.
\end{proposition}

\begin{proof}
Suppose $\ell_0 := \ell(x,y)>0$ and
let $\sigma$ be an action minimizing geodesic from $x=\sigma(0)$ to $y=\sigma(1)$ with 
$|\dot \sigma(s)|_g=\ell_0$.
Given $w \in T_xM$ with $|w|_{\tilde g}=1$, 
let $w(s)$ be the Lorentzian parallel transport of $w$ along $\sigma$ and set 
$$
W(s) := (1-s) w(s) 
$$
so $W'(s) = - w(s)$.
Use the Riemannian exponential map to 
define a variation $\beta(r,s)$ around $\sigma(s)$ by
$$
\sigma_r(s) := \beta(r,s) = \exp^{\tilde g}_{\sigma(s)} (r W(s))
$$
with variable initial point $\sigma_r(0) = \exp^{\tilde g}_x rw$ but fixed final point 
$\sigma_r(1) = y$. Now use \eqref{q-lapse} to estimate
$$
\ell(\exp^{\tilde g}_x rw,x) \ge 
-a(r),
$$
by the action $a(r):= A[\sigma_r;1]$; equality holds when $r=0$.  
Thus we can bound the Riemannian Hessian of the Lorentz distance by that of the length (or action) functional
$-A[\sigma;1]$:

$$
\frac{\ell(\exp_x^{\tilde g} rw,y)+ \ell(\exp_x^{\tilde g} -rw,y) - 2 \ell(x,y)}{2r^2} \ge  - \frac{a(r) + a(-r) - 2a(0)}{2r^2}
$$
The expression on the right converges:  letting $\frac{D}{dr}$ denote Lorentzian covariant differentiation
along the curve $r \mapsto \sigma_r(s)$,  its limit is given by Synge's second variation formula,
e.g.~Theorem 10.4 of \cite{O'Neill83}:
\begin{equation}\label{Synge}
\ell_0 \frac{d^2 a}{dr^2}\bigg|_{r=0}
= - \langle \sigma', \frac{D}{\p r} \frac{\p \beta}{\p r} 
 \rangle \bigg|_0^1 
- \int_0^1 [\langle W'_\perp, W'_\perp \rangle - R(W_\perp,\sigma',W_\perp, \sigma') ] ds,
\end{equation}
where $W_\perp := W - \langle W,\sigma' \rangle \sigma'/\ell_0^2 = (1-s) w_\perp(s)$ is the projection of $W$ 
onto the orthogonal subspace of the geodesic $\sigma$, whose geodesy implies 
$W_\perp' 
= W' - \langle W',\sigma' \rangle \sigma'/\ell_0^2
= - w_\perp(s)$.
Thus
\begin{eqnarray*}
\ell_0 \frac{d^2a}{dr^2}\bigg|_{r=0} 
&=& - \langle \sigma', \frac{D}{\p r} \frac{\p \beta}{\p r} \rangle \bigg|_0^1
+ \int_0^1 [  (1-s)^2 R(w_\perp,\sigma',w_\perp,\sigma')-  |w_\perp|_g^2] ds
\\ &\le& 
- \langle \sigma',  \frac{D}{\p r} \frac{\p \beta}{\p r}  \rangle \bigg|_0^1
+\sup_{s \in [0,1]} 
(1-s)^2 R(w_\perp, \sigma', _\perp, \sigma') - \langle w_\perp, w_\perp \rangle_g 
\\ &=:& C(x,y; w(0), \sigma'(0)),
\end{eqnarray*}
where $C(x,y; w, \sigma')$ is a continuous real-valued function of all four of its arguments
(since the expression under the supremum is smooth --- hence locally uniformly continuous ---
 in the same variables, and $[0,1]$ is compact).  Similarly,
continuity of
$$
C(x,y) :=  \sup_{1=|w|_{(T_{x}M, \tilde g)} }  \sup_{v \in (\overline{\exp}_{x})^{-1}y} C(x,y; w, v)
$$
on $\{\ell >0\}$ follows from the compactness of $(\overline{\exp}_x)^{-1}y$
established in Corollary~\ref{C:proper exponential}.  
Taking $\tilde C(x,y) = -C(x,y)/\ell(x,y)$,  the continuity of $\ell(x,y)$
recalled in Lemma~\ref{L:lapse continuity}
concludes the proof of \eqref{lapse semiconvexity}.
\end{proof}

\begin{theorem}[Semiconcavity fails on the timelike cut locus]
\label{T:lapse semiconcavity fails}
If $(x,y) \in \{\ell>0\} \cap \singl$ then
$$
\sup_{0<|w|_{\tilde g}<1} \frac{\ell(\exp^{\tilde g}_x w,y) + \ell(\exp^{\tilde g}_x -w,y) - 2 \ell(x,y)}{2|w|_{\tilde g}^2} = +\infty
$$
\end{theorem}

\begin{proof}
Fix $(x,y) \in \{\ell>0\} \cap \singl$.   If $x$ is a cut point --- meaning multiple action minimizing curves link $x$ to $y$ --- the proof is easy.   Therefore,  assume $x$ is a conjugate point of $y$.  To derive a contradiction,  assume also
\begin{equation}
\label{lapse semiconcavity}
\mathop{\lim \sup}\limits_{w \to 0} 
\frac{\ell(\exp^{\tilde g}_x w,y) + \ell(\exp^{\tilde g}_x -w,y) - 2 \ell(x,y)}{2|w|_{\tilde g}^2} 
< \tilde C \in \R.
\end{equation}
This means the function $u(\wdot) := \ell(\wdot,y)$ has a quadratic upper bound at $x$.
Proposition \ref{P:lapse semiconvexity} implies $u(\wdot)$ also admits a quadratic lower bound at the same point. To first order,  these bounds must agree, hence $u(\wdot)$ is differentiable at $x$.  
We claim \eqref{lapse semiconcavity} implies an analogous bound \eqref{lapse semiconcavity 2} 
for the second difference quotients of $u$ along Lorentzian rather than 
Riemannian geodesics,  but possibly with a larger constant $C>\tilde C$.  
Indeed, given a Riemannian unit vector $w \in T_x M$,  letting 
$\gamma_r = \exp_x r w$, in Riemannian normal coordinates around $x$ we find 
$$
{u(\gamma_r) + u(\gamma_{-r}) - 2 u(x)}
\le  \sum_{\alpha=1}^n \frac{\p u}{\p x^\alpha}[{\gamma_r + \gamma_{-r} - 2x}]^\alpha 
+ \tilde C (\gamma^\alpha _r)^2 + \tilde C(\gamma_{-r}^\alpha)^2
$$
for $r$ sufficiently small. Thus
\begin{equation}
\label{lapse semiconcavity 2}
\mathop{\lim \sup}\limits_{w \to 0} 
\frac{\ell(\exp_x w,y) + \ell(\exp_x -w,y) - 2 \ell(x,y)}{2|w|_{\tilde g}^2} 
\le  C 
\end{equation}
where 
$$C = \tilde C + |Du|_{\tilde g}\sup_{|w|_{\tilde g}=1} \Big| \frac{\tilde D}{\p r} \frac{d\gamma_r}{dr } \Big|_{\tilde g} <\infty
$$
and $\frac{\tilde D}{\p r}$ denotes Riemannian covariant differentiation.


Let $\sigma(s)$ be the proper time maximizing geodesic segment joining $x =\sigma(0)$ to $y = \sigma(1)$,
and define its index form by
$$
I(W_1,W_2) :=-  \int_0^1 \langle W_{1\perp}',W_{2\perp }' \rangle - R(W_{1\perp},\sigma', W_{2\perp},\sigma') ds.
$$
where $W_i'$ denotes the covariant derivative of $W_i$ along $\sigma$, 
and $W_\perp := W - g(W,\sigma'(s))\sigma'(s)/|\sigma'(s)|_g^2$ denotes the component of $W$ orthogonal to $\sigma$.

Let $U(s)$ be a non-zero Jacobi field along $\sigma$ 
vanishing at its endpoints $s \in \{0,1\}$.  Notice $w := U'(0)$ cannot be a multiple of $\sigma'(0)$,
since the initial conditions $(0,\sigma'(0))$ generate the solution $(s\sigma(s))_{s \in [0,1]}$ to Jacobi's equation
which corresponds to simply stretching the geodesic. Thus 
the component $w_\perp$ 
of $w$ orthogonal to $\sigma'(0)$ is spacelike. 
Scaling the Jacobi field $U$ and the Riemannian and Lorentzian metrics 
independently we may assume $\ell(x,y)=1$ and $|w|_{\tilde g}=1 = - g(w_\perp,w_\perp)$.
Let $w(s)$ be a parallel field along $\sigma$ with $w(0)=w$ and set $W(s) := (1-s) w(s)$.  
Fix $\epsilon>0$ small enough that 
$$
I(W,W) < - C + \frac2\epsilon,
$$
and then let $U_\epsilon(s) := U(s) + \epsilon W(s)$ be a perturbation of the Jacobi field in question.
Introduce the variation $\sigma_r(s) = \beta(r,s) := \exp_{\sigma(s)} r U_\epsilon(s)$ around the 
geodesic segment $\sigma_0$.  Its action is
$a(r) := A[\sigma_r;1]$.  Since $\sigma_r$ joins $\exp_x  r \epsilon w$ to $y$, \eqref{q-lapse} implies 
$$
\ell(\exp_x \epsilon r w,y) \ge -a(r),
$$
with equality when $r=0$.  Thus, by assumption \eqref{lapse semiconcavity 2}
$$
\lim_{r \to 0} \frac{a(r) + a(-r) - 2a(0)}{r^2\epsilon^2} \ge -C.
$$
Noting that $r \mapsto \beta(r,s)$ is a geodesic for each $s \in [0,1]$,
the endpoint terms vanish in Synge's second variation formula \eqref{Synge}, giving
\begin{eqnarray*}
-\epsilon^2 C &\le& a''(0) 
\\ &=& I(U_\epsilon,U_\epsilon) + g ( \frac{D}{Dr} \frac{\p \beta}{\p r}(r,s), \dot \sigma(s) ) \Big|_{(r,s)=(0,0)}^{(r,s)=(0,1)}
\\ &=& I(U,U) + 2\epsilon I(U,W) + \epsilon^2 I(W,W)
\\ &<& 0 + 2\epsilon g(w_\perp,w_\perp) - C\epsilon^2 + 2 \epsilon
\end{eqnarray*}
by our choice of $\epsilon$, since  $U(s)$ Jacobi with vanishing
endpoints implies $I(U,U)=0$  and 
$I(U,W)= -g(W_\perp(s),U_\perp'(s)) |_{s=0}^{s=1} = g(w_\perp,w_\perp)=-1$,
noting our choice of $W(s)$.
This yields the contradiction desired to establish the theorem. 
\end{proof}




For convenient reference, we collect together several consequences of the foregoing analysis along
with the known results of Lemma \ref{L:lapse continuity}, and provide the deferred proof of part (c) of that lemma.

\begin{theorem}[Smoothness of Lorentz distance]\label{T:lapse smoothness}
Let $(M,g)$ be a globally hyperbolic spacetime.
The Lorentz distance $\ell:M^2 \longrightarrow [0,\infty)\cup\{-\infty\}$ 
defined by $q=1$ in \eqref{q-lapse} is (a) upper semicontinuous.  It is 
(b) continuous on $\ell^{-1}([0,\infty))$,  
(c) smooth precisely on the complement of 
the closed set $\singl$,
(d) 
 locally Lipschitz and locally semiconvex on the open set $\{\ell>0\}$.
Moreover, if $y=\exp_x v$ and $x = \exp_y w$ for  $(x,y) \in \ell^{-1}((0,\infty))$, then
$-\frac{v_*}{|v_*|_g} \in \p_\cdot u(x)$ and  $-\frac{w_*}{|w_*|_g} \in \p_\cdot \bar u(y)$,
where  $u(\wdot) := \ell(\wdot,y)$,  $\bar u(\wdot) := \ell(x,\wdot)$ and $v_*[\wdot] = g(v, \wdot)$.
However, (e) the superdifferential of $\ell(\wdot,y)$ is empty at $x$ if $\ell(x,y)=0$ 
unless $x=y$,
in which case 
the supergradients lie in  the solid hyperboloid 
$\{p \in T_x^*M \mid H(p;1)=0\}$.  
\end{theorem}


\begin{proof}  (a)-(b) were proven in Lemma \ref{L:lapse continuity}, based on the 
continuity of the function $\ell_+ := \max\{\ell,0\}$ from
 \cite[Corollary 4.7]{BeemEhrlichEasley96} \cite[Lemma 14.21-22]{O'Neill83}.

(d) Openness of $\{ \ell>0 \}$ also follows from the 
continuity $\ell_+ := \max\{\ell,0\}$. 
Semiconvexity of $\ell$ was established in Proposition~\ref{P:lapse semiconvexity}
and, in combination with (b), implies $\ell$ is locally Lipschitz (since locally bounded convex functions are locally Lipschitz). Apart from an overall change of sign,  the explicit form of the subgradient of $\ell(\wdot,y)$ at $x$
and $\ell(x,\wdot)$ at $y$ 
follows from the endpoint terms in the first variation formula,  
as in the proof of Proposition 10.15(i) of Villani \cite{Villani09}.  Although the statement of that Proposition
requires $L \in C^1$,  the proof makes it clear that it is enough to have tangent bundle estimates for $L$ and its
first derivative in a neighbourhood of the minimizing geodesic 
joining $x$ to $y$.  We have these estimates since
$\ell(x,y)>0$ ensures the geodesic in question is timelike.  As in the Riemannian case \cite{McCann01},
an alternative proof may also be constructed based on the existence of convex neighbourhoods, 
the Gauss Lemma, and the reverse triangle inequality, similarly in strategy to the proof of (e) below.  


(c) To see $\singl$ is closed, suppose $(x_k,y_k) \in \singl$ converge to $(x_0,y_0)$.
By global hyperbolicity there is a proper-time maximizing segment joining 
$(x_0,y_0)$ \cite[Theorem 3.18]{BeemEhrlichEasley96};  
the only question is whether it has a proper-time maximizing extension in
one and hence both \cite[Theorem 9.12]{BeemEhrlichEasley96} directions.
If $\ell(x_0,y_0) \le 0$ then $(x_0,y_0) \in \singl$,  so assume $\ell(x_0,y_0)>0$. Then (a) shows $\ell(x_k,y_k)>0$ eventually.  If $(x_k,y_k)$ are conjugate along a subsequence then $(x_0,y_0)$ are conjugate,
hence  in $\singl$ \cite[Theorem 9.11]{BeemEhrlichEasley96}.
Otherwise, eventually each $(x_k,y_k)$ are joined by a pair of distinct proper-time maximizing segments.  
From this we can extract either distinct proper-time maximizing segments linking $(x_0,y_0)$,  or a Jacobi field
which shows $x_0$ to be conjugate to $y_0$.  In either case $(x_0,y_0) \in \singl$.

Concerning smoothness of $\ell$:  since points near $(x_0,y_0) \not\in \singl$ are timelike separated but not conjugate,
the inverse function theorem guarantees
$(x,v) \mapsto (x,\exp_x v)$ acts as a smooth diffeomorphism near
$(x_0,v_0):=(x_0,\exp_x^{-1} y_0)$,  as does $(y,w) \mapsto (y,\exp_y w)$ near $(y_0,\exp^{-1} x_0)$ 
\cite[pp 314--328]{BeemEhrlichEasley96}.
From (d),  we deduce $\ell$ is differentiable near $(x_0,y_0)$ and its gradient
$$
-D \ell(x,y) = (\frac{v_*}{|v_*|_g},\frac {w_*}{|w_*|_g})\bigg|_{(v,w) = (\exp_x^{-1} y, \exp_y^{-1}x)}
$$
depends smoothly on $(x,y)$ there.  Thus $\ell(x,y)$ is smooth near $(x_0,y_0)$.

We also claim $\ell$ fails to be smooth at each $(x_0,y_0) \in \singl$.
If $\ell(x_0,y_0) = -\infty$ this is obvious since smooth functions are by definition real-valued.
If $\ell(x_0,y_0)=0$, we will show in (e) below that differentiability of $\ell$ fails unless $x_0=y_0$,
in which case $\ell(x,y)=+\infty$ for points arbitrarily close to $(x_0,y_0)$.
If $\ell(x_0,y_0)>0$ then we are on the timelike cut locus where
Theorem~\ref{T:lapse semiconcavity fails} shows $\ell$ fails to be $C^{1,1}$ smooth.

(e) Suppose $\ell(x,y)=0$.  
Let $X \subset M$ denote a convex neighbourhood of $x$, meaning for each $z \in X$,
the inverse map to $\exp_z:T_z M \longrightarrow M$ acts diffeomorphically on $X$, as in 
e.g.~
\cite[\S 5.7]{O'Neill83}.  Let $\sigma:[0,1]\longrightarrow M$ be a (null) action minimizing segment 
joining $x$ to $y$.  For $s>0$ sufficiently small that $z:= \sigma(s) \in X$ we find
\begin{eqnarray}
\ell(\exp_x v,z) &=& \ell(\exp_z \circ (\exp_z)^{-1} \circ \exp_x v, z)
\nonumber \\ &=& -L(-(\exp_z)^{-1} \circ \exp_x v, z;1)
\nonumber \\ &=& -L(\dot \sigma(s) - (D\exp_z^{-1})_{\dot \sigma(s)} (D\exp_x)_{\bf 0} v, z;1).
\label{Gauss Lemma}
\end{eqnarray}
Now if $\ell(\wdot,y)$ has a supergradient $w \in T_x^*M$, the reverse triangle inequality yields
\begin{eqnarray*}
\ell(\exp_x v,z) &\le& \ell(\exp_x v,y)- \ell(z,y)
\\ &\le& \ell(x,y)  -\ell(z,y)+ 
w[v] +o(|v|_{\tilde g})
\end{eqnarray*}
as $v \to 0$.  Since $\ell(x,y)=\ell(z,y)=L(\dot \sigma(s);1)=0$,
this would imply subdifferentiability of $L(\wdot,z;1)$ at $\dot \sigma(s)$ --- a contradiction
with Corollary \ref{C:1-Lagrangian} unless
$\dot \sigma(s)=0$, in which case $x=y=z$, both derivatives in \eqref{Gauss Lemma}
are given by the identity map, and $H(w;1)=0$ as desired. 
\end{proof}




\begin{corollary}[Twist and non-degeneracy]\label{C:twist}
Fix $0<q<1$. Then (i) $\ell^q$ inherits properties (a)-(d) of Theorem~\ref{T:lapse smoothness} from $\ell$.
(ii) If $\frac1q\ell(\wdot,y)^q$ has supergradient $w$ at $x$ then 
$y  = \exp_x DH(w,x;q)$, where  $H$ is defined at \eqref{q-Hamiltonian}.
(iii) If $(x,y) \not\in \singl$ then $\det \frac{\p^2}{\p x^j \p y^{\bar i}} \ell^q(x,y) \ne 0$.
(iv) If $(x,y) \in \{\ell>0\} \cap \singl$ then
\begin{eqnarray}\label{q not semiconcave}
\sup_{0<|w|_{\tilde g}<1} \frac{\ell^q(\exp^{\tilde g}_x w,y) + \ell^q(\exp^{\tilde g}_x -w,y) - 2 \ell^q(x,y)}{2|w|_{\tilde g}^2} = +\infty
\end{eqnarray}
\end{corollary}

\begin{proof}
(i) For a function $u:\R^n \longrightarrow[0,\infty)$ to have semiconvexity constant $C$ at $\xo$
is equivalent to asserting $p \in \p_\cdot u(\xo)$ non-empty  and
$$u(x) \ge u(\xo) + p[x-\xo] - \frac12 C|x-\xo|^2 +o(|x-\xo|^2)
$$ 
as $x \to \xo$,  for all $\xo$ near $\xo$.  Raising this inequality to exponent $q$,  for $|t| < 1$ the existence
of $|t_*| \le |t|$ such that
$$ (1+t)^q = 1 + qt + q(q-1)t^2/2 + q(q-1)(q-2)t_*^3/6$$
shows $u^q/q$ inherits semiconvexity constant 
$C u^{q-1} + 2(1-q) u^{q-2}|Du|^2$ at $\xo$ from $u$.
Applying this argument in Riemannian normal coordinates 
establishes semiconvexity of the locally Lipschitz function 
$u^q(\wdot) := \ell^q(\wdot,y)$ at each point $\xo$ with $\ell(\xo,y)>0$
in view of Theorem \ref{T:lapse smoothness}(d).  The remaining properties (a)-(d) 
follow from the one-sided chain rule \cite[Lemma 5]{McCann01} and
our convention $(-\infty)^{1/q}:= -\infty =: (-\infty)^q$.  We shall obtain a strengthening of (e) 
in the course of proving (ii) below:  namely,
that  $\ell(x,y)=0$ implies 
the superdifferential of $\frac1q \ell(\wdot,y)^q$ at $x$ is empty.

(iv)  The alternative to \eqref{q not semiconcave} is that $u(\wdot ) := \ell^q(\wdot,y)$ has semiconcavity
constant $C<\infty$ at some $\bar x$ with $(\bar x,y) \in \pl \cap \singl$.  The same argument as above
then implies $u^{1/q}$ has semiconcavity constant $C u^{\frac1q -1 } + 2(\frac1 q -1) u^{\frac1 q - 2} |Du|^2$
at $\bar x$,  contradicting Theorem \ref{T:lapse semiconcavity fails}.  So \eqref{q not semiconcave} must hold.

(ii) 
If $\frac1q \ell(\wdot,y)^q$ admits $w \in T_xM $ 
as a supergradient,   then $\ell(\wdot,y)$ admits $\ell(x,y)^{1-q}w$ as a supergradient at $x$,
by the (one-sided) chain rule.
When $\ell(x,y)=0$ this contradicts Theorem~\ref{T:lapse smoothness}(e) 
whether or not $x$ is distinct from $y$, since $H(\0;1)= +\infty \ne 0$.
Thus $\ell(x,y)>0$.  Now (c) implies differentiability of
$\frac1q\ell(\wdot,y)$ at $x$, with
\begin{equation}\label{H twist}
w:= D_x \ell^q(x,y)/q = 
|v|_g^{q-2} v_* \bigg|_{v=-\exp_x^{-1}y}.
\end{equation}
Thus  $y = \exp_x -|w|_g^{q'-2} w = \exp_x DH(w,x;q)$ is uniquely determined by $x$ and $w$, 
where $\frac1q + \frac1{q'}=1$.

(iii) Now fix $(x,y) \not\in \singl$. 
Differentiating \eqref{H twist} with respect to $y$ yields
$$
-D^2_{yx} \ell(x,y)/q = \frac{|v|^2g - (2-q) v_* \otimes v_*}{|v|_g^{4-q}}\bigg|_{v=\exp_x^{-1}y} D_y (\exp_x^{-1} y).
$$
Our choice $(+ -\ldots -)$ of signature for $g$ shows the first factor is negative definite when $q<1$
since $v$ is timelike; the second factor 
has non-zero determinant since $y$ is in the chronological future 
but not in the conjugate locus of $x$.
\end{proof}

\begin{definition}[Approximate differentiability]\label{D:approximate derivative}
A map $F:M \longrightarrow N$ between differentiable manifolds is {\em approximately differentiable}
at $x \in M$ if there exists a map $\tilde F:M \longrightarrow N$, differentiable at $x$,  such that
the set $\{ \tilde F \ne F\}$ has zero density at $x$, i.e.
\begin{equation}
\lim_{r \to 0} \frac{\vol[\{x \in B_r(x) \mid F(x) \ne \tilde F(x) \}]}{\vol[B_r(x)]} =0,
\end{equation}
where the radius $r$ and $vol$ refers to the coordinate radius and volume in any (and hence all) coordinate system(s) at $x$.
\end{definition}

A well-known result of Aleksandrov  \cite{Aleksandrov39} whose Riemannian version \cite{Bangert79} we shall exploit 
asserts that convex (and hence semiconvex) functions have 
approximate second derivatives almost everywhere.
In fact, more is true: any semiconvex function agrees with a $C^2$ function,  outside a set of arbitrarily
small volume.  We shall make use of this Lusin style approximation result,  which follows from the 
fact that convex gradients are countably Lipschitz, e.g.~\cite{AmbrosioFuscoPallara00} \cite{Santambrogio15},
and the analogous $C^1$ approximation result for Lipschitz functions, e.g. \S 6.6 of \cite{EvansGariepy92}.

\section{Kantorovich duality with Lorentz distance}
\label{S:dual}

To characterize the $q$-geodesics defined above,  we must first study the optimization \eqref{MK},
which is a Monge-Kantorovich optimal transportation problem. 
As an infinite-dimensional linear program,  it is well-known to have the following dual problem,
provided 
the infinum is finite as described e.g.\ in \cite{Villani09}:  
\begin{equation}\label{KM}
 \frac1q \LL_q(\mu,\nu)^q = \inf \left\{ \int_M u d\mu + \int_M v d\nu \mid {\ts \frac{1}q} \ell^q \le u \oplus v \in L^1(\mu \times \nu) \right\},
\end{equation}
where
$$
(u \oplus v) (x,y) := u(x) + v(y)
$$
and $u$ and $v$ are lower semicontinuous.
Moreover,  given any sets $X \supset \spt \mu$ and $Y\supset \spt \nu$,
we may restrict the infimum \eqref{KM} to pairs of functions $u=v_{q}$ on $X$ and $v=u_{\tilde q}$ on $Y$, where
\begin{eqnarray}\label{l-transform}
v_q(x) &:=& \sup_{y \in Y} {\ts \frac1q} \ell^q(x,y) - v(y)
\\ u_{\tilde q}(y) &:=& \sup_{x \in X} {\ts \frac1q} \ell^q(x,y) -u(x).
\end{eqnarray}
Such pairs of functions $(u,v)=(v_q,u_{\tilde q})$ are called {\em $\frac{\ell^q}{q}$-convex} or 
{\em  $\frac{\ell^q}{q}$-conjugate}.  Notice however,  that these notions depend on the choice of sets $X$ and $Y$.



Unfortunately,  since the function $\ell$ jumps to $-\infty$ outside the causal future,  
it is not clear whether the infimum \eqref{KM} is generally attained.  However, we shortly
show that it will be attained when the measures $\mu$ and $\nu$ satisfy the following condition. 


\begin{definition}[$q$-separated] \label{D:q-separated}
Fix $0<q\le 1$. 
We say $(\mu,\nu) \in \PP_c(M)^2$ is {\em $q$-separated} by $\pi \in \Pi(\mu,\nu)$
and lower semicontinuous $u:\spt \mu \longrightarrow \R \cup \{+\infty\}$ and
$v:\spt \nu \longrightarrow \R \cup \{+\infty\}$ if 
\begin{eqnarray*}
u(x) + v(y) \ge& \frac1q \ell (x,y)^q \phantom{>0}  &\forall (x,y) \in \spt [\mu \times \nu],
\end{eqnarray*}
$\spt \pi \subset S := \{(x,y) \in \spt[\mu \times \nu] \mid u(x)+v(y) = \frac1q \ell(x,y)^q\}$ and $S \subset \pl$.
%
\end{definition}

\begin{remark} 
Compactness of $S$ and its 
disjointness from $\npl$ are essential to this definition: 
$\ell> 0$ on $S$ shows 
events described by $d\mu(x)$ and $d\nu(y)$ can be matched so that each $x$ lies in the chronological
--- as opposed to the causal --- past of its assigned $y$.
One can also strengthen the definition of $q$-separation by requiring disjointness 
of $S$ from $\singl$.  This leads to a simpler variant of the theory,  but one which
is unable to rule out transportation to the timelike cut locus.  This limitation is unsatisfactory
in the context of our intended application of these developments to nonsmooth spacetimes.  
\end{remark}


Although mildly restrictive,  this definition has the following theorem
as an immediate consequence, which allows us to circumvent various subtleties
involving non-compact support and/or  
null geodesics.  In Section \S\ref{S:relaxing separation}
we 
relax this restriction by approximation. 
Of course, it may turn out that the dual is actually attained in this greater generality,
as in the Riemannian case \cite{FathiFigalli10}.
Ideas of Bertrand, Pratelli, Puel \cite{BertrandPuel13} \cite{BertrandPratelliPuel18} and Suhr \cite{Suhr16p}
may prove relevant to this question, as may my own work with Puel \cite{McCannPuel09}.

A set $S \subset M \times M$ is called {\em $\ell^q$-cyclically monotone},
if for each $k \in \N$,  permutation 
$\sigma$ on $k$ letters,  for each 
$(x_1,y_1), \ldots, (x_k,y_k) \in S$ we find
\begin{equation}\label{cyclical monotonicity}
\sum_{i=1}^k \ell(x_i,y_i)^q 
\ge \sum_{i=1}^k  \ell(x_i,y_{\sigma(i)})^q.
\end{equation} 
This property is well-known to characterize the support of $\ell^q$-optimal measures $\pi$
for \eqref{MK}, provided the set where the cost is not finite is 
closed \cite{BeiglbockGoldsternMareschSchachermayer09};
{\em $\ell^q$-monotonicity} refers to the case
$k=2$ in \eqref{cyclical monotonicity}. 

\begin{theorem}[Duality by $q$-separation]
\label{T:duality by fiat}
Fix $0 < q \le 1$.
If $(\mu,\nu) \in \PP_c(M)^2$ is $q$-separated by $(\pi,u,v)$,
then 
(i) $(u,v) = (v_q,u_{\tilde q})$ on $X\times Y := \spt [\mu \times \nu]$ in \eqref{l-transform}.
(ii) The set $S = \{ (x,y) \in X \times Y \mid u \oplus v = \frac1q \ell^q\}$ is compact and
 $\ell^q$-cyclically monotone.
(iii) The potentials $(u,v)$ minimize \eqref{KM}, while $\pi$ maximizes \eqref{MK}.
(iv) The extensions $u:=v_q$ and $v:=u_{\tilde q}$ are semiconvex Lipschitz functions on neighbourhoods of
 $X$ and $Y$, respectively, with Lipschitz and semiconvexity constants estimated by those of 
$\frac1q \ell^q$ on $S$.
\end{theorem}

\begin{proof}
(i) Let $(\mu,\nu) \in \PP_c(M)^2$ be $q$-separated by $(\pi,u,v)$. For $\xo \in M$
the supremum
\begin{equation}\label{S-dual-loc}
v_q(\xo) := \sup_{y \in Y} {\ts \frac1q} \ell^q(\xo,y) - v(y)
\end{equation}
is attained,  due to the compactness of $Y:= \spt \nu$ and the upper semicontinuity assumed
for $-v$ and established for
of $\ell^q$  in Corollary \ref{C:twist}. 
If $\xo \in X$ then $u(\xo) \ge v_q(\xo)$ follows from the definition of $q$-separation.
In this case the existence of $\yo \in Y$ such that
$(\xo,\yo) \in \spt \pi \subset S:= \{(x,y) \in X \times Y \mid  u(x) + v(y) = \frac1q \ell(x,y)^q\}$
follows from compactness of $X:= \spt \mu$ and $Y$.
This $\yo$ must then maximize \eqref{S-dual-loc}, and $S \subset \pl$ shows $u(\xo)=v_q(\xo)$ 
to be finite on $X$;  since $u$ 
was not defined
outside $X$ we may take $u:=v_q$ as a definition there.
The identity $v=u_{\tilde q}$ 
is proved similarly.

(ii) Compactness of $X,Y$ and the lower semicontinuity of $u\oplus v - \frac1q \ell^q$ 
asserted by Corollary \ref{C:twist} 
show that $S$ is compact.
Choosing any $k \in \N$ and a permutation $\sigma$ on $k$ letters,  for each 
$(x_1,y_1), \ldots, (x_k,y_k) \in S$ we find
\begin{eqnarray*}
 \frac1q \sum_{i=1}^k \ell(x_i,y_i)^q 
&=& \sum_{i=1}^k u(x_i) + v(y_i)
\\ &=& \sum_{i=1}^k u(x_i) + v(y_{\sigma(i)})
\\ &\ge&  \frac1q \sum_{i=1}^k  \ell(x_i,y_{\sigma(i)})^q
\end{eqnarray*}
as desired.

(iii) Since $S$ is compact, 
Corollary \ref{C:twist} shows $\ell$ is bounded above on $S$.
Being lower semicontinuous, $u$ and $v$ are bounded below on $S$.  Because $\pi$ vanishes outside
$S$,  
\begin{eqnarray*}
\frac1q \int_{M^2} \ell(x,y)^q d\pi(x,y) 
&=& \int_{M^2} [u(x)+ v(y)] d\pi(x,y)
\\ &=& \int_M u d\mu + \int_M v d\nu
\end{eqnarray*}
where the second equality follows from the fact that $\pi\in \Pi(\mu,\nu)$ 
has $\mu$ and $\nu$ for its left and right
marginals.  Observing that the inequality $\frac1q\LL_q(\mu,\nu)^q \le \inf$ of \eqref{KM} is elementary to derive
makes it clear that $(u,v)$ attain the infimum and $\pi$ attains the maximum \eqref{MK}.

(iv) Compactness of $S$ shows
its (Riemannian) distance $3R:=d_{\tilde g \oplus \tilde g}(\npl,S)$ from $\npl$ is positive.
Given $r \ge 0$,
let $X_r := \{ x \in M \mid d_{\tilde g}(x,X) \le r \}$ denote the set of points whose Riemannian distance
from $X$ is at most $r$.
Define $Y_r \subset M$ and $S_r \subset M^2$ analogously.
According to Theorem~\ref{T:lapse smoothness}
the restriction of $\ell^q/q$ to the Riemannian neighbourhour $S_{2R}$ of size $2R$ around $S$ is Lipschitz and
has semiconvexity constant $C_{2R} > -\infty$.

We claim the (Riemannian) Lipschitz constant of $v_q$ on some sufficiently small neighbourhood $X_r$ 
of $X=\spt \mu$ is no greater than $\|\ell\|_{C^{0,1}(S_{2R})}$.  However, let us first establish lower semicontinuity 
\eqref{ulsc} of $v_q$
at each point $\xo \in X$.  Fixing $\xo \in X$,  there exists $\yo$ with $(\xo,\yo) \in S$ as above.
Letting $x \in B_R(\xo)$, we find 
\begin{eqnarray*}
v_q(x) &\ge& \ell (x,\yo;q) - v(\yo)
\\ &\ge& \ell(\xo,\yo;q) - \|\ell(\cdot, \yo;q)\|_{C^{0,1}(B_R(\xo))} d_{\tilde g}(x,\xo) - v(\yo)
\\ &\ge& v_q(\xo) - \frac1q\|\ell^q\|_{C^{0,1}(S_R)} d_{\tilde g}(x,\xo).
\end{eqnarray*}
If $x \in X$ we can interchange $\xo$ with $x$ to obtain the desired Lipschitz bound for $v_q$ on $X$
(and for $u_{\tilde q}$ on $Y$), though not yet on $X_r$ (or $Y_r$, respectively).
However, even for $x \in X_r \setminus X$ we deduce the desired lower semicontinuity:
\begin{equation}\label{ulsc}
\lim\inf_{x \to \xo} v_q(x) \ge v_q(\xo).
\end{equation}

Taking $r>0$ sufficiently small ensures that 
$X_r$ inherits compactness from $X$ (and $Y_r$ 
 from $Y$).
Taking $r>0$ smaller still ensures
$S_{(r,0)} := \{(x,y) \in X_r \times Y \mid v_q(x) + v(y) = \ell^q(x,y)/q \}$ is contained in $S_{R}$.
If not,  there exists a sequence $(x_k,y_k) \in S_{(1/k,0)} \setminus S_R$ with convergent subsequence.
Lower semicontinuity of $v$ on $Y$,  Theorem~\ref{T:lapse smoothness} and \eqref{ulsc} imply 
the limit $(x_\infty,y_\infty)$ belongs to $S$.  But this contradicts $(x_k,y_k) \not\in S_R$.

Now apply the preceding argument to an arbitrary pair of points $\xo \in X_r$ and $x \in B_R(\xo) \cap X_r$.
As before there exists $\yo$ with $(\xo,\yo) \in S_{(r,0)} \subset S_R$, and
\begin{eqnarray*}
v_q(x) &\ge& \ell (x,\yo;q) - v(\yo)
\\ &\ge& \ell(\xo,\yo;q) - \|\ell(\cdot, \cdot;q)\|_{C^1(B_R(\xo,\yo))} d_{\tilde g}(x,\xo) - v(\yo)
\\ &=& v_q(\xo) - \frac1q\|\ell^q\|_{C^1(S_{2R})} d_{\tilde g}(x,\xo);
\end{eqnarray*}
interchanging $\xo$ with $x$ yields the desired Lipschitz bound for $v_q$ on $X_r$.

Finally,  for $|w|_{\tilde g} <R$ and $(\xo,\yo)\in S_{(r,0)} \subset S_R$
as above, the Riemannian exponential yields
\begin{eqnarray*}
&& v_q(\exp_{\xo}^{\tilde g} w) + v_q(\exp_{\xo}^{\tilde g} -w) - 2v_q(\xo)
\\&\ge& 
\ell(\exp_{\xo}^{\tilde g} w,\yo;q) + \ell(\exp_{\xo}^{\tilde g} -w,\yo;q) - 2\ell(\xo,\yo;q)
\\ &\ge& 
C_{2R} |w|^2_{\tilde g}.
\end{eqnarray*}
This shows the semiconvexity of $v_q$ on $X_r$. 
 Similarly,
taking $r>0$ small enough yields $u_{\tilde q}$ semiconvex and Lipschitz on $Y_r$.
\end{proof}

The following lemma shows the notion of $q$-separation is not vacuous;
instead it puts us back into the more standard framework of optimal transportation
with respect to uniformly continuous cost functions. 

\begin{lemma}[Existence of $q$-separation]\label{L:well-separated}
Fix $0<q \le 1$ and $\mu,\nu \in \PP_c(M)$.  If $\spt[\mu \times \nu] \subset  M \times M \setminus \npl$
then $(\mu,\nu)$ is $q$-separated.  
\end{lemma}

\begin{proof}
Theorem~\ref{T:lapse smoothness} asserts continuity of $\ell$ on the compact set 
$X \times Y := \spt[\mu \times \nu]$.
In this case, the supremum \eqref{MK} and infimum \eqref{KM} are well-known to be attained by 
$\pi \in \Pi(\mu,\nu)$ and a pair of uniformly continuous
functions $(u,v)=(v_q,u_{\tilde q})$ satisfying \eqref{l-transform}, e.g. \cite{Santambrogio15} \cite{Villani09}.
Duality \eqref{KM} implies $\spt \pi$ is contained in the equality set $S \subset X \times Y$ for 
$u \oplus v - \frac1q \ell^q \ge 0$
--- which inherits both compactness and disjointness to $\singl$ from $X \times Y$.  
%
%
%
%
\end{proof}

\section{Characterizing $q$-geodesics via duality}
\label{S:map}

Armed with a duality theory for $q$-separated probability measures $(\mu_0,\mu_1)$, 
we now turn to the analytical characteristics of the $q$-geodesic 
$s \in [0,1] \mapsto \mu_s \in \PP_c(M)$ which links them.
These can in principle be described using either an Eulerian \cite{BenamouBrenier00} 
\cite{OttoVillani00} \cite{Brenier03} \cite{AmbrosioGigliSavare05} \cite{Schachter17} 
or Lagrangian framework \cite{Villani09}.
Here we employ a Lagrangian approach consistent with the analogous results originally obtained
in Euclidean space 
\cite{McCann97} and, with Cordero-Erausquin and Schmuckenschl\"ager, 
on Riemannian manifolds
\cite{McCann01} 
\cite{CorderoMcCannSchmuckenschlager01} \cite{CorderoMcCannSchmuckenschlager06}.
For the case $q=1$ not covered here, Suhr develops a different approach, 
based on dynamical transport plans 
(i.e. measures on the space of geodesic segments)~\cite{Suhr16p}.

Take $M$ to be a globally hyperbolic spacetime, $N=M \times M$ and $0<q <1$.
When $\mu_0,\mu_1 \in \PP_c(M)$ are $q$-separated by $(\pi,u,v)$ and $\mu_0 \in \Pac(M)$,  one task will be
to show $\pi=(id \times F_1)_{\#}\mu_0$ where $F_s(x) := \exp_x(s DH(Du(x),x;q))$ for each $s \in[0,1]$
and $id:M \longrightarrow M$ denotes the identity map $id(x)=x$ on $M$.  In other words, the $\ell^q$-optimal coupling
$\pi \in \Pi(\mu_0,\mu_1)$ from Theorem \ref{T:duality by fiat}
concentrates its mass on the graph of a map $F_1:M\longrightarrow M$.
By analogy with the Euclidean case \cite{Monge81}, 
such a map is said  
to solve Monge's problem
\cite{Santambrogio15} \cite{Villani09}.
This is accomplished in Theorem \ref{T:map}, which also characterizes the unique coupling
achieving the maximum \eqref{MK}, and is the analog of Brenier's theorem from the 
Euclidean setting \cite{Brenier91} \cite{McCann95} and my Riemannian generalization \cite{McCann01}.
Its corollaries go further by showing 
$\mu_s := F_{s \#} \mu_0$ is the unique $q$-geodesic with the given endpoints,
that it is absolutely continuous with respect to $\vg$ and that its density $\rho_s := d\mu_s/d\vg$,
is related to the Jacobian $JF_s(x)$ of $F_s$ by the Monge-Amp\`ere type equation
\begin{equation}\label{MAs}
\rho_{0}(x) = \rho_s (F_s(x)) JF_s(x) \qquad \mbox{\rm $\mu_0$-a.e.},
\end{equation}
whenever $s<1$ or $\mu_1 \in \Pac(M)$. 
In particular,  the Lagrangian path description of the worldlines of the individual events 
making up this geodesic is given by the map $(x,s) \in M \times [0,1] \mapsto F_s(x)$.

To achieve this description we will need to establish various analytical properties of $F_s$ along the way,
such as the fact that $F_s^{-1}$ is (Lipschitz) continuous for each $s<1$ in 
Theorem \ref{T:Lipschitz inverse maps}.  Similarly, the fact that $F_s$ is countably Lipschitz
follows from Theorem \ref{T:duality by fiat}(iv), allowing us to make sense of its Jacobian
$JF_s(x):= |\det \tilde DF_s(x)|$ almost everywhere.
The Monge-Amp\`ere type equation \eqref{MAs}
will be the key to analyzing convexity properties of the Boltzmann-Shannon or
relative entropy $\ei(s) := \EV(\mu_s)$ along the $q$-geodesic in question, 
so we will need to be able to compute two derivatives of \eqref{MAs} with respect to $s$.
Fortunately,  $s \in [0,1]\mapsto F_s(x)$ is a proper-time maximizing segment for each $x \in \dom Du$,
so the derivatives desired can be computed using Jacobi fields.  This is accomplished in 
Lemma~\ref{L:Jacobi}, where we see the first  Lorentzian connection linking optimal transport
to semi-Riemannian curvature.

We
 begin by showing that if an $\ell^q$-optimal measure $\pi$ couples two distinct pairs of events,
i.e.~$(x,y),(x',y') \in \spt \pi$, then the midpoint $\bar z_{1/2}(x,y)$
of a proper-time maximizing geodesic segment joining $x$ to $y$ cannot coincide with
the midpoint of any proper-time maximizing geodesic segment joining $x'$ to $y'$.  Similarly
$\bar z_s(x,y) \ne \bar z_s(x',y')$ for $0<s<1$.  Such pairs of coupled events satisfy \eqref{2-monotone N}
by Theorem \ref{T:duality by fiat}(ii).

\begin{proposition}[Lagrangian trajectories don't cross]\label{P:trajectories don't cross}
Fix $q,s \in (0,1)$.  
 If $Z_s(x,y)$ from \eqref{Z_s(x,y)} intersects $Z_s(x',y')$ yet
\begin{equation}\label{2-monotone N}
\ell(x,y')^q + \ell(x',y)^q \le \ell(x,y)^q + \ell(x',y')^q
\end{equation}
then $(x,y)=(x',y')$.
\end{proposition}

\begin{proof}
This argument is  inspired by the Riemannian case \cite{CorderoMcCannSchmuckenschlager01}.
The reverse triangle inequality and strict concavity of $r \mapsto r^q$ assert
\begin{eqnarray*}
\ell(x,y')^q &\ge& \bigg( s \frac{\ell(x,m)}{s} + (1-s) \frac{\ell(m,y')}{1-s} \bigg)^q
\\ &\ge& s^{1-q} \ell(x,m)^q + (1-s)^{1-q} \ell(m,y')^q 
\\ &=& s \ell(x,y)^q + (1-s) \ell(x',y')^q.
\end{eqnarray*}
The first inequality is strict unless $m$ lies on an action minimizing segment joining $x$ to $y'$ ---
or equivalently $y'$ lies beyond $m$ on the unique future directed geodesic from $x$ passing through $m$;
the second inequality is strict unless $\frac{\ell(x,m)}{s} = \frac{\ell(m,y')}{1-s}$, or equivalently
$\ell(x,y) = \ell(x',y')$.
Similarly,
\begin{eqnarray*}
\ell(x',y)^q &\ge& \bigg(s \frac{\ell(x',m)}{s} + (1-s) \frac{\ell(m,y)}{1-s}\bigg)^q
\\ &\ge & s \ell(x',y')^q + (1-s) \ell(x,y)^q,
\end{eqnarray*}
and at least one of these two inequalities is strict unless 
$\ell(x,y) = \ell(x',y')$
and $y$ lies beyond $m$ on the extension of the geodesic from $x'$ through $m$.

Summing these contradicts \eqref{2-monotone N} unless equalities hold throughout.
But this forces $\ell(x,y)=\ell(x',y')$ and all five points $x,x',m,y',y$ onto the same timelike geodesic, with
$x$ and $x'$ in the past of $m$ and $y$ and $y'$ in its future.  Since the segments $xy$ and $x'y'$ of this
geodesic have
the same proper time and $m$ divides them both in the same ratio,  we conclude $x=x'$ and $y=y'$.
as desired.
\end{proof}

\begin{corollary}[Continuous inverse maps]\label{C:continuous inverse maps}
Fix $q,s \in (0,1)$.  If $(\mu_0,\mu_1) \in \PP_c(M)^2$ is $q$-separated and $X_i:=\spt \mu_i$, 
there is a continuous map
$W:\dom W \subset M \longrightarrow S \subset X_0 \times X_1$ 
such that if $\mu_s$ lies on a $q$-geodesic \eqref{q-geodesic}
then $W_\#\mu_s$ maximizes $\ell^q$ in $\Pi(\mu_0,\mu_1)$.
Here $\dom W = Z_s(S)$ where 
$Z_s$ is from \eqref{Z_s(X,Y)} and
$S$ 
from the Definition \ref{D:q-separated} of $q$-separated.
Moreover, $\bar z_s \circ W$ acts as the identity map on 
$\bar z_s(S)$ whenever $\bar z_s$ 
is consistent with Lemma \ref{L:midpoint selection}. 
\end{corollary}


\begin{proof}
Duality \eqref{KM} holds with continuous semiconvex optimizers $(u,v)=(v_q,u_{\td q})$ 
according to Theorem \ref{T:duality by fiat}, 
which also shows the set 
$S=\{(x,y) \in 
X_0 \times X_1 \mid u(x) + v(y) = \ell(x,y; q) \}$ 
to be compact and $\ell^q$-cyclically monotone.
Recall $q$-separation requires $S$ to be disjoint from $\npl$.
Let $m_k \in Z_s(x_k,y_k)$ for some sequence $(x_k,y_k) \in S$.
where $Z_s(x,y)$ is from \eqref{Z_s(x,y)}.
Assume $m_k \to m$, and
extract a subsequential limit $(x_{k(j)},y_{k(j)}) \to (x,y)$ using compactness of $S$.
Then $m \in Z_s(x,y)$.  Similarly,  if another subsequence of $(x_k,y_k)$ converges to a different limit $(x',y')\in S$
then $m\in Z_s(x',y')$.  Thus $Z_s(x,y)$ intersects $Z_s(x',y')$.  The $\ell^q$-cylical monotonicity of $S$ 
implies \eqref{2-monotone N},
 which forces $(x,y)=(x',y')$ according to Proposition~\ref{P:trajectories don't cross}. This means  
$W(m):=(x,y) \in X_0 \times X_1$ is well-defined and continuous, since all subsequences of $(x_k,y_k)=W(m_k)$ 
converge to the same limit $(x,y)=W(m)$.  Moreover, $W$ acts as a right-inverse for $\bar z_s$ on $\bar z_s(S)$.



Now let $\mu_s$ satisfy \eqref{q-geodesic}.  
Since $q$-separation yields $0<\LL_q(\mu_0,\mu_1)<\infty$,
Proposition \ref{P:RTIE} provides $\omega \in \PP(M^3)$ with marginals $(\mu_0,\mu_s,\mu_1)$ 
whose projection onto any pair of coordinates is $\ell^q$-optimal and has
$z \in Z_s(x,y)$ for $\omega$-a.e.~$(x,z,y)$.  In particular $\pi = \proj_{13\#}\omega$ maximizes $\ell^q$
on $\Pi(\mu_0,\mu_1)$, hence is supported in the compact set $S$ according to Theorem 
\ref{T:duality by fiat} and the duality \eqref{KM}.  It follows that $\mu_s =\proj_{2}\omega$ 
vanishes outside $\dom W := Z_s(S)$,
which is compact according to Lemma \ref{L:compact support}.
Denoting $(X(m),Y(m)):=W(m)$, the preceding paragraph shows $\omega$ to vanish outside
the graph of $W$.  Thus $\omega = (X \times id \times Y)_\#\mu_s$ by e.g.~Lemma 3.1 of \cite{AhmadKimMcCann11},
 hence $\pi = W_\#\mu_s$ as desired. 
\end{proof}

The continuous map $W$ of the preceding corollary is actually Lipschitz:

\begin{theorem}[Lipschitz inverse maps]\label{T:Lipschitz inverse maps}
Under the hypotheses of Corollary \ref{C:continuous inverse maps},
the map $W:Z_s(S) \longrightarrow M^2$ defined in that corollary is Lipschitz
continuous with respect to any fixed choice of Riemannian distance $d_{\tilde g}$ on $M$.
\end{theorem}
To avoid interrupting the flow of ideas, we defer the discussion and rather technical proof of Theorem \ref{T:Lipschitz inverse maps} to Appendix \ref{S:Monge-Mather proof}.

\begin{lemma}[Variational characterization of 
geodesic endpoints]\label{L:q-star shaped}
Fix $0 < q <1$ and a timelike proper-time maximizing segment $s \in [0,1] \mapsto x_s \in M$.
For each $0<s<1$ and $x\in M$,
\begin{equation}\label{q-star shaped}
\ell(x,x_1)^q \ge s^{1-q} \ell(x,x_s)^q + (1-s)^{1-q} \ell(x_s,x_1)^q
\end{equation}
with equality if and only if $x=x_0$.
\end{lemma}

\begin{proof}
The reverse triangle inequality  yields
\begin{eqnarray*}
\ell(x,x_1) \ge s \frac{\ell(x,x_s)}{s} + (1-s)\frac{\ell(x_s,1)}{1-s},
\end{eqnarray*}
with equality only if $x_s$ lies on the minimizing segment joining $x$ to $x_1$.
In other words,  equality holds only if $x$ lies beyond $x_s$ on the past-directed geodesic 
from $x_1$ through $x_s$ (this geodesic is unique since $x_s$ is internal to the minimizing segment 
joining $x_0$ to $x_1$).
Strict concavity of the function $r \mapsto r^q$ yields \eqref{q-star shaped}, with equality forcing
$$
\frac{\ell(x,x_s)}{s} = \frac{\ell(x_s,1)}{1-s}.
$$
This equation is uniquely solved on the geodesic in question by $x=x_0$.
\end{proof}


The following proposition shows our $q$-separation property propogates from
the endpoints to the interior of a $q$-geodesic.

\begin{proposition}[Star-shapedness of $q$-separation] 
Fix $0<q<1$. If $s \in [0,1] \mapsto \mu_s \in \PP_c(M)$ is a $q$-geodesic and
$(\mu_0,\mu_1)$ is $q$-separated,  then $(\mu_s,\mu_t)$ is $q$-separated for all $0 \le s < t \le 1$.
\end{proposition}

\begin{proof}
Let $(\mu_0,\mu_1)$ be $q$-separated by $\pi$ and $(u,v)$.  
Since every subsegment of a $q$-geodesic is itself a $q$-geodesic (after affine reparameterization),
it suffices to prove $(\mu_0,\mu_s)$ and $(\mu_t,\mu_1)$ are $q$-separated.  We show this for $(\mu_0,\mu_s)$;
the proof for $(\mu_t,\mu_1)$ is similar.


Setting $X_s := \spt \mu_s$, 
Theorem \ref{T:duality by fiat} asserts that $u$ and $v$ are continuous on $X_0$ and $X_1$ respectively,
$S := \{(x,y) \in X_0 \times X_1 
\mid u(x)+v(y) = \ell(x,y; q)\}$ is compact 
and 
$$
u(m) = \max_{(x,y) \in S} \ell(m,y; q) - v(y) \qquad \mbox{\rm for all}\ m \in X_0,
$$
where we note that $q$-separation implies the projections of $S \subset M \times M$ onto the first and second copies
of $M$ cover $X_0$ and $X_1$, respectively.
Moreover, for fixed $m \in X_0$ the supremum is attained at
$(x,y) =(m,y) \in S$.
Lemma~\ref{L:q-star shaped} implies
$$
s^{q-1} u(m) = \max_{(x,y) \in S, z\in Z_s(x,y)} \ell(m,z; q) + (s^{-1}-1)^{1-q} \ell(z,y;q)- s^{q-1}v(y)
$$
and that the maximum is attained at some $(x,y) =(m,y) \in S$ and  each $z \in Z_s(m,y)$. 
According to Corollary \ref{C:continuous inverse maps}, there is a continuous map
$W:Z_s(S) \longrightarrow S\subset M \times M$ 
for which $z \in Z_s(x,y)$ with $(x,y) \in S$ implies $(x,y)=(U_s(z),V_s(z)):=W(z)$.
Thus
\begin{equation}\label{q-star more}
s^{q-1} u(m) = 
\max_{z \in Z_s(S)} \ell(m,z; q) + (s^{-1}-1)^{1-q} \ell(z,V_s(z);q)-s^{q-1} v(V_s(z))
\end{equation}
and the maximum is attained at some $z$ satisfying  $U_s(z)=m$.

We claim  $(\mu_0,\mu_s)$ is $q$-separated by $\bar \pi = (U_s \times id)_\#\mu_s$, and
\begin{equation}\label{calibrating potentials}
(\bar u, \bar v) = s^{q-1}(u, v \circ V_s - (1-s)^{1-q} \frac{\ell^q}{q} \circ (id \times V_s)).
\end{equation}
Since $z \in Z_s(S)$ lies on a geodesic segment whose endpoints $W(z) \in S$ are chronologically separated,
$0<s<1$ implies $(z,V_s(z)) \not \in \singl$; thus
$\bar u\in C(X_0)$ and $\bar v \in C(Z_s((S)))$ inherit continuity from that of $(u,v), V_s$
and that of $\ell$ outside $\singl$.  Since Lemma \ref{L:compact support}  and Proposition \ref{P:RTIE} 
imply $X_s \subset Z_s(S)$,
compactness of 
$\bar S := \{(x,z) \in \spt[\mu_0 \times \mu_s] \mid \bar u(x)+\bar v(z) = \ell(x,z; q)\}$ 
follows from that
of $\spt [\mu_0 \times \mu_s]$ and the upper semicontinuity of $\ell$ shown in 
Corollary~\ref{C:twist}. Moreover, $\ell \ge 0$ on $\bar S$. 
Our identification of the maximizers in \eqref{q-star more} shows $\spt \bar \pi \subset \bar S$,
but we must still establish $\ell \ne 0$ on $\bar S$.

Given $(x,z) \in \bar S$, the identification above asserts $x=U_s(z)$.
Moreover,  $z \in Z_s(x,y)$ for $y=V_s(z)$.  Since Corollary \ref{C:continuous inverse maps} 
also asserts $\pi := (U_s \times V_s)_\# \mu_s$ maximizes $\ell^q$ on $\Pi(\mu_0,\mu_1)$,  we find 
$(x,y) \in S$ and furthermore, $\bar \pi \in \Pi(\mu_0,\mu_s)$.  
The disjointness of $S$ from $\npl$ guaranteed by $q$-separation implies $y$ is in the chronological future of $x$.
Since $z$ lies on the timelike geodesic segment joining $x$ to $y$,  this shows $\ell(x,z)>0$ as well.
Thus $\bar S \subset \pl$ 
to conclude the proof.
\end{proof}


\begin{remark}[Hopf-Lax / Hamilton-Jacobi semigroup] 
By symmetry,  the potentials which $q$-separate $\mu_t$ from $\mu_1$ are given by
$$ 
(\bar u, \bar v) = (1-t)^{q-1}(u \circ U_t - t^{1-q} \frac{\ell^q}{q} \circ (U_t \times id),v).
$$ 
instead of \eqref{calibrating potentials}.
Apart from an overall change of sign, $\bar u$ should be compared with the Hopf-Lax solution
$$
\tilde u(z,t) = \inf_{\sigma \in C^{1}([0,t];M) \atop \sigma(t) =z} - u(\sigma(0)) + \int_0^t L(\dot \sigma(s),\sigma(s);q) ds
$$
to the Hamilton-Jacobi semigroup  \cite{Villani09}
$$
\frac{\p \tilde u}{\p t} + H(D\tilde u;q) =0
$$
associated with Hamiltonian $H$ from \eqref{q-Hamiltonian}.
\end{remark}



\begin{lemma}[Maps and their Jacobian derivatives]
\label{L:Jacobi}
Fix $X, Y \subset M$ compact, $0<q<1$ 
and $u$ semiconvex and Lipschitz with $u \ge u_{\td qq}$ 
in a neighbourhood of $X$.

(i)
If $u \oplus u_{\tilde q} - \frac1q \ell^q \ge 0$ vanishes at $(\xo,\yo) \in X \times Y$ then
$\xo \in \dom Du$ implies 
$\yo= F_1(\xo)$ where $F_s(x):=\exp_x  sDH(Du(x),x;q)$ while
$\xo \in \dom \tilde D^2 u$ implies $(\xo,\yo) \not\in \singl$.
Similarly, 
$\yo \in \dom Du_{\td q}$ gives 
$\xo= \exp_{\yo} -DH(-Du_{\td q}(\yo),\yo;q)$
while  $\yo \in \dom \tilde D^2 u_{\tilde q}$ gives $(\xo,\yo) \not\in \singl$.

(ii)
For $\vg$-a.e. $x \in X$, 
the approximate derivative $\td DF_s(x):T_x M \longrightarrow T_{F_s(x)}M$ 
from Definition \ref{D:approximate derivative} exists,  depends smoothly on $s$, and
$\td DF_s(x)w$ gives a Jacobi field along the geodesic $s\in[0,1] \mapsto F_s(x)$ for each $w \in T_xM$.

(iii)
Moreover, 
\begin{equation}\label{Hu}
\frac{\p }{\p s}\Big|_{s=0}\td DF_s =\td D \frac{\p F_s }{\p s}\Big|_{s=0}=  (D^2H \circ Du) \td D^2 u 
\end{equation}
holds $\vg$-a.e. on $X$,
where the derivatives are computed with respect to the Lorentzian connection,
(c.f. \eqref{Hu in coords}), where $H$ is from \eqref{q-Hamiltonian}
and we use $\td D^2u$ to denote the approximate Hessian of $u$.

\end{lemma}




\begin{proof}
(i)
Observe $u_{\td qq} \oplus u_{\td q} -\frac1q \ell^q \ge 0$ holds on $M \times Y$,
thus $u \oplus u_{\td q} - \frac1q \ell^q \ge 0$ on $U \times Y$ where $U$ is the hypothesized
neighbourhood of $X$ on which $u$ is Lipschitz and semiconvex.
If the latter inequality is saturated at $(\xo,\yo) \in X \times Y$ then
$u_{\td qq} \oplus u_{\td q} -\frac1q \ell^q$ has zero as a subgradient at $(\xo,\yo)$.
If $\xo \in \dom Du$, it follows that $\ell^q(\wdot,\yo)$ is superdifferentiable at $\xo$ with supergradient 
$Du(\xo)$,
whence Corollary \ref{C:twist}(ii) implies 
$\yo = \exp_{\xo} DH(Du(\xo), \xo; q)$ as desired.
If, in addition, $\xo \in \dom \tilde D^2 u$ then the second-order Taylor expansion
for $u(x)$ around $\xo$ provides a quadratic upper-bound for $\frac1q \ell^q(x,\yo) - u_{\tilde q}(\yo)$
at $\xo$. This rules out $(\xo,\yo) \in \singl$ according to (iv) of the same corollary. 
Since $u \oplus u_{\td q} - \frac1q \ell^q \ge 0$ holds on $X \times M$, when equality holds
at $(\xo,\yo) \in X \times Y$ with $\yo \in \dom Du_{\td q}$ it follows similarly that 
$\exp_{\yo}^{-1} \xo = -DH(-Du_{\td q}(\yo), \yo; q)$ and --- when $\yo \in \dom \tilde D^2 u_{\td q}$ 
--- that $(\xo,\yo) \not\in \singl$.

(ii)
For every $\epsilon>0$,  
semiconvexity implies that outside of a set of volume $\epsilon$ in $U \supset X$,  $Du$ agrees with a 
continuously differentiable vector field $V_\epsilon$ on $M$;  moreover, its approximate second 
derivative agrees with $D V_\epsilon$ outside of this small set.  Thus
$F^\epsilon_s(x) := \exp_x s DH(V_\epsilon(x),x;q)$ is $C^1$ in $x$ and smooth in $s$,
and its mixed partial derivatives are continuous and equal: $\frac{\p}{\p s} D F^\epsilon_s = 
D \frac{\p}{\p s} F^\epsilon_s(x)$
where $D$ denotes derivative with respect to $x$.  Given $(w,x(0)) \in TM$, let $r \in[-1,1] \mapsto x(r) \in M$ be a 
$C^1$ curve through $x(0)$ with tangent vector $\dot x(0) = w$.  Then $r \in [-1,1] \mapsto F^\epsilon_s(x(r))$ is a $C^1$ geodesic variation since $s \in [0,1] \mapsto F^\epsilon_s(x(r))$ is a geodesic segment for each $r \in [-1,1]$.  Thus
$\frac{\p}{\p r}\Big|_{r=0} F^\epsilon_s(x(r)) = DF^\epsilon_s(x(0))w$ is a Jacobi field (by e.g.~Lemma 8.3 of \cite{O'Neill83}).  Since the approximate derivative $\td DF_s(x)$ agrees with $DF^\epsilon_s(x)$ outside of a 
set of volume $\epsilon$,  and $\epsilon>0$ is arbitrary,  we find $\td DF_s(x(0))w$ to depend smoothly on $s$ 
and be a Jacobi field for $x(0) \in U$ in a subset of full volume.

(iii)
Differentiating the vector field $\frac{\p F^\epsilon_s(x)}{\p s}\Big|_{s=0}= DH(Du^\epsilon(x),x;q)$ 
using the Lorentzian connection yields 
\begin{equation}\label{Hu in coords}
D_k \frac{\p }{\p s}\Big|_{s=0} F^\epsilon_s(x)^i=H^{ij} {u^\epsilon}_{jk}
\end{equation}
since $H(p,x;q)=-|p|_g^{q'}/{q'}$ with $\frac1q + \frac1{q'}=1$ whenever $p=Du$ is past-directed and timelike.
We may interchange the order of $x$ and $s$ derivatives as in (ii).  Since these derivatives of $F^\epsilon$
and $u^\epsilon$ agree with the corresponding 
approximate derivatives of $F$ and $u$ outside a set of volume $\epsilon >0$, 
we obtain \eqref{Hu}.
\end{proof}


We are now in a position to characterize the joint measure $\pi \in \Pi(\mu,\nu)$ maximizing \eqref{MK}.
Let $\Pac(M) \subset \PP_c(M)$ denote the measures $\mu$ which are absolutely continuous
with respect to the Lorentzian volume $\vol_g$.  



\begin{theorem}[Characterizing optimal maps]
\label{T:map}
Fix $0<q <1$. If $(\mu,\nu) \in \PP_c(M)^2$ is $q$-separated by $(\pi,u,v)$,
and $\mu \in \Pac(M)$, setting $X \times Y := \spt [\mu \times \nu]$ implies
(i) there is a unique map
$F(x) = \exp_x  DH(D\bar u(x),x;q)$ with $\nu = F_{\#}\mu$
such that  
$\bar u$ is Lipschitz 
and satisfies
\begin{equation}\label{instead of vlsc}
\bar u(x) = \max_{y \in Y} \frac1q \ell^q(x,y) - \bar u_{\td q}(y) 
\end{equation}
on a neighbourhood of $X$; in this case 
$\pi=(id \times F)_\# \mu$ uniquely maximizes \eqref{MK}, 
$u$ is semiconvex in a neighbourhood of $X$,
and both $Du = D\bar u$ and $(x,F(x)) \not \in \singl$
hold $\mu$-a.e.
 (ii) If, in addition, $\nu \in \Pac(M)$ then $F\circ G(y)=y$ holds $\nu$-a.e.\ and $G(F(x))=x$ holds $\mu$-a.e.\ 
where $G(y) := \exp_y -DH(-Du_{\tilde q}(y),y;q)$.  Here $H$ is from \eqref{q-Hamiltonian} and $F_\#$
from Definition \ref{D:push-forward}.
\end{theorem}

\begin{proof}
%
(i) Theorem \ref{T:duality by fiat} 
shows $(u,v)=(v_q,u_{\tilde q})$ on $X \times Y$, and that $u:=v_q$ and $v:=u_{\tilde q}$ are semiconvex 
Lipschitz functions on neighbourhoods of $X$ and $Y$. 
It also shows
$(u,v)$ attains the infimum \eqref{KM} and $\pi$ attains the maximum \eqref{MK}.
Let $S \subset X \times Y$ be the zero set of the non-negative function $u \oplus v - \frac1q \ell^q$.
When $q<1$,  for each 
$(x,y) \in S$ with $x \in \dom \tilde D^2u$,   Lemma \ref{L:Jacobi} goes on to assert $y = F_1(x) := \exp_x DH(Du(x),x;q)$
and $(x,F_1(x)) \not \in \singl$.
Since $\dom \tilde D^2 u$ is a set of full $\vg$ (hence $\mu \ll \vg$) measure by Alexandrov's theorem
(e.g.~\cite{Bangert79}), we deduce $\pi = (id \times F_1)_\#\mu$ from
e.g.~Lemma 3 of \cite{AhmadKimMcCann11}.
If $\pi' \in \Pi(\mu,\nu)$ also maximizes \eqref{MK},  then $\pi'$ vanishes outside $S$ because
of the duality \eqref{KM}, and we conclude $\pi' = (id \times F_1)_\#\mu$ as above.
This shows uniqueness of the maximizer when $q<1$ and $\mu \in \Pac(M)$.

Now suppose $F_\# \mu =\nu$,  where $F$ is defined as in the statement
of the theorem and $\bar u$ is Lipschitz, semiconvex and 
satisfies 
\eqref{instead of vlsc} 
in a neighbourhood of $X$.  We claim $\bar \pi = (id \times F)_\#\mu$ maximizes \eqref{KM}.
For each $x \in X \cap \dom D\bar u$,  the point $y \in Y$ attaining the maximum \eqref{instead of vlsc}
is given by $y=F(x)$, according to Lemma \ref{L:Jacobi}.  Thus
$$
\frac1q \ell^q(x,F(x)) = \bar u(x) + \bar u_{\td q}(F(x))
$$
holds on a set $X\cap \dom D\bar u$ whose complement is $\mu$-negligible.
Integrating this identity against $\mu$ yields
$$
\frac1q \int_{M \times M} \ell^q(x,y) d\bar \pi(x,y)= \int_M \bar u(x) d\mu(x) + \int_M \bar u_{\td q}(y) d\nu(y).
$$
where $F_\# \mu = \nu$ has been used.  This shows $\bar \pi$ maximizes \eqref{MK},
in view of the duality \eqref{KM}.  The uniqueness of maximizer established above implies 
$(id \times F)_\# \mu= (id \times F_1)_\# \mu$,  from which we conclude $F=F_1$ holds $\mu$-a.e. 
Finally,  $Du(\xo)=D \bar u(\xo)$ on the set $X \cap \dom Du \cap \dom D\bar u$
of full $\mu$-measure: Theorem~\ref{T:lapse smoothness} and its corollary show there cannot be multiple action minimizing geodesics 
joining $\xo$ to $F(\xo)$ unless $x \in M \mapsto \ell^q(x,F(\xo))$ is subdifferentiable but not superdifferentiable at $x=\xo$,  which would contradict the vanishing of 
$u(x) + u_{\td q}(F(\xo)) - \frac1q \ell^q(x,F(\xo)) \ge 0$ at $x = \xo \in \dom Du$.

(ii) When $\nu \in \Pac(M)$ a similar argument (or symmetry) shows $\pi = (G \times id)_\# \nu$.
In particular,  the set $(X \cap \dom Du) \times (Y \cap \dom Dv)$ is full measure for $\pi$,
and for each point $(x,y)$ in this set we have $y=F(x)$ and $x=G(y)$.  This shows $G$ acts 
$\mu$-a.e as left-inverse to $F$, and $\nu$-a.e.\ as right-inverse to $F$. 
\end{proof}

\begin{corollary}[Lagrangian characterization of $q$-geodesics]\label{C:q-geodesic}
Fix $0<q<1$. If $(\mu_0,\mu_1) \in \PP_c(M)^2$ is $q$-separated by $(\pi,u,v)$ and $\mu_0 \in \Pac(M)$
then $F_s(x) := \exp_x s DH(Du(x),x;q)$ defines the unique $q$-geodesic 
$s \in [0,1] \mapsto \mu_s := F_{s \#} \mu_0$ in $\PP(M)$ linking $\mu_0$ to $\mu_1$.
(We assume $u$ has been extended to a neighbourhood of $X$ by setting $u:= v_q$ in 
 \eqref{l-transform}, where $X \times Y := \spt[\mu_0 \times \mu_1]$.)
Moreover, $\mu_s \in \Pac(M)$ if $s<1$. 
\end{corollary}

\begin{proof}
Under these hypotheses,
Theorem \ref{T:map}(i)-(ii) assert the maximum \eqref{MK} to be uniquely attained 
by $\pi = (id \times F_1)_\#\mu$, where $\pi[\singl]=0$. Theorem~\ref{T:q-geodesics exist} then
implies the unique $q$-geodesic $\mu_s$ joining $\mu_0$ to $\mu_1$ to be
given by $z_{s\#} \pi = F_{s\#}\mu$, where the last identification follows from $z_s(x,F_1(x))=F_s(x)$.

For $s<1$,  Theorem \ref{T:Lipschitz inverse maps} asserts $F_s$ has a Lipschitz inverse.
Thus $F_s^{-1}(V)$ has zero Lorentzian volume if $V\subset M$ does,  
in which case absolute continuity of $\mu_0$ implies $\mu_s(V) = \mu_0(F_s^{-1}(V))$ also vanishes,
establishing absolute continuity of $\mu_s$. Compactness of its support is asserted by Corollary
\ref{C:compact support}.
\end{proof}

For reference, let us also state
the Lorentzian analog of Theorem 11.1 of \cite{Villani09};
its omitted proof
combines Theorem 3.83 of \cite{AmbrosioFuscoPallara00}
with Lemma 5.5.3 of \cite{AmbrosioGigliSavare05}
applied in local coordinates,  as in the Riemannian case.

\begin{theorem}[Jacobian equation]\label{T:Jacobian equation}
Let $(M^{n},g)$ be a 
Lorentzian manifold 
with a compatible Riemannian metric $\tilde g$. 
Let $0 \le f \in L^1(M,d\vg)$ and let $F:M\longrightarrow M$ be Borel.
Define $d\mu(x) = f(x)d\vg(x)$ and $\nu := F_\#\mu$.  Assume that: 
(i) $f$ vanishes outside a measurable set $\Sigma \subset M$ on which $F$ is injective; and
(ii) $F$ is approximately differentiable almost everywhere on $\Sigma$.

Define $JF(x):= |\det \td DF(x)|$ a.e. on $\Sigma$, where $\td DF$ denotes the approximate
gradient of $F$.  Then $\nu \ll \vg$ if and only if $JF(x)>0$ a.e.  In that case $\nu$ vanishes
outside $F(\Sigma)$, and its density $\rho$ is determined by the equation
\begin{equation}
f(x) = \rho(F(x)) JF(x).
\end{equation}
\end{theorem}

\begin{corollary}[Monge-Amp\`ere type equation]\label{C:Jacobian equation}
Under the hypotheses of Theorem \ref{T:map}(i)-(ii), $F$ is countably Lipschitz and the Jacobian equation
\begin{equation}\label{MA}
\rho_{0}(x) = \rho_1 (F(x)) JF(x)
\end{equation}
holds $\rho_0$-a.e.,  where $\rho_0 = d\mu/d\vg$,  $\rho_1 = d\nu/d\vg$ and $JF(x) = |\det \td D F(x)|$,
with $\td DF$ denoting the approximate derivative of $F$ from Definition \ref{D:approximate derivative}.
\end{corollary}

\begin{proof}
The potential $u=u_{\tilde q q}$ of Theorem \ref{T:map} is semiconvex by Theorem \ref{T:duality by fiat}.
As a consequence $u$ agrees with a $C^2$ function 
outside of a set of arbitrarily small volume.  Thus $F$ is countably Lipschitz,  hence approximately
differentiable $\vg$-a.e.  It is also injective $\mu$-a.e., according to Theorem \ref{T:map}(ii).
The Jacobian equation \eqref{MA} now follows
from Theorem \ref{T:Jacobian equation}.
\end{proof}

We call \eqref{MA} a Monge-Amp\`ere type equation since it reduces to a second-order degenerate elliptic
equation for the $\frac1q\ell^q$-convex potential $u$ of Theorem~\ref{T:map}, as in 
e.g.~
\cite {Villani09}.
Combining Corollaries \ref{C:q-geodesic} and \ref{C:Jacobian equation} yields an analogous 
equation \eqref{MAs}
for the density $\rho_s := d \mu_s/ d\vg$ along the $q$-geodesic $s \in [0,1] \mapsto \mu_s \in \Pac(M)$.  
This equation holds $\mu_0$-a.e., though the set where it holds may depend on $s \in [0,1]$.


\section{Entropic convexity from Ricci lower bounds}
\label{S:dc}

The key to understanding the behaviour of entropy along $q$-geodesics $s \in [0,1] \mapsto F_{s\#}\mu_0 \in \Pac(M)$ is to analyze the Jacobian factors $JF_s(x) := |\det \tilde DF_s(x)|$ 
which appear in the Monge-Amp\`ere type equations
\eqref{MAs}.  In a moving frame along the proper time maximizing segment $s\in[0,1] \mapsto F_s(x)$, 
Lemma \ref{L:Jacobi} asserts $A_s(x) := \tilde DF_s(x)$ is a matrix of Jacobi fields.  
The present section begins with a proposition harvesting consequences of the fact that its logarithmic derivative 
$B_s(x):=A'_s(x)A_s(x)^{-1}$ in time satisfies a matrix Riccati equation, whose trace involves the Ricci
curvature in the direction of the worldline $s \in [0,1] \mapsto F_s(x)$; 
 c.f.~\cite{CorderoMcCannSchmuckenschlager06} \cite{Case10}  \cite{TreudeGrant13} and Raychaudhuri's equation.
After a technical lemma,  Theorem~\ref{T:Boltzmann Hessian} gives explicit expressions
for the first two derivatives of the Boltzmann-Shannon and relative entropies $\EV(F_{s\#}\mu_0)$
along the geodesic in question.  Its corollary translates a non-negative lower Ricci curvature bound
into quantified convexity of the Boltzmann-Shannon entropy along $q$-geodesics.

\begin{proposition}[Jacobian along $q$-geodesics]\label{P:Jacobian evolution}
Fix $0<q<1$ and let $(\mu_0,\mu_1) \in \Pac(M)^2$ be $q$-separated by $(\pi,u,v)$.
Set $X \times Y = \spt[\mu_0\times\mu_1]$, $u:=v_q$ and $F_s(x) := \exp_x s DH(Du(x),x;q)$.
For $\vg$-a.e. $x \in X$,  the approximate derivative $A_s(x):=\td DF_s(x)  :T_xM \longrightarrow T_{F_s(x)}M$ 
exists, is invertible, depends smoothly on $s\in[0,1]$, and $\phi(s) := - \log |\det A_s(x)|$ satisfies
\begin{eqnarray}
\label{logdet1}
\phi'(s) &=& -\trace B_s(x),
\\ \label{logdet2}
\phi''(s) &=& \Rc_{F_s(x)}(F'_s(x),F'_s(x)) + \trace [B_s^2(t)],
\\ \label{logdet3}
{\rm and}\ \trace [B^2_s(x)] &\ge& \frac1{n}(\trace B_s(x))^2,
\end{eqnarray}
where $B_s(x):= A'_s(x)A_s(x)^{-1}$ and
$':=\frac{\p}{\p s}$ 
and the Ricci curvature $\Rc$ is computed with respect to the Lorentzian connection.
\end{proposition}

\begin{proof}
For $\vg\dae\ x\in X$, Lemma \ref{L:Jacobi} asserts that $A_s(x)$ and $\td D^2 u(x)$ exist,
and that $s \in[0,1] \mapsto A_s(x)w$ is a (smooth) Jacobi field for each $w \in T_xM$,
with $B_0(x) = D^2 H(Du(x)) \td D^2 u(x)$,
in view of Theorem \ref{T:duality by fiat}(iv). 
   Corollorary \ref{C:Jacobian equation} asserts $A_1(x)$ is invertible a.e.
Fixing such an $x \in X$, since $A_0(x)=I$ the set of $s$ values for which $\det A_s(x)=0$ forms a closed
subset of $(0,1)$ which we shall presently show to be empty.    Outside of this set, from 
$\phi(s)=-\trace \log |A_s(x)|$ we compute
\begin{eqnarray*}
\phi'(s) &=& -\trace B_s(x),
\\ {\rm and}
\quad \phi''(s) &=& - \trace[A_s''(x)A_s(x)^{-1}] + \trace [B_s^2(x)].
\end{eqnarray*}
Since $s\in [0,1] \mapsto A(s)w \in T_{F_s(x)}M$ is a Jacobi field for each $w \in T_xM$,
we can evaluate $\trace  \bar A''(s)\bar A(s)^{-1}$ via Jacobi's equation:
\begin{eqnarray*}
0 &=& (\nabla_{F'} (\nabla_{F'} A^i_{\bar j} )+ {R_{jkl}}^i{F'}^j {F'}^l A^k_{\bar j}) (A^{-1})^{\bar j}_i
\\ &=& \trace  A''(s) A(s)^{-1} + \Rc(F',F')
\end{eqnarray*}
to arrive at \eqref{logdet2};
here barred and unbarred indices refer to coordinate systems at $x$ and $F_s(x)$ respectively. 

We can now prove \eqref{logdet3}, at least when $s=0$.
Indeed,  this follows from Cauchy-Schwartz inequality for the Hilbert-Schmidt norm
$\|C\|^2 := \trace C^*C$ on $n \times n$ matrices $C$,  which asserts 
$$(\trace C^*D)^2 \le (\trace C^*C)( \trace D^* D),$$
when applied to $C= \sqrt{D^2 H} D^2 u \sqrt{D^2 H}$ and $D=I$,
noting $\trace D^*D =n$,
$\trace C^* = \trace C = \trace B_0$ and $\trace C^*C = \trace C^2 = \trace B^2_0$.
Here convexity of $H(p)$ plays the crucial role of ensuring $D^2 H$ is non-negative definite,
hence admits a matrix square-root.

The next step in the proof is to propagate the estimate \eqref{logdet3} from $s=0$ to $s>0$ using
the (Hopf-Lax) semigroup property for $q$-geodesics.  Theorem~\ref{T:Lipschitz inverse maps} asserts
that $F_s^{-1}$ extends to a Lipschitz map on $\spt \mu_s$,  whose image must have full measure
in $\spt \mu_0$ since $(F_{s}^{-1})_\#\mu_s=\mu_0$. Defining $F_s^t := F_s\circ F_t^{-1}$ whenever $t \le s$,
we deduce $s \in[0,1] \mapsto \mu_s = (F_s^t)_\#\mu_t$ is the $q$-geodesic connecting $\mu_t$ to $\mu_1$.
Moreover,  $F_s^t$ can be confirmed to be the $\ell^q$-optimal map between $\mu_t$ and $\mu_s$
as a consequence of Proposition \ref{P:RTIE} and Theorem~\ref{T:q-geodesics exist}.  
For fixed $t$  and $\mu_t\dae\ z$ set $\bar A_s(z) = DF_s^t(z)$ and 
$\bar B_s(z) = \bar A'_s(z)A_s(z)^{-1}$.  The preceding paragraph yields 
\begin{equation}\label{logdet4}
\trace [\bar B^2_t(z)] \ge \frac1{n}(\trace \bar B_{t}(z))^2.
\end{equation}
But $DF_s^t = DF_s \circ DF_t^{-1}$ and $(DF_s^t)' = DF_s' \circ DF_t^{-1}$,
whence $\bar B_s(z) =  (DF_s^t(x))' (DF_s^t(z))^{-1}= B_s(F_t^{-1}(z))$.  Thus \eqref{logdet4}
translates into the desired bound \eqref{logdet3}, at least on a set $X_s$ of full $\mu_0$ measure.
Although $X_s$ here depends on $s=t \in [0,1]$,  the bound \eqref{logdet3} holds on the intersection
 $\cap_{s \in \mathbf Q \cap [0,1]} X_s$ for all rational $s$, 
hence for all $s \in [0,1]$ since $B_s(x)$ depends smoothly on $s$.

Finally, \eqref{logdet1}-\eqref{logdet3} combine with $|F_s'(x)|=\ell(x,F_1(x))$ to show
\begin{eqnarray*}
\phi''(s) - \frac1{n}(\phi'(s))^2 
&\ge& K \ell^2(x,F_1(x)) 
\end{eqnarray*}
where the constant $K$ is a lower bound for the Ricci curvature of $M$ on the compact set 
$Z(\spt [\mu_0 \times \mu_1])$ of Lemma \ref{L:compact support}.  In particular,
$\phi(s)$ is semiconvex on the open set $S(x):=\{ s \mid \phi(s) \ne -\infty\}$. This yields a lower
bound for $\phi(s)$ throughout $[0,1]$ in terms of $\phi(0)$ and $\phi'(0)$ (or of
$(\phi,\phi')(\epsilon)$ if $\phi'(0)=-\infty$),  which shows $S(x)$ to be empty and $A_s(x)$ to be invertible.
\end{proof}

\begin{remark}[Relevance of Lipschitz inverse maps]\label{R:relevance of inverse maps}
The Monge-Mather shortening estimate of Theorem \ref{T:Lipschitz inverse maps}
is essential only to extend \eqref{logdet3} from $s=0$ to $s>0$.  Once we have this extension,
one can deduce the absolute continuity of $(F_s)_{\#} \mu_0$ for $s \in (0,1)$ from 
Theorem \ref{T:Jacobian equation} using the positivity of $JF_s(x)$ provided by
Proposition~\ref{P:Jacobian evolution}, as an alternative to Corollary \ref{C:q-geodesic}.
\end{remark}

\begin{lemma}[Second finite-difference representation]\label{L:second difference}
If $\phi \in L^\infty([0,1])$ is semiconvex on $(0,1)$ and $g(s,t) := \min\{s,t\}-st$, then
\begin{equation}\label{s-second difference}
(1-t)\phi(0) + t\phi(1) - \phi(t) = \int_{[0,1]} \phi''(s) g(s,t) ds
\end{equation}
for each $t\in[0,1]$,  where $\phi''$ denotes the distributional second derivative of $\phi$.
\end{lemma}

\begin{proof}
Semiconvexity and boundedness implies $\phi$ has a continuous extension $\bar \phi$ to $[0,1]$,
which coincides with $\phi$ except perhaps at the endpoints. 
For $\bar \phi$,  the representation \eqref{s-second difference} is asserted by Villani in (16.5) of \cite{Villani09}.
When $\phi$ differs from $\bar \phi$,  then $\phi''$ differs from $\bar \phi''$ only by derivatives of Dirac distributions
at the endpoints:
$$
\phi''(s) - \bar \phi''(s) = -(\phi(0)-\bar \phi(0)) \delta'(s) +(\phi(1)-\bar \phi(1)) \delta'(s-1).
$$
It is not hard to verify the representation \eqref{s-second difference} extends from $\bar \phi$ to $\phi$, after 
noting for each $t\in[0,1]$ that $g(s,t)$ depends smoothly on $s$
in a neighborhood of the endpoints of $[0,1]$,  where it vanishes.  (We can extend $\phi$ and $\bar \phi$ 
to be locally constant outside $(0,1)$ and $g(s,t)$ to be compactly supported and smooth outside $s=t\in[0,1]$ 
to facilitate this calculation.)
\end{proof}

\begin{theorem}[Displacement Hessian of relative entropy]\label{T:Boltzmann Hessian}
Fix $0<q<1$ and $V \in C^2(M)$ on a globally hyperbolic spacetime.
Let $s\in[0,1] \mapsto \mu_s =(F_s)_\# \mu_0 \in \Pac(M)$ be one of the $q$-geodesics
described by Corollary \ref{C:q-geodesic}.
 If $\ei(0)$ and $\ei(1)$ are finite, then:  (a) 
the relative entropy 
$e(s):=\EV(\mu_s)$ of \eqref{V-tropy}
is continuous and semiconvex on $s\in [0,1]$ and continuously differentiable 
on $s \in (0,1)$, with
\begin{eqnarray}
\label{Boltzmann gradient}
\ei'(s) &=& \int_M  [DV_{F_s(x)} F_s'(x) -\trace B_s(x)] d\mu_0(x) \qquad {\rm and}
\\ \label{Boltzmann Hessian}
\ei''(s) &=& \int_M  [\trace (B^2_s(x)) + (\Rc+D^2V)_{F_s(x)}(F'_s(x),F'_s(x))] d\mu_0(x)
\end{eqnarray}
holding on $[0,1]$ in the distributional sense.
Here $A_s(x) := \td D F_s(x) :T_xM \longrightarrow T_{F_s(x)}M$
denotes the approximate derivative of $F_s$,
$B_s(x):= A'_s(x)A_s(x)^{-1}$, 
$':=\frac{\p}{\p s}$ and $\trace [B_s(x)^2] \ge \frac1n
(\trace B_s(x))^2$. 
(b) 
The integral expression \eqref{Boltzmann Hessian} for $e''(s)$ 
depends lower semicontinuously on $s\in [0,1]$; the integrand is bounded below.
\end{theorem}

\begin{proof}
Our strategy will be to produce a finite second difference representation of $\ei$ using 
Lemma \ref{L:second difference}.

Let $F_s(x) := \exp_x sDH(Du(x),x;q)$ 
and $\mu_s=(F_s)_\#\mu_0 \in \Pac(M)$ be from Corollary \ref{C:q-geodesic}.
Proposition \ref{P:Jacobian evolution} asserts that $JF_s(x) := |\det \tilde DF_s(x)|$ exists and 
depends smoothly on $s\in[0,1]$
for each $x$ in a subset $X_0$ of full measure in $\spt \mu_0$.  
Letting $\rho_s:= d\mu_s/d\vg$,  Corollary \ref{C:Jacobian equation} gives
\begin{equation}\label{SMA}
\rho_s(F_s(x)) JF_s(x) = \rho_{0}(x)>0 
\end{equation}
on a subset $X_s \subset X_0$ of full $\mu_0$ measure.

Letting $Z:=Z(\spt [\mu_0 \times \mu_1])$ denote the compact set from Lemma \ref{L:compact support},
since $\spt \mu_s \subset Z$ the (Borel) change of variables $y=F_s(x)$ and  \eqref{U Jensen} yield
\begin{eqnarray}
\label{entropy lower bound}
-\infty &<& - \log \int_Z e^{-V} d\vg
\\ &\le& \ei(s) 
\nonumber \\ &=& \int_{M} [\log \rho_s(y) + V(y)]d\mu_s(y)
\nonumber \\ &=& \int_M [\log \rho_s(F_s(x)) + V(F_s(x)) ] d\mu_0(x)
\nonumber \\ &=& \int_M [ \log \rho_0(x) - \log |{ JF_s(x)}| + V(F_s(x))] d\mu_{0}(x)
\nonumber
\end{eqnarray}
where the last identity follows from \eqref{SMA}.  Thus
\begin{equation}\label{schematic second difference}
(1-t)\ei(0) + t\ei(1) - \ei(t) = \int_M [(1-t)\phi_x(0) + t \phi_x(1) - \phi_x(t)]  d\mu_{0}(x).
\end{equation}
where
\begin{equation}\label{phi0}
\phi_x(s) = -\log |JF_s(x)| + V(F_s(x))
\end{equation}

For $x \in X_0$ (which forms a set
of full $\mu_0$ measure),  
setting $A_s(x) = \tilde DF_s(x)$ and $B_s(x) = A_x'(x)A_s(x)^{-1}$,
Proposition \ref{P:Jacobian evolution} yields $\trace [B_s(x)^2] \ge \frac1{n}(\trace B_s(x))^2 \ge 0$,
\begin{eqnarray}\label{phi gradient}
\phi_x'(s) &=& DV(F_s(x))F_s'(x) - \trace B_s(x) \qquad {\rm and}
\\ 
\phi_x''(s) &=&  
\label{phi Hessian} 
\trace [B_s(x)^2] + (\Rc+D^2V)(F_s'(x),F_s'(x))
\\ &\ge & K_Z \ell (x,F_1(x))^2,
\label{phi extra}
\end{eqnarray}
where $F_s' = \frac{\p F_s}{\p s} \in T_{F_s(x)}M$ and $\nabla_{F_s'} F_s'=0$ since $s \in[0,1] \mapsto F_s(x)$
is an action minimizing geodesic segment.  
Here $K_Z$ denotes a lower bound for $\Rc + D^2V \ge K_Z g$ on the compact
set $Z \supset \spt \mu_s$ defined above, and we have used geodesy to conclude
$|F_s'(x)|= \ell(x,F_1(x))$.

Applying Lemma \ref{L:second difference} to \eqref{schematic second difference} yields
\begin{eqnarray} 
\nonumber && (1-t)\ei(0) + t\ei(1) - \ei(t) 
\\ \nonumber&=& \int_M \int_{[0,1]} \phi_x''(s) g(s,t) ds d\mu_{0}(x)
\\ &=& \int_M \int_{[0,1]} [\trace (B^2_s(x)) + (\Rc+D^2V)(F'_s(x),F'_s(x))] g(s,t) ds d\mu_{0}(x),
\label{iterated second difference}
\\ &\ge& \frac{K_Z}{2} t(1-t) \int_M \ell(x,F_1(x))^2 d\mu_{0}(x).
\nonumber
\end{eqnarray}
Since each subsegment of a $q$-geodesic is a $q$-geodesic, we deduce 
\begin{eqnarray*}
\frac{\ei(s) + \ei(t)}{2} -  \ei(\frac{s+t}{2})
&\ge& \frac{K_Z}8 \int_M \ell(F_s(x),F_t(x))^2 d\mu_{0}(x).
\\ &=& \frac{K_Z}8 \int_M \ell(x,F_1(x))^2 d\mu_{0}(x)
\\ &\ge& -\frac{1}8 \min\{K_Z,0\} \sup_{x,y \in Z} \ell(x,y)^2 
\\ &>& - \infty
\end{eqnarray*}
for all $0\le s \le t \le1$.  This shows the semiconvexity and upper boundedness
of $\ei$ on $[0,1]$,  and continuity on $(0,1)$, since \eqref{entropy lower bound} bounds $e(s)$ below 
and we have assumed finiteness of $e(0)$ and $e(1)$.


Applying Lemma \ref{L:second difference} to $\ei$, \eqref{iterated second difference} now yields
\begin{eqnarray} \label{mid Boltzmann Hessian}
\ei''(s) 
&=& \int_M [\trace (B^2_s(x)) + \Rc(F'_s(x),F'_s(x))] d\mu_{0}(x) 
\\ &\ge& K_Z \int_M \ell(x,F_1(x))^2 d\mu_{0}(x), 
\label{mid Boltzmann bound}
\end{eqnarray}
in the distributional sense.  The lower bound \eqref{mid Boltzmann bound}
implies continuity of $\ei$ at the endpoints of $[0,1]$,  
since otherwise $\ei''$ would contain a derivative of a Dirac delta measure.
  Using \eqref{phi0}--\eqref{phi Hessian} and Fubini's theorem, we can also
integrate \eqref{Boltzmann Hessian} twice to obtain 
\begin{eqnarray*}
\ei'(s) &=& c_1+\int_M  [DV_{F_s(x)} F_s'(x) -\trace B_s(x)] d\mu_0(x) \qquad {\rm and}
\\ 
\ei(s) &=& c_0+c_1s + \int_M [ V(F_s(x))- \log |{ JF_s(x)}|] d\mu_{0}(x).
\end{eqnarray*}
The boundary values determine the constants $c_0=\Ei(\mu_0)$ and $c_1=0$ of integration
by comparison with \eqref{entropy lower bound}, to establish \eqref{Boltzmann gradient}.

On a set $X_0$ of full measure,  the integrand $\phi''_x(s)$ depends smoothly on $s \in [0,1]$ 
and can be bounded below  independently of $x \in Z$ using \eqref{phi Hessian}. Lower semicontinuity 
of the integral \eqref{Boltzmann Hessian} representing $e''(s)$ therefore follows from Fatou's lemma.  
Similarly,  the addition of a linear term $ks$ makes the integrand $\phi'_x(s)$ from \eqref{phi gradient}
increase continuously in $s \in[0,1]$; continuity of $e'(s)$ on $(0,1)$ then follows from the representation
\eqref{Boltzmann gradient} by Lebesgue's dominated convergence theorem, to conclude the proof.
\end{proof}

\begin{definition}[(K,N) convexity; c.f. \cite{ErbarKuwadaSturm15}]
\label{D:KN convex}
Fix $K \in \R$ and $N>0$.  A function $e:[0,1] \longrightarrow [-\infty,\infty]$ 
is said to be 
{\em $(K,N)$-convex}
if  $e$ is upper semicontinuous,
$\dom e := \{ s \in [0,1] \mid e(s)<\infty\}$ is connected,
and either $e^{-1}(-\infty)$ contains the interior $I$ of $\dom e$
or is empty,  and in the latter case:  $e$ is semiconvex throughout $I$ and satisfies
$$
e''(s) - \frac1N (e'(s))^2 \ge K
$$
there, in the distributional sense.  The last clause merely means the second derivative of $e$
is interpreted distributionally;  semiconvexity implies $e'(s)$ has no singular part, hence $e'(s)^2$ can be 
interpreted in the pointwise a.e. sense.

Given a  globally hyperbolic spacetime $(M^n,g)$ and $0<q \le 1$,  
a functional $E:\PP(M) \longrightarrow \R \cup \{\pm \infty\}$
is said to be {\em weakly
$(K,N,q)$-convex} for $Q\subset \PP(M)^2$ if for each $(\mu_0,\mu_1) \in Q$ 
there is a $q$-geodesic in $\PP(M)$ joining $\mu_0$ to $\mu_1$ on which
$E(\mu_s)$ is $(K \LL_q(\mu_0,\mu_1)^2,N)$-convex.  $E$ is said to be 
{\em $(K,N,q)$-convex} for $Q$ if, 
in addition,  $E(\mu_s)$ is $(K \LL_q(\mu_0,\mu_1)^2,N)$-convex for {all} $q$-geodesics 
$s \in[0,1] \mapsto \mu_s \in \PP(M)$ with endpoints in $Q$.
\end{definition}


Recall also the definition \eqref{NBER tensor} 
of the $N$-Bakry-\'Emery-Ricci tensor:  $\Rc^{(n,0)}:=\Rc$ unless $N\ne n$, in which case
$$ 
\NRc_{ab} : = \Rc_{ab} + \nabla_a \nabla_b V - \frac{1}{N-n}\nabla_a V \nabla_b V.
$$ 

\begin{corollary}[
Entropic convexity from timelike lower Ricci bounds]\label{C:sufficiency}
Let $(M^n,g)$ be a 
globally hyperbolic spacetime.
Fix  $V \in C^2(M)$ and $N>n$. 
If  $\NRc(v,v) \ge K|v|^2_g \ge 0$ holds
in every timelike direction $(v,x) \in TM$, then 
for each $0<q<1$ the relative entropy $\EV(\mu)$ of \eqref{V-tropy} is 
$(K,N,q)$-convex for the set $Q \subset \Pac(M)^2$ of probability measures with 
$q$-separated endpoints.
\end{corollary}

\begin{proof}
Fix $0<q<1$.  If $(\mu_0,\mu_1) \in Q$ then
Corollary \ref{C:q-geodesic} describes the unique $q$-geodesic $s \in[0,1] \mapsto \mu_s \in \PP(M)$
joining any such pair of $q$-separated endpoints, and asserts that $\mu_s \in \Pac(M)$ for each $s\in[0,1]$.
Moreover, $e(s) := \EV(\mu_s) > -\infty$ by \eqref{U Jensen}. If $e(s)$ is finite at $s=0$ and $s=1$,
Jensen's inequality combines with Theorem \ref{T:Boltzmann Hessian} to estimate
\begin{eqnarray*}
\frac1N e'(s)^2 
&\le & \int_M  \Big(\frac{1+\epsilon^{-1}}N |DV(F_s) \cdot F'_s|^2 + (1+\epsilon)\frac {n} N \trace [B_s^2] 
\Big) d\mu_0
\\ &=& \int_M \Big(  \frac1{N-n} |DV(F_s) \cdot F'_s|^2 + \trace [B_s^2] \Big)d\mu_0
\end{eqnarray*}
by choosing $\epsilon=\frac{N-n}{n} > 0$. The same theorem yields 
continuity of $e(s)$ on $[0,1]$, semiconvexity on $(0,1)$, and the distributional bound
on $e''(s)$ given by
\begin{eqnarray}
\nonumber
e''(s) -\frac1N e'(s)^2 &\ge& \int_M \NRc(F'_s,F'_s) d\mu_0 
\\ &\ge& K \int_M \ell^2(x,F_1(x)) d\mu_0 
\label{K bad}
\\ &\ge& K \LL_q(\mu_0,\mu_1)^2,
\label{K good}
\end{eqnarray}
where the second and third estimates follow from the lower bound 
$\NRc \ge K g \ge 0$ in timelike directions and the $q$-separation 
$|F'_s(x)|_g = \ell(x,F_1(x)) >0$ via Jensen's inequality. 
If $e(s)$ is infinite at either endpoint, we can apply the foregoing argument on any subinterval of $[0,1]$
having finite entropy at its endpoints to reach the desired conclusion.
\end{proof}

\begin{remark}[Restrictions $K \ge 0$ and $N\ge n$]\label{R:Jensen}
The preceding proof uses $K \ge 0$ only to pass from \eqref{K bad} to \eqref{K good}.
Its conclusion also extends directly to $N=n$ using Theorem \ref{T:Boltzmann Hessian}, provided
$V=0$ (and recalling $\Rc^{(n,0)}_{ab}  := \Rc_{ab}$). 
\end{remark}




\section{Relaxing separation from the null future}
\label{S:relaxing separation}



Considerations henceforth have been restricted to $q$-geodesics
whose endpoints $(\mu_0,\mu_1)$ are $q$-separated.  In this chapter we relax this restriction,
to allow endpoints which merely admit an $\ell^q$-optimal 
$\pi \in \Pi(\mu_0,\mu_1)$ with $\ell>0$ holding $\pi$-a.e.
Corollary \ref{C:sufficiency2} asserts equivalence of timelike lower Ricci
bounds to {\em weak} 
$(K,N,q)$ convexity of the relative entropy 
on the enlarged set of geodesics which arise in this more general setting.

Under these weaker hypotheses,  we no longer know whether or not 
strong duality holds: i.e.\ we assume only that the dual infimum \eqref{KM} is finite, but not that it is attained;
see e.g.~\cite{BeiglbockGoldsternMareschSchachermayer09} and its references. 
Nevertheless,  the following theorem 
decomposes the more general $\ell^q$-optimal measures $\pi$ which vanish on $\singl$
into countably many 
components whose left and right marginals are $q$-separated (iii).
This allows us to deduce (i) the uniqueness of $\pi$ and $\ell^q$-cyclical monotonicity of its support; 
(ii) the existence of Monge maps $F$; 
(iv) absolute continuity of $\mu_s$ along the corresponding $q$-geodesic.
Our strategy for obtaining the existence and uniqueness results (i)-(ii) without dual attainment is inspired by
Gigli's approach to a similar question 
 in a less smooth setting \cite{Gigli12}.  
The arguments of this section become somewhat simpler 
if one is satisfied to have results only for compactly supported measures


 




\begin{theorem}[Maps characterizing interpolants without duality]
\label{T:merge}
Let $(M,g)$ be a globally hyperbolic spacetime.
Fix  $V \in C^2(M)$, $0<q<1$, $\mu \in \PPac(M)$ and $\nu \in \PP(M)$
for which the infimum \eqref{KM} is finite. 
 Then
  (i) at most one $\ell^q$-optimal $\pi \in \Pi(\mu,\nu)$ has the additional property that 
$\ell>0$ holds $\pi$-a.e.
 (ii) If such a joint measure exists, 
then $\pi=(id \times F)_\#\mu$ for some map $F: 
\spt\mu \longrightarrow \spt \nu$ and $\pi[\singl]=0$.
 (iii)~Moreover, $\pi=\sum_{i=1}^\infty \pi^i$ 
decomposes into countably many non-negative, 
mutually singular  measures 
such that $\cup_{i=1}^\infty \spt \pi^i$ is $\ell^q$-cylically monotone and
 the marginals $(\mu^i,\nu^i)$ of $\hat \pi^i := \pi^i/\pi^i[M^2]$ 
have $\spt [\mu^i \times \nu^i]$ compact and disjoint from $ \npl$.
For each $i \in \N$,
the map $F$ agrees $\mu^i$-a.e. with the unique
$\ell^q$-optimal map $F^i$ pushing $\mu^i$ forward to $\nu^i$ from Theorem \ref{T:map};
moreover Graph$(F^i) \subset \spt \pi^i$.
 (iv) The $q$-geodesic $(\mu_s)_{s\in[0,1]}$ defined by  $\mu_s:=(z_s)_\# \pi$
 and \eqref{affine segment} satisfies $\mu_s \in \PPac(M)$ for $s<1$.
(v) The measures $\mu_s^i:=(z_s)_\# \pi^i$ 
decompose $\mu_s$ into 
mutually singular
pieces for $s<1$.
 (vi)  The sum $\pi = \sum_{i} \pi^i$ is finite if and only if
$\spt \pi$ is compact and disjoint from $\npl$. 
%
\end{theorem}

\begin{proof}
(iii)-(iv) and (vi):
Suppose $\pi \in \Pi(\mu,\nu)$ is $\ell^q$-optimal and $\ell>0$ holds $\pi$-a.e.
Since $M$ is a manifold and $\pl$ is open by Theorem~\ref{T:lapse smoothness}, 
$\pl \cap \spt \pi$ can be covered
by open rectangles $U \times W$ whose compact closures are contained in $\pl$.  In fact, countably
many such rectangles suffice due to the second countability of $M$;
finitely many suffice if $\spt \pi$ is compact and contained in $\pl$.
Setting $\pi^0=0$, define $\pi^i$ inductively as the restriction of $\pi-\pi^{i-1}$ to the $i$-th rectangle,
so that $\pi = \sum_{i=1}^\infty \pi^i$, where the summands $\pi^i$ are mutually singular 
and each $\pi^i$ vanishes outside the $i$th rectangle.
Denote the marginals of $\pi^i$ by $\mu^i$ and $\nu^i$, and 
normalize $\hat \pi^i := \pi^i/\pi^i[M^2]$ 
whenever $\pi^i$ is non-vanishing. 
Its marginals $(\hat \mu^i,\hat \nu^i)$ are $q$-separated 
by a pair of potentials $(u^i,v^i)$ 
according to Lemma~\ref{L:well-separated},
and $\pi^i$ and the partial sum $\sum_{k=1}^i \pi^k$ both 
inherit $\ell^q$-optimality from $\pi$ by e.g. Theorem~4.6 of \cite{Villani09}, which requires
finiteness of \eqref{KM}.  
Theorem~\ref{T:map} then asserts that $\pi^i = (id \times F_1^i)_\# \mu^i$ and $\pi^i[\singl]=0$, 
where $F_s^i = \exp s DH \circ Du^i$. Corollary \ref{C:q-geodesic} asserts that 
$\mu^i_s:=(z_s)_\#\pi^i$ is absolutely continuous for each $s<1$, 
establishing (iv).
Compactness of $\spt \pi^i$ allows us to 
extend $F_1^i$ from $\dom Du^i$ to $\spt \mu^i$ so as to ensure Graph$(F^i_1) \subset \spt \pi^i$.
Since the support of $\sum_{k=1}^i \pi^k$ is compact,  it lies a positive distance from the closed
set $\npl$, establishing (iv). 
Continuity of $\ell^q$ on a neighbourhood of $\cup_{k=1}^i \spt \pi^k$
ensures the latter is $\ell^q$-cyclically monotone by the well-known perturbation argument 
from my work with Gangbo~\cite{GangboMcCann96}.  Since $c$-cyclical monotonicity is checked on finite collections
of points,  it also holds for the limiting set $\cup_{k=1}^\infty \spt \pi^k$.
Setting $\mu^{ij} := \min\{\mu_i,\mu_j\}$, we next claim that 
$F_1^i = F_1^j$ holds $\mu^{ij}$-a.e.  

To derive a contradiction suppose for some $i<j $ there is a set $S$ of positive measure
for both $\mu^i$ and $\mu^j$ on which $F_1^i \ne F_1^j$.  
We may also suppose $\mu^i$ and $\mu^j$
to be given by densities 
with respect to $\vg$ which are bounded above and below on $S$.
Take $S$ smaller if necessary to be compact, and so that for each $k\in \{i,j\}$,
the map
$F^k_s$ has approximate derivative $\tilde DF_s^{k}(x)$ 
depending smoothly on $s\in[0,\frac12]$ and bounded above and below throughout $S$ in view of
Proposition \ref{P:Jacobian evolution}.

The compactness of $S$ ensures the existence of an $r$-neighbourhood $S^r$ of $S$ for some $r>0$ whose volume
$\vg[S^r] < \frac32 \vg[S]$ is not much larger than that of $S$. 
Since the maps $F_{s}^{i/j}$ have bi-Lipschitz restrictions to $S$ for $s\le 1/2$, stay far away from the cut locus,
and coincide with the identity map when $s=0$,  
taking $s>0$ sufficiently
small ensures 
that the compact sets $F_s^i(S)$ and $F_s^j(S)$
are contained in $S^r$ and both
have volume larger than, say, $\frac34\vg [S]$.
Their intersection therefore has positive volume, so
there exist $x,y \in S$ with $F^i_s(x)=F^j_s(y)$.  By Proposition~\ref{P:trajectories don't cross}
this forces $x=y$ and $F^i_1(x)=F^j_1(y)$,  since apart from a negligible set,
the graphs of both $F^i_1$ and $F^j_1$ lie in the 
$\ell^q$-cyclically monotone set $\spt [\pi^i + \pi^j]$.  This contradicts the definition of $S$,
to establish (iii) that $F_1^i=F_1^j$ holds $\mu^{ij}$-a.e.

(i)--(ii) Now $F:=F^i$ is well-defined $\mu$-a.e.  Since $\pi^i$ vanishes outside Graph$(F) \cap \singl$  
for each $i$,
we see $\pi = (id \times F)_\#\mu$ by e.g.\ Lemma 3.1 of \cite{AhmadKimMcCann11}.
If there were a second $\ell^q$-optimal $\pi' \in \Pi(\mu,\nu)$ with $\ell>0$ holding $\pi'$-a.e.,
we could apply the foregoing argument to $\tilde \pi := (\pi + \pi')/2$ to deduce the existence of a 
map $\tilde F$ such that $\tilde \pi = (id \times \tilde F)_\# \mu$.  Since both $\pi$ and $\pi'$ vanish 
outside the graph of $\tilde F$,  we conclude $\pi = (id \times \tilde F)_\# \mu = \pi'$ as before, to establish
 the uniqueness of $\pi$.

(v)   Fix $i\ne j$. Then $\mu^i$ and $\mu^j$ inherit mutual singularity from $\pi^i$ and $\pi^j$,
because $(id \times F)_\#\min\{\mu^i,\mu^j\}$ --- being common to $\pi^i$ and $\pi^j$ --- must vanish.
Inner regularity provides disjoint $\sigma$-compact sets 
$U^i\subset \spt \mu^i$ such that 
\begin{equation}\label{scompact separation}
\mu^i[U^{j}] = \left\{
\begin{array}{cl}
\mu^i[M] & {\rm if}\ i=j \\
0 & {\rm else.} 
\end{array}
\right.
\end{equation}
We claim the $\{\mu^i_s\}_{j=1}^{\infty}$ remain mutually singular for each $s\in (0,1)$.
Indeed, $\mu^i_s$ vanishes outside the $\sigma$-compact set $F_s(U^i)$,
which we claim is disjoint from $F_s(U^{j})$ unless $i=j$.
Notice $z \in F_s(U^i) \cap F_s(U^{j})$ implies  $U^i$ intersects $U^{j}$
by Proposition \ref{P:trajectories don't cross} and the $c$-cyclical monotonicity of $\cup_{i=1}^\infty \spt \pi_i$.
But this intersection forces $i=j$ to conclude the proof.
\end{proof}

We next aim to establish expressions for the first two derivatives of the relative entropy 
$e(s):=\EV(\mu_s)$ along $q$-geodesics whose endpoints need not be $q$-separated,  
by extending Theorem \ref{T:Boltzmann Hessian} to the present setting.  
We extend the entropy $\EV(\mu)$ to subprobability measures by the same prescription \eqref{V-tropy}
as for probability measures.  
The following pair of lemmas are known but included for completeness.

\begin{lemma}[Domain of the relative entropy]\label{L:entropic domain}
Let $m$ and $\mu$ be Borel measures on a metric space $(M,d)$,
with $\mu$ absolutely continuous with respect to $m$ and $\mu[M]<\infty$. Set
\begin{equation}\label{relative entropy}
E_\pm(\mu | m) := \int_{M} \left[\frac{d\mu}{dm} \log \frac{d\mu}{dm} \right]_\pm   dm,
\end{equation}
where $[a]_\pm := \max\{\pm a,0\}$.
(i) If $0\le \nu \le \mu$ 
and $E_+(\mu | m)$ (or $E_-(\mu | m)$) is finite,
then $E_+(\nu | m)$ (respectively $E_-(\nu | m)$) is finite. 
If neither is finite then $E(\mu | m) :=-\infty$; otherwise 
$E(\nu | m) := E_+(\nu | m) - E_-(\nu | m)$ satisfies
\begin{equation} \label{semibounds}
-\mu[M]-E_-(\mu | m) \le E(\nu | m) \le E_+(\mu | m).
\end{equation}
(ii) If $\mu = \sum_{i=1}^\infty \mu^i$ and the $\mu^i$ are mutually singular,
then either $E(\mu | m) = -\infty$ or $E(\mu | m) = \lim\limits_{k\to\infty} E(\sum\limits_{i=1}^k \mu^i\mid m)$.
\end{lemma}

\begin{proof}
(i) Fix Borel measures $0 \le \nu \le \mu$ and $m$ on $(M,d)$ with $\mu[M]<\infty$ and 
$\mu$ absolutely continuous with respect to $m$.
Let $\rho := d\mu/dm$ and $\sigma := d\nu/dm$ denote the Radon-Nikodym
derivatives of $\mu$ and $\nu$ with respect to $m$.   Since $\sigma \le \rho$ and $\sigma \log \sigma \ge -1/e$, if $r>0$ then 
\begin{eqnarray*}
\int_{\{\rho > r\}} \sigma \log \sigma dm 
&\le&  \int_{\{\rho > 1\}} \rho \log \rho  dm
= E_+(\mu | m)
\\ {\rm and} \quad 
\int_{\{\rho > r\}} \sigma \log \sigma dm 
&\ge&  -\frac {\mu[M]}{er}  > -\infty
\end{eqnarray*}
by Chebyshev's inequality. This shows $E_+(\nu | m)$ is finite if $E_+(\mu \mid m)$ is.
On the other hand,
monotonicity of $\rho \log \rho$ on $[0,1/e]$ yields 
\begin{eqnarray*} 
0 \ge \int_{\{\rho \le \frac1 e\}} \sigma \log \sigma dm
 &\ge & \int_{\{\rho \le \frac1 e\}} \rho \log \rho dm 
\ge E_-(\mu | m).
\end{eqnarray*}
Taking $r=1/e$ we can sum these two estimates to conclude 
$E_-(\nu | m)$ is finite if $E_-(\mu | m)$ is, and obtain \eqref{semibounds}
unless both bounds diverge.

%

(ii) Let $\sigma := d\mu/dm$ and  $\sigma^i := d\mu^i/dm$.  
Since the $\mu^i\ge 0$ are 
mutually singular and $\mu=\sum \mu_i$ is absolutely continuous with respect to $m$,  
for $m$-a.e. $x$ only one of the three inequalities 
$0 \le \sigma^k(x) \le \sigma^{k+1}(x) \le \sigma(x)$ can be strict. Thus
\begin{eqnarray}
 \lim_{k \to \infty} \int_{\{\sigma > 1\}} \sigma^k \log \sigma^k dm
&=& \int_{\{\sigma > 1\}} \sigma \log \sigma dm
\\ {\rm and}\ 
 \lim_{k \to \infty} \int_{\{\sigma \le 1\}} \sigma^k \log \sigma^k dm
&=& \int_{\{\sigma \le 1\}} \sigma \log \sigma dm
\label{lower divergence}
\end{eqnarray}
follow from Lebesgue's monotone convergence theorem, establishing (ii).
\end{proof}

\begin{lemma}
[Consequences of Helly's selection theorem]
\label{L:Helly}
Given $c \in \R$ and a sequence of convex functions $f_k:[0,1] \longrightarrow [-\infty,c]$,
a subsequence $f_{k(j)}$ converges pointwise a.e. to a convex limit $f:[0,1] \longrightarrow [-\infty,c]$
satisfying either
\begin{eqnarray}
\label{proper}
\inf_{0\le s \le 1} f(s) &>& -\infty \qquad {\rm (proper)}
\\ {\rm or } \qquad \sup_{0<s<1} f(s) &=& -\infty \qquad {\rm (improper)}.
\label{improper}
\end{eqnarray}
In the proper case, the derivatives $f' = \lim_{j \to \infty} f'_{k(j)}$ converge pointwise a.e. 
and the second derivatives $\ds f'' = \lim_{j \to \infty} f''_{k(j)}$ converge distributionally on $(0,1)$.
\end{lemma} 

\begin{proof}
The proof is standard, hence omitted.
\end{proof}

\begin{theorem}
[Displacement Hessian of the relative entropy again]
\label{T:Boltzmann Hessian2}
Let $(M^{n},g)$ be a globally hyperbolic spacetime.
Fix  $V \in C^2(M)$, $N \ge n$, and $0<q<1$.  Assume $V=0$ if $N=n$. 
Fix $\mu,\nu \in \PPac(M)$ 
for which the infimum \eqref{KM} is finite and the supremum \eqref{MK} is attained by
some $\pi \in \Pi(\mu,\nu)$ with $\ell>0$ holding $\pi$-a.e.  
Assume the relative entropy $e(s):=\EV(\mu_s)$ wth $\mu_s := z_{s \#}\pi$ and 
map $F_s(x) := z_s(x,F(x))$ from Theorem \ref{T:merge} satisfy $\max\{e(0),e(1)\}<\infty$
and $\ds \sup_{0<s<1} e(s)>-\infty$ and 
\begin{equation}\label{uniform L1 bound}
C:= \left\| \int_M \min \{\NRc_{F_s(x)}(\frac{\p F}{\p s}, \frac{\p F}{\p s}), 0\}d\mu_s \right\|_{L^\infty([0,1])} <\infty.
\end{equation}
Then the conclusions of Theorem \ref{T:Boltzmann Hessian}(a) remain true,  except that $e(\cdot)$
may be upper semicontinuous rather than continuous at the the endpoints of the interval $s\in[0,1]$.
\end{theorem} 

\begin{proof}
Fix $\mu,\nu \in \PPac(M)$ and $\pi \in \Pi(\mu,\nu)$ as described.
Let the map $F$, $q$-geodesic $(\mu_s)_{s \in [0,1]} \subset \PPac(M)$ 
and mutually singular decompositions $\pi = \sum_{i=1}^\infty \pi^i$ 
and  $\mu_s := \sum \mu_s^i$ with $\mu_s^i := z_{s\#}\pi^i$
and $\ell^q$-cyclically monotone
$\spt \pi^i \subset \pl$ 
be given by Theorem \ref{T:merge}, which also asserts $\pi[\singl]=0$.
Normalizing $\hat \mu^i := \mu^i/\mu^i[M]$ and defining $\hat \nu^i$ and $\hat \pi^i$ analogously, the marginals
$(\hat \mu^i,\hat \nu^i)$ of $\hat \pi^i$ are $q$-separated by Lemma \ref{L:well-separated},
and $F$ coincides a.e.~with the unique
optimal map between them provided by Theorem \ref{T:map}, 
so $\hat \pi^i$ is $\ell^q$-optimal.
Moreover, $\mu^i:=\mu^i_0$ and $\nu^i:=\mu^i_1$ inherit
an upper bound on their entropy from $\max\{e(0),e(1)\}<\infty$ by 
Lemma \ref{L:entropic domain}; being compactly supported they inherit a 
lower bound on their entropy from \eqref{U Jensen}.
Their normalized versions 
also have finite entropy according to the scaling law 
\begin{equation}\label{entropy scaling}
\EV(\lambda \nu) = \lambda \EV(\nu) + \nu[M] \lambda \log \lambda 
\end{equation}
for $\lambda>0$.
Setting 
$\hat e_i(s):= \EV(\hat \mu_s^i)$ and
applying Theorem \ref{T:Boltzmann Hessian} to $(\hat \mu^i_0, \hat \mu^i_1)$  yields 
\begin{eqnarray}
\label{Boltzmann gradient2}
e_i'(s) &=& \int_{M}  [DV_{F_s(x)} F_s'(x) -\trace B_s(x)] d \mu^i_0(x) \qquad {\rm and}
\\ \label{Boltzmann Hessian2}
e_i''(s) 
&=& \int_{M}  [\trace (B^2_s(x)) + (\Rc+D^2V)_{F_s(x)}(F'_s(x),F'_s(x))] d \mu^i_0(x)
\end{eqnarray}
on $s \in (0,1)$ with and therefore also without the hats denoting normalization.

The mutual singularity of $\{\mu^i_s\}_{i=1}^\infty$ asserted by Theorem \ref{T:merge}
extends to $s=1$ by the $s \leftrightarrow 1-s$ symmetry.
We can therefore obtain \eqref{Boltzmann gradient}--\eqref{Boltzmann Hessian}
by summing  \eqref{Boltzmann gradient2}--\eqref{Boltzmann Hessian2} over $i \in \N$
--- provided these sums do not diverge.  More precisely,  define $f_k(s) := \sum_{i=1}^k e_i(s)$. 
Hypothesis \eqref{uniform L1 bound}  combines with $\trace (B_s(x)^2)\ge 0$ 
 from Theorem \ref{T:Boltzmann Hessian} in \eqref{Boltzmann Hessian2} to show
$\frac C2 s^2 + f_k(s)$ is convex on $[0,1]$.
Lemma \ref{L:entropic domain} shows $f_k(0)$ and $f_k(1)$ to be bounded above in terms of $C$ and
the endpoints $\mu$ and $\nu$,  and asserts for each $s \in [0,1]$ that either
$e(s) := \EV(\mu_s) = \lim_{k\to \infty} f_k(s)$ or $e(s)=-\infty$.   We assert:

{\em Claim 1:} If $e(s) = -\infty$ for some $s \in (0,1)$ then $\sup_{0<t<1} e(t) = -\infty$.

Taking Claim 1 for granted (and postponing its proof),  
if $e(t)\ne -\infty$ for some $t \in (0,1)$, then
Claim 1 yields $e(s) \ne -\infty$ for all $s\in(0,1)$, hence pointwise convergence of the full sequence 
$\frac C2 |\cdot|^2 + f_k(\cdot)$ to the limit $\frac C2 |\cdot|^2 + e(\cdot)$, which must be convex
and real-valued by Lemma~\ref{L:Helly}.
The same lemma asserts $\ds  e'(s) = \lim_{k \to \infty} f_k'(s)$ pointwise a.e.
and  $\ds e''(s) = \lim_{k \to \infty} f_k''(s)$ distributionally on $(0,1)$.

Unless $(0,1) \subset e^{-1}(-\infty)$,  
\eqref{Boltzmann Hessian} therefore follows by summing \eqref{Boltzmann Hessian2} using Lebesgue's monotone
convergence theorem and the pointwise lower bound 
established above showing its integrand $\ge -C$.
Integrating \eqref{Boltzmann Hessian} yields
\eqref{Boltzmann gradient} and its continuous
dependence on $s \in (0,1)$ exactly as in the proof of 
Theorem \ref{T:Boltzmann Hessian}. To complete the proof of the theorem,  it remains only to verify Claim 1,
which we do in a series of subclaims. 

{\em Proof of Claim 1:}
Let $\rho_t := d\mu_t/dm$ denote the Radon-Nikodym derivative of $d\mu_t(x)$ with respect to
the weighted Lorentzian volume $dm(x)=e^{-V(x)} d\vg(x)$.
Set $N^i := (\spt \pi^i)\setminus \singl$ 
and $N^\infty = \cup_{i=1}^\infty N^k$.
%
Assume $e(s) = -\infty$ for some fixed value of $s \in(0,1)$, since otherwise there is nothing to prove. 
To establish $e(t) =-\infty$ for all other $t\in (0,1)$, observe $z_t:N^\infty \longrightarrow M$ is smooth
by Lemma \ref{L:midpoint continuity} and its inverse map is countably Lipschitz on $z_t(N^\infty)$ 
by Theorem \ref{T:Lipschitz inverse maps}.  Since $N^\infty$ carries the full mass of $\pi$,
inner regularity of $\mu_s$ yields a $\sigma$-compact $U$ of $\{z \in z_s(N^\infty) \mid \rho_s(z) \le 1\}$ 
which differs from the latter by a $\mu_s$ negligible set.
Let $\bar \pi^i$ denote the restriction of $\pi^i$ to the $\sigma$-compact set $S:=z_s^{-1}(U)$ 
with the convention 
$\pi^\infty := \pi$. Set $\bar \mu^i_t := z_{t\#}(\bar \pi^i)$ and $\bar \nu^i_t := \sum_{j=1}^i \bar \mu^j_t$
for each $i \in \N \cup \{\infty\}$ and $t\in [0,1]$, and denote their entropies by $\bar e_i(t) := \EV(\bar \mu^i_t)$ and
$\bar f_i(t) := \EV(\bar \nu^i_t)$.
Then $0 \le \bar \mu^i_t \le \mu^i_t$ inherit absolute continuity and mutual singularity from $\{\mu^i_t\}_{i \in \N}$
so $\bar f_k(t)=\sum_{i=1}^k \bar e_i(t)$.

{\em Claim 2:} Setting $\rho^\infty_t = \rho_t$, the following identity holds $m$-a.e.:
$$
\bar \rho^i_t := \frac{d\bar \mu^i_t }{dm} =1_{z_t(S)} \rho^i_t, \qquad \forall t \in (0,1) \mbox{\rm\ and}\ i \in \N \cup \{\infty\}.
$$
{\em Proof of Claim 2:}
Absolute continuity of $\mu_s$ implies that $U=z_s(S)$ has either full or zero Lebesgue density $\mu_s$-a.e.
Since $z_t$ is countably biLipschitz on $S$,  and $\mu_s = z_{s\#}\pi$,  it follows that $S$ has either 
full or zero $n$-dimensional density in $N^\infty$ (or alternately,  in the $n$-dimensional Lipschitz
submanifold $W \subset M \times M$ guaranteed to contain
$N^\infty$ (hence $\spt \pi$) by my result proved with Pass and Warren
\cite{McCannPassWarren12}, which relies on the non-degeneracy of $\ell^q$ asserted in Corollary \ref{C:twist}).
In fact,  absolute continuity of $\pi$ in coordinates on $W$ also follows from that of $\rho_s$ and the 
countably biLipschitz character of $z_s$;  thus we may say 
$S$ has full or zero density Lebesgue density $\pi$-a.e.\ on $W$.
Similarly, it follows that $z_t(S)$ has either full or zero Lebesgue density $\mu_t:=z_{t\#}\pi$ 
a.e.\ for each $t\in(0,1)$ to establish claim 2.

{\em Claim 3:} If $e(s)=-\infty$ then $\lim_{k \to \infty} \bar f_k(s) = -\infty$.

{\em Proof of Claim 3:} The absolute continuity and mutual singularity of $\bar \mu^i_t \le \mu^i_t$ 
and our choice of $U=z_s(S)$ imply
\begin{eqnarray*}
-\infty = e(s)
 &=& \int_M \rho_s \log \rho_s dm
\\ &\ge& \int_{\{ \rho_s \le 1\}} \rho_s \log \rho_s dm
\\ &=& \sum_{i=1}^\infty \int_{z_s(S)} \rho_s^i \log \rho_s^i dm
\\ &=& \lim_{k\to \infty} \bar f_k(s)
\end{eqnarray*}
to establish Claim 3.

{\em Claim 4:}   If $e(s)=-\infty$ for some $s \in (0,1)$ then 
$\ds \sup_{0<t<1}\mathop{\lim\sup}\limits_{k\to \infty} \bar f_k(t) =-\infty$.

{\em Proof of Claim 4:}
Let $\hat {\bar \pi}^i := \bar \pi^i/\bar \pi^i[M]$ and normalize $\hat {\bar \mu}^i_t$ similarly.
Now $\bar \pi^i \le \pi^i$ implies
$\hat {\bar \pi}^i$ inherits $\ell^q$-optimality from $\hat \pi^i$, and
its marginals $(\hat {\bar \mu}^i_0,\hat {\bar \mu}^i_1)$ inherit $q$-separation from those of $\hat \pi^i$.
Thus $\{\hat {\bar \pi}^i_t\}_{t \in[0,1]}$ is a $q$-geodesic for each $i$ and 
Theorem~\ref{T:Boltzmann Hessian} shows convexity of $\frac C2 t^2 + {\bar f}_i(t)$ on $t\in[0,1]$
using \eqref{entropy scaling}.
Lemma \ref{L:entropic domain} bounds ${\bar f}_i(0)$ and ${\bar f}_i(1)$ above in terms of $\mu,\nu$ and $C$.
Defining $\bar f(t):= \lim\sup_{i \to \infty} f_i(t)$,  claim 3 yields $\bar f(s)=-\infty$ hence 
Lemma \ref{L:Helly} implies $\ds \sup_{0<t<1}\bar f(t)=-\infty$ 
to establish Claim 4.

{\em Claim 5:}   If $e(s)=-\infty$ for some $s \in (0,1)$ then $\ds \sup_{0<t<1} e(t) = -\infty$.

{\em Proof of claim 5:}
Claim 2 yields
\begin{eqnarray*}
e(t) &=& \int_M \rho_t \log \rho_t dm
\\ &=& \int_{z_t(S)} \bar \rho_t \log \bar \rho_t dm + \int_{M \setminus z_t(S)} \rho_t \log \rho_t dm.
\end{eqnarray*}
The first summand coincides with $\bar f_\infty(t):=\EV(\bar \mu^\infty_t)$, which diverges to $-\infty$
by  Lemma \ref{L:entropic domain} combined with Claim 4.  Thus $e(t) =-\infty$ by the convention
from Definition \ref{D:V-tropy},
regardless of whether or not the other integrals are well-defined.
This establishes Claims 1 and 5, 
hence the theorem.
\end{proof}


The following result provides an analog to Corollary \ref{C:sufficiency}.
We obtain weak rather than strong $(K,N,q)$ convexity in this 
setting since we do not
know whether or not there are other $\ell^q$-optimal measures $\pi \in \Pi(\mu,\nu)$ for
which $\ell>0$ fails to hold $\pi$-a.e.  If such measures exist,  they generate $q$-geodesics
via Theorem \ref{T:q-geodesics exist} which we have not developed the machinery to analyze.

\begin{corollary}
[Weak convexity from timelike lower Ricci bounds]
\label{C:sufficiency2}
Let $(M^{n},g)$ be a globally hyperbolic spacetime.
Fix  $V \in C^2(M)$ bounded, $N \ge n$ and $0<q<1$ (with $V=0$ if $N=n$).  
If  $\NRc(v,v) \ge K|v|^2_g \ge 0$ holds
in every timelike direction $(v,x) \in TM$, 
then the relative entropy $\EV(\mu)$ of \eqref{V-tropy}
is weakly $(K,N,q)$-convex for the set $Q \subset \PPac(M)^2$ 
of 
measures $(\mu,\nu)$ having 
infimum \eqref{KM} finite and 
supremum 
\eqref{MK} attained by some
$\pi \in \Pi(\mu,\nu)$ with $\ell>0$ holding $\pi$-a.e.
\end{corollary}

\begin{proof}
The proof of this corollary follows from Theorem \ref{T:Boltzmann Hessian2}
exactly as Corollary \ref{C:sufficiency} follows from Theorem \ref{T:Boltzmann Hessian}(a);
we may take $C=0$ due to our hypothesized timelike lower Ricci curvature bound.

The only difference is that, for non-compactly supported measures, we do not have the a priori lower 
bound \eqref{U Jensen} on $e(s):=\EV(\mu_s)$,  where $(\mu_s)_{s \in [0,1]} \subset \PPac(M)$ 
is the $q$-geodesic with endpoints $(\mu,\nu) \in Q$ provided by Theorem \ref{T:Boltzmann Hessian2}.
However,  as long as $\max\{e(0),e(1)\}<\infty$, the convexity of $e(s)$ established in that theorem ensures 
$e(s)$ is real-valued unless $\sup_{0<t<1} e(t) = -\infty$. 
If $\max\{e(0),e(1)\}=+\infty$, we can apply the foregoing argument on any subinterval 
$[t_0,t_1] \subset [0,1]$ satisfying $\max\{e(t_0),e(t_1)\} <\infty$ to conclude that 
$e(s)$ is real-valued, convex and satisfies the desired estimates on $[t_0,t_1]$ 
unless $(t_0,t_1) \subset e^{-1}(-\infty)$.  Either way,  we obtain the weak $(K,N,q)$ convexity
 from Definition \ref{D:KN convex} of $\EV$ for $Q$.
\end{proof}

\section{Ricci lower bounds from entropic convexity}
\label{S:RLB from convexity}


This final section is devoted to establishing converses to the corollaries of the preceding sections,
by constructing
a $q$-geodesic which shows the sufficient conditions for $(K,N,q)$-convexity of $\EV$ they provide %
are also necessary.
The strategy is based on developing a Lorentzian analog for constructions given in the Riemannian setting
by von Renesse and Sturm  \cite{SturmvonRenesse05}, and generalized by Sturm \cite{Sturm06ab}, Lott and Villani \cite{LottVillani09}.
It culminates in Theorem \ref{T:necessity}, which produces
a $q$-geodesic along which this convexity fails in the absence of  the appropriate timelike lower Ricci
curvature bound.

\begin{lemma}[Hessian of the Lorentz distance]
\label{L:lapse concavity}
Let $(M^{n},g)$ be a globally hyperbolic Lorentzian manifold. 
Fix $0<q<1$ and a future-directed proper-time parameterized geodesic segment $t \in [0,t_0] \mapsto y(t) \in M$.
Then 
\begin{eqnarray}
\label{lapse two}
- \frac{\p^2  }{\p x^\alpha \p x^\beta}  \ell(x,y(t);q) \bigg|_{x=y(0)} 
&=& \frac{\p^2 L}{\p v^\alpha \p v^\beta}(ty'(t), y(t);q) +O(t^q) 
\\ &=& O(t^{q-2})
\label{lapse three}
\end{eqnarray}
as $t \to 0^+$,
where the derivatives 
are taken in Fermi coordinates along the geodesic segment in question and the Hessian of $L$ is positive definite.
\end{lemma}

\begin{proof}
Recall that Fermi coordinates both flatten the geodesic $y(t)$ and act as Lorentzian
normal coordinates at each point along it.  Given $0 \ne w \in T_x M$, set $x(s) = \exp_{y(0)} s w$ and let 
$\gst:[0,1] \longrightarrow M$ denote the proper-time maximizing geodesic joining 
$\gst(0)=x(s)$ to $\gst(1)=y(t)$.  Taking two derivatives of
$$
- \frac1q \ell(x(s),y(t))^q =  \int_0^1 L(\dot\gamma_{(s,t)}(\lambda); q) d\lambda 
= \frac1q \int_0^1 |\dot \gamma_{s,t}(\lambda)|^q d\lambda
$$
and using $\frac{\p^2 x^\alpha}{\p s^2}|_{s=0} = 0$ leads to
\begin{eqnarray}
\nonumber
- w^\alpha w^\beta \frac{\p^2}{\p x^\alpha \p x^\beta} \ell(x,y(t);q) 
&=& \int_0^1 [D^2 
L \frac{\p\dot \gamma}{\p s} \frac{\p \dot \gamma}{\p s} + DL \frac{\p^2 \dot \gamma}{\p s^2}]_{s=0} d\lambda
\\ &=& \int_0^1 D^2 L(x,y(\lambda);q) \frac{\p\dot \gamma}{\p s} \frac{\p \dot \gamma}{\p s} d\lambda,
\label{lapse second variation}
\end{eqnarray}
where the $DL$ integral vanishes (after integrating by parts)
by the geodesic property of $\gamma=\gamma_{(s,t)}$,  and the facts that one endpoint
$\gst(1) =y(t)$ is independent of $s$ while the other
$\gst(0) = x(s)$ is a geodesic whose second $s$ derivative vanishes in our chosen coordinates.

Since the Lorentzian geodesic $\gst$ depends smoothly on its endpoints,
$\frac{\p \gamma}{\p s}$ is a Jacobi fields along $\gamma_{(0,t)}$ with end vectors
$w$ and $0$.
Since the geodesics in question are collapsing 
 to a point where the geometry is asymptotic to Minkowski space,
these Jacobi fields are asympotically linear.  The intermediate value theorem and Jacobi
equation yield
$$
\frac{\p \dot \gamma_{(0,t)}^\alpha}{\p s}(\lambda) = w^\alpha + O(wt^2).
$$
Inserting $v=\dot \gamma_{(0,t)}(\lambda)=t y'(t)$ hence $|v|=t$ into \eqref{q-Lagrangian Hessian}
yields
$$
\int_0^1 D^2 L \frac{\p\dot \gamma}{\p s} \frac{\p \dot \gamma}{\p s} d\lambda
= t^{q-2}[(2-q)(\langle y'(t), w\rangle_g^2 + |w|_g^2](1 + O(t^2)) 
$$
where the quantity in square brackets is positive due to the the uniform convexity of $L$
proved in Lemma \ref{L:q-Lagrangian}.
Comparison with \eqref{lapse second variation} yields the claims of the present lemma. 
\end{proof}



\begin{corollary}[Local concavity of the Lorentz distance]\label{C:lapse concavity}
The hypotheses and terminology of Lemma \ref{L:lapse concavity} imply
the second Lorentzian derivative of $\ell(y(0),y(t))^q$ with respect to $y(0)$
is negative-definite for $t>0$ sufficiently small.
\end{corollary}

\begin{proof}
Apart from its sign, the left-hand side of \eqref{lapse two} 
gives the second covariant derivative in question.
For $t>0$ sufficiently small, the equated right-hand side 
becomes positive-definite by 
uniform convexity of $L$ proved in Lemma \ref{L:q-Lagrangian}.
\end{proof}

By Lemma \ref{L:lapse concavity} and Corollary \ref{C:twist},  choosing
$(\po,\xo) \in T^*M$ non-zero, time-like, past-directed and
 sufficiently small ensures $\ell^q(\cdot,\yo)$ is non-degenerate with Hessian
$D^2 \ell^q (\xo, \yo) <0$ at $\xo$, where $\yo = \exp_{\xo} DH(\po,\xo;q)$.  
The next proposition provides an $\frac{\ts \ell^q}q$-convex function $u=u_{\tilde qq}$
which is smooth on a neighbourhood $U$ of $\xo$ and satisfies $Du(\xo) = \po$
and, e.g. $D^2 u (\xo)=0$. As the remark following indicates, the proof works in 
greater generality than stated. 


\begin{lemma}[Prescribing the 2-jet of an $\frac{\ts \ell^q}{q}$-convex function at $\xo$]
\label{P:q-convex construction}
Fix $0<q<1$, a compact set $X \times Y \subset M \times M \setminus \singl$
with $(\xo,\yo)$ in its interior, and a smooth function $u$ satisfying the first- and second-order conditions
$Du(\xo)= D_x b(\xo,\yo) $ and $D^2 u(\xo) > D^2_{xx} b(\xo,\yo)$ strictly,
where $b:= \frac1q \ell^q$.
Then there is a $b$-convex function $w$ on $X$ which agrees with $u$ in some neighbourhood of $\xo$.
\end{lemma}

Proof: Two applications of the implicit function theorem show that the relation 
$D_x b(x,y)-Du(x) =0$ defines a diffeomorphic correspondence $F$ between $x$ and $y$ near $(\xo,\yo)$:
the non-degeneracy of $b$ from Corollary \ref{C:twist}(iii)
gives $y=F(x)$ locally as a graph over $x$; conversely,
$F$ is smoothly invertible since $D^2 u(\xo) - D^2_{xx} b(\xo,\yo)$ has full rank.
Use this correspondence to define $v$ near $\yo=F(\xo)$ so that $v(F(x)) = b(x,F(x)) - u(x)$.
On a small enough neighbourhood $U \times F(U)$ of $(\xo,\yo)$,  the second-order hypothesis
implies for each $y \in F(U)$ that $x \in U \mapsto u(x) + v(y) -b(x,y)$ has no critical points 
save the local minimum $x = F^{-1}(x)$ at which it vanishes.  In other words 
$u(x) +  v(y) - b(x,y)$ is non-negative on $U \times F(U)$ and vanishes on the graph of the diffeomorphism 
$F:U \longrightarrow F(U)$.  Then
$$
w(x) := \sup_{y \in F(U)} b(x,y) - v(y)
$$
defines the desired $b$-convex function and coincides with $u$ throughout $U$.
\endproof

\begin{remark}
Adopting the usual definion of $b$-convexity from e.g.~\cite{Santambrogio15},
the preceding proposition and proof extend immediately to any smooth cost function $-b(x,y)$
on a compact product $X \times Y$ of equal dimensional manifolds-with-boundary satisfying the non-degeneracy condition $\det D^2_{x^i y^j} b(\xo,\yo) \ne 0$.  
No other properties specific to the Lorentz distance have been used.
\end{remark}


\begin{theorem}[Entropic convexity implies a timelike Ricci bound]\label{T:necessity}
Let $(M^{n},g)$ be a globally hyperbolic spacetime.
Fix  $V \in C^2(M)$ and $K,N \in \R$, with $V=0$ 
if $N=n$. 
If  $\NRc(v,v) \ge K|v|^2_g$ fails at some
timelike vector $(v,x) \in TM$, 
then the relative entropy $\EV(\mu)$ of \eqref{V-tropy}
fails to be weakly $(K,N,q)$-convex for any $0<q<1$.
In fact, the $q$-geodesic $s\in [0,1] \mapsto \mu_s \in \Pac(M)$ 
along which $(\LL_q(\mu_0,\mu_1)^2K,N)$ 
convexity of $e(s) := \EV(\mu_s)$ fails may be constructed so that $e \in C^2([0,1])$, 
and $\spt [\mu_0 \times \mu_1]$ is disjoint from $\npl$ but contained in an arbitrarily small neighbourhood of $(x,x)$.
\end{theorem} 

\begin{proof}
Suppose $\NRc(\hat v,\hat v)<K \in \R$ at some future-directed vector 
$(\hat v,\xo) \in TM$ with $|\hat v|_g=1$.  The idea of the proof is to construct a $q$-geodesic 
starting from measure $\mu_0$ 
which is concentrated (say uniformly) within a (Riemannian) ball of radius $r$ around $\xo$,  
and to transport it in the direction $\hat v$ for proper-time $t$,  where $r\ll t\ll 1$ are chosen
sufficiently small that the Ricci curvature remains approximately constant along the geodesic,
to facilitate computation and to contradict the $(K,N,q)$ convexity of $\EV(\mu)$.
The transport will be generated by a
smooth potential $u=u_{\td q q}$ whose first two derivatives at $\xo$ may be freely prescribed within
limits imposed by Proposition \ref{P:q-convex construction}.   Once $Du(\xo)$ has been selected to
transport $\xo$ to $y_t:=\exp_{\xo} t\hat v$, we'll choose $D^2 u(\xo)$ to make the product 
$D^2 H D^2 u$ from \eqref{Hu} become a suitable multiple of the identity operator on $T_{\xo}M$,
thus achieving the case of equality in certain inequalities in the proof.
We assume $N\ne n$, but the proof adapts easily to  $(N,V)=(n,0)$ by choosing $D^2 u(\bar x)=0$ in this case,
which is consistent with the sign definiteness of \eqref{lapse two}.

The construction,
which is localized at $\xo$,  will be carried out in Fermi coordinates
around the geodesic $y_t=\exp_{\xo} t \hat v$.
Lemma \ref{L:lapse concavity} provides 
$t>0$ sufficiently small that the Hessian of $x \mapsto \ell^q(x,y_t)$ at $\xo$
 is negative-definite and satisfies
\begin{eqnarray*}
O(t^{q-1}) &=& 
\pm t \frac{(DV(\xo) \hat v)}{N-n} \frac{\p^2 L}{\p v^\alpha \p v^\beta}(\xo, t \hat v; q) 
\\ &>& \frac{\p^2 \ell}{\p x^\alpha \p x^\beta}(\xo,y_t; q) 
\\ &=& O(t^{q-2})
\end{eqnarray*}
plus the non-degeneracy condition of Corollary \ref{C:twist}.
Fix $v_t := t\hat v$ and $p_t := DL(v_t, \xo;q)$; since we are inside the cut locus 
we know $H$ is smooth at $(p_t,\xo)$ and $v_t=DH(p_t,\xo;q)$. 
Since $y_t$ lies in the future of $\xo$ but within the timelike cut locus,    
there is a compact neighbourhood $X \times Y$ of $(\xo,y_t)$ which is disjoint from $\singl$.
Proposition \ref{P:q-convex construction} provides an $\frac{\ts \ell^q}q$-convex $u=u_{\tilde q q} \in C^3$ 
with $Du(\xo)=p_t$ and $D^2u(\xo) =  - \frac1{N-n} (DV(\xo)v_t) D^2 H(p_t,\xo;q)^{-1}$ where
$D^2 H(p_t,\xo;q)^{-1} = D^2 L(v_t,\xo;q)$ from Lemma \ref{L:q-Lagrangian} has been exploited.  
Thus $F_s(x) := \exp_x sDH(Du(x),x;q)$ is $C^2$ and $y_t=F_1(\xo)$.

Take $\mu_0^{(r)}$ to be the uniform distribution (with respect to $\vg$ say) over the Riemannian ball 
$X_r:=\tilde B_r(\xo)$, so that $\mu_0^{(r)} \to \delta_{\xo}$ against continuous test functions.
For $r>0$ sufficiently small,
$X_r \times F_1(X_r) \subset X \times Y$ hence disjoint from $\singl$.  Lemma \ref{L:well-separated}
combines with Theorem \ref{T:map} and its corollary 
to show $\mu_s^{(r)} := F_{s\#} \mu_0^{(r)} \in \Pac(M)$ defines the unique $q$-geodesic on $s \in [0,1)$
connecting its endpoints.
Moreover
\begin{equation}
\lim_{r \to 0} \LL_q \big(\mu_0^{(r)},\mu_1^{(r)}\big) = |v_t|_g = t.
\end{equation}
Regarding $r>0$ as fixed for the moment,
let $\rho_s := d\mu_s^{(r)}/d\vg$ and $e(s;r):=\EV(\mu_s^{(r)})$ denote 
the relative entropy along the geodesic in question.
Since $u$ is smooth, for $s<1$ the Monge-Amp\`ere type equation 
of Corollary \ref{C:Jacobian equation}
bounds $\|\rho_s\|_\infty$ in terms of $\|\rho_0\|_\infty$.
Thus $e(0;r)$ and $e(s;r)$ are finite, and Theorem~\ref{T:Boltzmann Hessian}
yields $\ei(\cdot;r)$ 
continuous and semiconvex on $s \in [0,1)$.
Moreover, smoothness of $F_s(x)$ implies the terms $B_s(x)= DF'_s(x) DF_s(x)^{-1}$ 
which appear in \eqref{Boltzmann gradient}--\eqref{Boltzmann Hessian} depend continuously
on $(s,x)\in  [0,1)\times M$.  Thus Lebesgue's dominated convergence theorem yields
$\ei(\,\cdot\,; r) \in C^2([0,1))$ with Theorem \ref{T:Boltzmann Hessian} and Lemma~\ref{L:Jacobi}
giving its first two $s$ derivatives 
 \begin{eqnarray*}
e'(0;r) &=& 
\int_M [DV DH(Du) - H^{ij} u_{ji} ] d\mu_0^{(r)} 
\\ &\to& (1 + \frac {n}{N-n}) DV(\xo) v_t \qquad {\rm as}\ r\to 0
\\ e''(0;r) &=& 
\int_M 
[H^{ij}u_{jk}H^{kl}u_{li} + (\Rc+D^2V)(DH(Du),DH(Du))] d\mu_0^{(r)} 
\\ &\to& \frac {n}{(N-n)^2} (DV(\xo) v_t)^2 + \NRc(v_t,v_t) + \frac{1}{N-n}(DV(\xo)v_t)^2,
\end{eqnarray*}
in view of \eqref{NBER tensor}. Thus 
\begin{eqnarray*}
\lim_{r \to 0} e''(0;r)-\frac1N e'(0;r)^2 &=&  \NRc(v_t,v_t)^2 
\\ &<& K|v_t|_g^2
\\ &=& K \lim_{r\to0} \LL_q(\mu_0^{(r)},\mu_1^{(r)})^2
\end{eqnarray*}
For $r>0$ sufficiently small,  this contradicts $(K\LL_q(\mu_0^{(r)},\mu_1^{(r)})^2,N)$ convexity of $e(s;r)$ on $[0,1)$,
as desired.
\end{proof}

\appendix

\section{Monge-Mather shortening estimate}
\label{S:Monge-Mather proof}

This appendix contains the deferred proof of Theorem \ref{T:Lipschitz inverse maps},
which we restate for convenience below.
If the Lagrangian \eqref{q-Lagrangian} were smooth and uniformly convex, 
this would follow from Corollary 8.2 of the Monge \cite{Monge81}-Mather \cite{Mather91} shortening principal from~\cite{Villani09}; see also \cite{BernardBuffoni07}. 
However, 
things are made delicate by the fact that both smoothness and uniform convexity of our
Lagrangian $L(v,x;q)$ degenerate at the light cone (Lemma \ref{L:q-Lagrangian}).
Inspired by~\cite{CorderoMcCannSchmuckenschlager01} and Theorem~8.23 of \cite{Villani09}, 
we use compactness and the $q$-separation hypothesis to 
derive the desired Lipschitz continuity directly. 
For $q=1$, 
related estimates are established by Suhr~\cite{Suhr16p}.

\begin{theorem}[Lipschitz inverse maps] 
Fix $q,s \in (0,1)$.  If $(\mu_0,\mu_1) \in \PP_c(M)^2$ is $q$-separated and $X_i:=\spt \mu_i$, 
there is a continuous map
$W:Z_s(S) \subset M \longrightarrow S \subset X_0 \times X_1$ 
such that if $\mu_s$ lies on the $q$-geodesic \eqref{q-geodesic}
then $W_\#\mu_s$ maximizes $\ell^q$ in $\Pi(\mu_0,\mu_1)$.
In fact, the map $W$ is Lipschitz
continuous with respect to any fixed choice of Riemannian distance $d_{\tilde g}$ on $M$.
Here $Z_s$ is
from 
\eqref{Z_s(X,Y)} and $S$ 
from the Definition \ref{D:q-separated} of $q$-separated.
\end{theorem}

\begin{proof}
Fix $q,s \in (0,1)$ and let $(\mu_0,\mu_1) \in \PP_c(M)^2$ be $q$-separated
and $\mu_s$ satisfy \eqref{q-geodesic}.
Setting $X=\spt \mu_0$ and $Y:= \spt \mu_1$, by Theorem \ref{T:duality by fiat}
there exist potentials $u \oplus v \ge \frac1q \ell^q$ such that the compact set
$S := \{ (x,y) \in X \times Y \mid u \oplus v = \frac1q \ell^q \}$ is disjoint from $\npl$
and contains the support of one --- hence all,  in view of \eqref{KM} --- maximizers $\pi \in \Pi(\mu,\nu)$
for \eqref{MK}.  We claim the map $W:Z \longrightarrow S$ from 
Corollary \ref{C:continuous inverse maps} is Lipschitz with respect to the Riemannian distance $d=d_{\tilde g}$,  
where $Z:= Z_s(S)$ is the compact image of $S$ from Lemma \ref{L:compact support}.
Equivalently, there exists a constant $C_s<\infty$ such that whenever 
$(x^\pm,y^\pm) = W(z^\pm)$ with $z^\pm \in Z$,
$$
d(x^+,x^-) +  d(y^+,y^-) \le C_sd(z^+,z^-).
$$
We'll establish this for $s=\frac12$ without losing generality.

 For each integer $k \in \N$ set
\begin{eqnarray*}
I_k &:=& \inf_{d(z^+,z^-) \ge 1/k} \frac{d(z^+,z^-)}{d(x^+,x^-) + d(y^+,y^-)}
\\ 
&=& \frac{d(z^+_k,z^-_k)} {d(x^+_k,x^-_k) + d(y^+_k,y^-_k)},
\end{eqnarray*}
where the infimum is over pairs $(x^\pm,y^\pm)=W(z^\pm)$ with $z^\pm \in Z$.
Compactness of $Z$ implies  $I_k$ is attained, positive and non-increasing; our goal is to show that its limit $I_\infty$ is also strictly positive. If so, then $C_s= 1/ I_\infty$ is the desired Lipschitz constant.

Use compactness of $Z$ to extract convergent subsequences $z^\pm_k \to z^\pm_\infty$;
the properties of $W:Z \longrightarrow S$ established in Corollary \ref{C:continuous inverse maps}
guarantee  $(x_k^\pm,y_k^\pm) \to (x^\pm_\infty,y^\pm_\infty)$ 
and $z_\infty^\pm \in Z_{\frac12}(x_\infty^\pm, y_\infty^\pm)$ along these subsequences.
We henceforth assume $d(z_\infty^+,z_\infty^-)=0$, since otherwise we are done. 
Continuity of $W$ then implies $(x_\infty^+,y_\infty^+)= (x_\infty^-,y_\infty^-)=:(x_\infty,y_\infty)$.
Let $t \in [-\frac12,\frac12] \mapsto \sigma_k^\pm(t) = \exp^g_{z_k^\pm} t v_k^\pm$ denote the timelike geodesic joining  $x_k^\pm$ to $y_k^\pm$ passing through $z_k^\pm = \sigma_k^\pm(0)$.
This means $\sigma_k^+$ and $\sigma_k^-$ have the same subsequential limiting geodesic $\sigma_\infty$.
Since $(x_\infty,y_\infty) \in S \subset \pl$ this geodesic is timelike: $y_\infty$ lies in the 
chronological future of $x_\infty$.

Setting 
$R_k := d_{TM}((v_k^+,z_k^+),(v_k^-,z_k^-))$
yields $r_k := d(z_k^+,z_k^-) \in [\frac1k,R_k]$ and $R_k \to 0$.
Adopting Fermi coordinates along the limiting timelike geodesic $\sigma_\infty$,
and suppressing the subscripts $k$, for $k$ sufficiently large
set $(\Delta v,\Delta z) := (v_k^- - v_k^+, z_k^--z_k^+)$  and
\begin{eqnarray*}
J_k(t) &:=& \frac1{R_k} [\exp^{\tilde g}_{\sigma^+(t)}]^{-1} \sigma^-(t)
\\ &=&\frac1{R_k} (D \exp^g)_{(tv^+,z^+)} ({t\Delta v \atop \Delta z}) + O(R_k).
\end{eqnarray*}
Choosing a further subsequence (without relabelling) along which 
\begin{equation}\label{slow limit}
\lim_{k \to \infty} \frac1{R_k} (\Delta z, \Delta v) = (\Delta v_\infty, \Delta z_\infty) \in T_{(v_\infty,z_\infty)} TM
\end{equation}
converges to a vector with unit Riemannian length. 
Along this subsequence $J_\infty(t)=\ds \lim_{k \to \infty} J_k(t)$ 
converges to a Lorentzian Jacobi field along $\sigma_\infty$.  This Jacobi field is non-trivial,  since 
$J_\infty(0) = \Delta z_\infty$,  and when $\Delta z_\infty=0$ then 
$J'_\infty(0) = \Delta v_\infty$ has unit Riemannian norm.  Although the rate of convergence of 
\eqref{slow limit} can be slow, $\Delta z_\infty=0$ implies $r_k = o(R_k)$ and
\begin{equation}\label{quantified slow limit}
d_{TM}(J_k(c_k),J_\infty(c_k)) = o(c_k) + O(R_k) \quad {\rm when}\ \frac{r_k}{R_k} \ll  c_k,
\end{equation}
i.e. as $k \to \infty$ when $c_k \ne 0$ is bounded away from zero or tends to zero more slowly than $r_k/R_k$. 

Now, since 
$
\sigma^-(t) = \exp^{\tilde g}_{\sigma^+(t)} R_k J_k(t) 
$
and hence
$$
d(\sigma^+(t),\sigma^-(t)) = R_k|J_k(t)|_{\tilde g}
$$
we find
\begin{eqnarray*}
I_\infty
&=& \lim_{k \to \infty} \frac{d(z_k^+,z_k^-)}{d(x_k^+,x_k^-) + d(y_k^+,y_k^-)}
\\ &=& \frac{|J_\infty(0)|_{\tilde g}}{|J_\infty(-\frac12)|_{\tilde g} + |J_\infty(\frac12)|_{\tilde g}}.
\end{eqnarray*}
If $J_\infty(0) \ne 0$ the denominator cannot vanish since $I_\infty \le I_k <\infty$; in this case we are done.
To derive a contradiction,  we may therefore assume $J_\infty(0)=0$.   Then $J_\infty'(0) \ne 0$ and 
\begin{equation}\label{Jacobi Taylor}
J_\infty(t) = t J'(0) + O(t^3)
\end{equation}
as $t \to 0$. (In fact $o(t)$ would be enough for our purposes:
we shall never need the fact that $J_\infty$ is a Jacobi field except to guarantee its differentiability
at the origin; it is another irrelevant fact that
the denominator above cannot vanish since no non-trivial Jacobi field vanishes both at the endpoints 
and the midpoint of a proper-time maximizing geodesic segment.)


Choose any decaying sequence of times $c_k \gg  \max\{\frac{r_k}{R_k},R_k\}$.
For large $k$, fixed and tacit, $a,b \in [0,1]$ and $c>0$ sufficiently small, Riemannian geodesics
\begin{eqnarray*}
x(a) &:=& \exp^{\tilde g}_{\sigma^+(-c)} [aRJ_k(-c)] 
\\ y(b) &:=& \exp^{\tilde g}_{\sigma^+(+c)} [bRJ_k(+c)]
\end{eqnarray*}
can be defined 
so that $x(\cdot)$ interpolates between $\sigma^\pm(-c)$ while 
$y(\cdot)$ interpolates between $\sigma^\pm(c)$.  From \eqref{quantified slow limit}--\eqref{Jacobi Taylor}
these geodesics have length $O(cR)$ much
smaller than the time separation $O(c)$ between their endpoints,  hence $k$ large enough
implies $y(b)$ lies in the chronological future of $x(a)$ for all $a,b \in[0,1]$.
Recalling Theorem~\ref{T:lapse smoothness}(c), introduce the smooth function
$$
f(a,b) := \frac1q \ell^q(x(a),y(b))
$$
where both $x$ and $y$ depend implicitly on $k$.
Since $(x_k^\pm,y_k^\pm) \in S$,
the $\ell^q$-monotonicity of $S$ established in Theorem \ref{T:duality by fiat} implies
\begin{equation}\label{Spence-Mirrlees}
0 \le f(0,0) + f(1,1) - f(0,1) - f(1,0) = \int_0^1 \int_0^1 \frac{\p^2 f}{\p a \p b} da db
\end{equation}
holds for $c=\frac12$;  in fact it holds also for each $c \in [0,\frac12]$
by the same theorem applied to the support of $\ell^q$-optimal measure $(z_{\frac12-c} \times z_{\frac12 + c})_\# \pi$
from Theorem~\ref{T:q-geodesics exist}.
Always assuming $J_\infty(0)=0$,  we'll derive a contradiction to this conclusion by showing the mixed partials
of $f$ are negative for $k$ sufficiently large.
Let $\gamma:=\gab:[0,1] \longrightarrow M$ denote the 
proper-time maximizing geodesic connecting
$x(a)$ to $y(b)$.  From 
$$
f(a,b) = - \int_0^1 L(\dot\gamma_{(a,b)}(t); q) dt 
= -\frac1q \int_0^1 |\dot \gamma_{(a,b)}(t)|^q dt
$$
we compute
\begin{eqnarray}
\nonumber
- \frac{\p^2 f}{\p a\p b} &=& \int_0^1 [D^2 
L (\frac{\p\dot \gamma}{\p a}, \frac{\p \dot \gamma}{\p b}) + DL \frac{\p^2 \dot \gamma}{\p a \p b}] dt
\\ &=& \int_0^1 D^2 L (\frac{\p\dot \gamma}{\p a}, \frac{\p \dot \gamma}{\p b}) dt,
\label{lapse cross partials}
\end{eqnarray}
where the $DL$ integral vanishes (after integrating by parts)
by the geodesy of $\gamma=\gab$,  and the fact that each of its endpoints
$\gab(0) =x(a)$ and $\gab(1) = y(b)$ depend only on one of the two variables $a$ and $b$.

Recall that the Lorentzian geodesic $\gab$ depends smoothly on its endpoints,
which lie at distance $O(cR)$ from those of $\gamma_{(0,0)}$.  Observe that
$\frac{\p \gamma}{\p a}$ and $\frac{\p \gamma}{\p b}$ are both Jacobi fields along $\gab$.
Moreover $(a,b)=(0,0)$ implies $\frac{\p \gamma}{\p a}$ is the Jacobi field with end vectors
$RJ_k(-c)$ and $0$,  while $\frac{\p \gamma}{\p b}$ has end vectors $0$ and $RJ_k(c)$.
Since the geodesics in question are collapsing 
 to a point  where the geometry is asymptotic to Minkowski space,
these Jacobi fields are asympotically linear.  The intermediate value theorem and Jacobi
equation yield
\begin{eqnarray*}
\frac{\p \dot \gamma_{(a,b)}}{\p a} &=& \frac{-RJ_k(-c)}{2c\ell(x_k^+,y_k^+)} + O(cR)
\\ \frac{\p \dot \gamma_{(a,b)}}{\p b} &=& \frac{RJ_k(c)}{2c\ell(x_k^+,y_k^+)} + O(cR).
\end{eqnarray*}
From \eqref{quantified slow limit}--\eqref{Jacobi Taylor} we find
$$
\frac{\p \dot \gamma_{(a,b)}}{\p a} = \frac{RJ_\infty'(0)}{2\ell(x_\infty,y_\infty)} + o(R) = \frac{\p \dot \gamma_{(a,b)}}{\p b}.
$$
The positive-definiteness (Lemma \ref{L:q-Lagrangian}) of $D^2 L$ 
in \eqref{lapse cross partials}
at $\dot \gamma_{(a,b)}(t)  = \dot \gamma_{0,0}(t) + O(cR)$  
gives the desired contradiction
$\frac{\p^2 f}{\p a\p b}(a,b)<0$ to \eqref{Spence-Mirrlees} for all $a,b \in [0,1]$ and $k$ sufficiently large,
thus establishing the theorem.
\end{proof}

\end{document}